\documentclass[10pt,journal]{IEEEtran}

\ifCLASSOPTIONcompsoc
  \usepackage[nocompress]{cite}
\else
  \usepackage{cite}
\fi

\ifCLASSINFOpdf

\else

\fi

\usepackage{amsmath,amssymb,amsfonts}
\allowdisplaybreaks
\usepackage{algorithmic}
\usepackage{graphicx}
\usepackage{textcomp}
\usepackage{algorithm}
\usepackage{xcolor}
\usepackage{subcaption}
\usepackage{amsthm}
\usepackage{stfloats}
\usepackage{booktabs}
\usepackage{orcidlink}

\PassOptionsToPackage{hyphens}{url}\usepackage{hyperref}
\newtheorem{theorem}{Theorem}[section]
\newtheorem{lemma}{Lemma}[section]
\newtheorem{proposition}{Proposition}[section]

\def\BibTeX{{\rm B\kern-.05em{\sc i\kern-.025em b}\kern-.08em
    T\kern-.1667em\lower.7ex\hbox{E}\kern-.125emX}}
\newcommand\MYhyperrefoptions{bookmarks=true,bookmarksnumbered=true,
pdfpagemode={UseOutlines},plainpages=false,pdfpagelabels=true,
colorlinks=true,linkcolor={black},citecolor={black},urlcolor={black},
pdftitle={Optimal Pricing and Charging Strategy Design for Non-cooperative Battery Swapping Stations},
pdfsubject={Game Theory and Optimization for Electric Vehicle Battery Swapping Stations},
pdfauthor={Huanyu Yan},
pdfkeywords={Electric vehicle, Battery swapping station, Game theory, Nash equilibrium, pricing strategy, charging strategy, aggregator, demand response, peak shaving, smart grid, subgame perfect Nash equilibrium, hierarchical game, convex optimization, quadratic programming, profit maximization}}
\expandafter\hypersetup\expandafter{\MYhyperrefoptions}
\begin{document}
\title{Optimal Pricing and Charging Strategy Design for Non-cooperative Battery Swapping Stations}

\author{Huanyu~Yan~\orcidlink{0000-0001-9619-066X},
        Chenxi~Sun,
        Huanxin~Liao~\orcidlink{0000-0001-7888-5222},~\IEEEmembership{Student Member,~IEEE,}
        and~Xiaoying~Tang~\orcidlink{0000-0003-3955-1195},~\IEEEmembership{Member,~IEEE}% <-this % stops a space
\thanks{
\par This work has been published in IEEE Transactions on Mobile Computing, vol. 23, no. 12, pp. 13573-13588, Dec. 2024. DOI: 10.1109/TMC.2024.3427784. Code is available at \href{https://github.com/T-Lab-CUHKSZ/TMC24-Optimal-Pricing-and-Charging-Strategy-Design-for-Non-cooperative-Battery-Swapping-Stations}{GitHub}. \textit{(Corresponding author: Xiaoying Tang.)}
\par Huanyu Yan, Huanxin Liao, and Xiaoying Tang are with the School of Science and Engineering, The Chinese University of Hong Kong, Shenzhen, Guangdong, 518172, P.R. China, and the Shenzhen Institute of Artificial Intelligence and Robotics for Society, The Chinese University of Hong Kong, Shenzhen, Guangdong, 518172, P.R. China, and the Shenzhen Key Laboratory of Crowd Intelligence Empowered Low-Carbon Energy Network, School of Science and Engineering, The Chinese University of Hong Kong, Shenzhen, Guangdong, 518172, P.R. China (e-mail: huanyuyan@link.cuhk.edu.cn; huanxinliao@link.cuhk.edu.cn; tangxiaoying@cuhk.edu.cn).
\par Chenxi Sun is with the Shenzhen Institute of Artificial Intelligence and Robotics for Society, The Chinese University of Hong Kong, Shenzhen, Guangdong, 518172, P.R. China (e-mail: sunchenxi@cuhk.edu.cn).}}

\IEEEtitleabstractindextext{%
\begin{abstract}
Battery swapping is a rapid way to recharge electric vehicles (EVs). As more and more entities are involved in building Battery Swapping Stations (BSSs), how non-cooperative BSSs maximize their profit in a competitive market needs further investigation. In this paper, we focus on a practical scenario where competitive BSSs are coordinated by the same aggregator. To study the optimal pricing and battery charging, we formulate a hierarchical game-theoretic model, where BSSs determine the swapping price in the day-ahead market in the first stage, and then determine the optimal battery charging strategy in the real-time market in the second stage. We rigorously prove the existence and uniqueness of the Subgame Perfect Nash Equilibrium (SPNE). In particular, the uniqueness property provides theoretical support that the strategy under equilibrium is optimal in the competitive environment. Based on the unique SPNE, we propose an optimal pricing and charging strategy for each BSS to maximize profit in the competitive market. A prediction error handling method is also proposed to deal with unexpected fluctuations in swapping demand. Our simulation with a 12-BSS system based on real-life data from Xi'an, China shows that our pricing and charging strategy increases the individual BSS profit by at least 18.1\%, while the optimal charging strategy naturally achieves peak shaving for the power grid.
\end{abstract}

}

\maketitle

\IEEEdisplaynontitleabstractindextext

\IEEEpeerreviewmaketitle

\ifCLASSOPTIONcompsoc
\IEEEraisesectionheading{\section{Introduction}\label{sec:intro}}
\else
\section{Introduction}
\label{sec:intro}
\fi

Due to the low operating cost and the environmentally friendly features, electric vehicles (EVs) have gained popularity throughout the world and pose more challenges to mobile computing \cite{fan2020enabling,liu2021reciprocal}. Battery swapping is a new service technique to shorten EV refueling time. Battery swapping stations (BSSs) first charge the batteries inside the station, and then swap a depleted battery from an EV with a fully-charged battery, charging the depleted battery later. Commercial applications have demonstrated the success of this technique in dramatically reducing refueling time. For example, the second-generation NIO BSS can complete a swap in 5 minutes \cite{NIOAPP}, comparable to the refueling speed of fossil vehicles. The significant time savings offered by battery swapping have led to the rapid development of this technology around the world \cite{NIO2022Report}.

The popularity and good prospects have attracted many entities to join the construction and operation of BSSs. In Xi'an, P.R.China, for example, EV brands, shopping centers, and the power grid are currently operating some BSSs \cite{NIOAPP}. Gas stations have also been reported to plan to construct BSSs \cite{SINOPEC}. The diversity of operation entities creates a competitive market, where each entity behaves selfishly when deciding battery charging and pricing strategies to maximize its own profit.

Although operated by different entities, BSSs in the same region have a significant electric demand in total, and can be managed together via an aggregator \footnote{The aggregator can be the retail electricity providers in real-life, like \cite{YuedeanDianli} in China or \cite{4Charge} in the U.S.} to participate in the electricity market to avoid the electricity price uncertainty and reduce the charging cost \cite{lu2020fundamentals}. To better participate in the market, the aggregator needs a demand-side management tool to coordinate electric usage. Time-of-use (TOU) pricing is a commonly used economic demand side management strategy for the aggregator in real life \cite{liu2017optimal, nogales2002forecasting, MediwaththeGameTheoretic2018}, where the electric price is affected by the total aggregated load. Both the aggregator and the BSSs benefit from TOU pricing, as it incentivizes BSSs to change their battery use away from peak hours \cite{yang2012game} to reduce the electric purchase price for the aggregator, which will further reduce the charging cost of every BSS inside the aggregator.

However, the economic coordination of electric usage makes the interests of competitive BSSs highly correlated. The BSSs in the same region compete with each other through pricing to attract more EVs. The pricing in the day-ahead market will also affect the swapping demand for each BSS, which will further bring the constraint for charging strategy in the real-time market. Besides, BSSs compete with each other during charging under TOU pricing, as they all want to charge more when other BSSs charge less. The highly correlated interests make analyzing the pricing and charging strategy challenging. This motivates us to investigate the optimal pricing and charging problem among competitive BSSs under a unified aggregator. Nevertheless, analyzing the optimality in the competitive market is difficult. It not only requires solving the equilibrium points where no BSS increases his profit by changing the strategy, but also requires comparing the profit under each equilibrium to verify which one is the maximum. In this paper, we rigorously prove the uniqueness of the equilibrium, thereby establishing a theoretical guarantee for optimality in the competitive market.

Our contribution is summarized as follows.
\begin{enumerate}
    \item \textbf{Novel Hierarchical Game Model:} We propose a hierarchical game model to study the optimal pricing and battery charging strategy. In the hierarchical game model, BSSs simultaneously set the swapping price in the day-ahead market in the first stage, and the BSSs determine the battery charging strategies in the real-time market based on the pricing result in the second stage. To the best of our knowledge, this is the first work to analyze the BSSs' correlated interests under the same aggregator. The prediction error handling method is also provided to deal with unexpected fluctuations in swapping demand.
    \item \textbf{Equilibrium Properties and Solution Approach:} We rigorously prove the existence and uniqueness of the pure strategy Subgame Perfect Nash Equilibrium (SPNE) in the hierarchical game under arbitrary BSS parameters. Specifically, the uniqueness of the Nash equilibrium provides a theoretical guarantee of the coalitional optimal, that every BSS maximizes his profit and no player is willing to deviate from the equilibrium. Based on the property, we proposed a Best Response Dynamics-based algorithm with convergence rigorously proved to solve the SPNE within any given tolerance.
    \item \textbf{BSS Optimal Pricing and Battery Charging Strategy:} We propose an optimal pricing and battery charging strategy for competitive BSSs coordinated by a unified aggregator. By implementing our proposed pricing and battery charging strategy, every competitive BSS in the system achieves coalitional optimal profit. Our simulations with real-life data demonstrate that our pricing and battery charging strategy significantly improves the BSSs' profit compared to the pricing and charging strategy adopted by swapping industry giant, the NIO company, in Xi'an, P.R. China. Furthermore, our pricing and battery charging strategy also shave the load of the aggregator.
\end{enumerate}

The remainder of the paper is structured as follows. Section~\ref{sec:related_work} introduces the related work. Section~\ref{sec:model} provides the hierarchical game model. Section~\ref{sec:property} analyzes the property of the proposed game and Section~\ref{sec:solution} presents the solution method to Subgame Perfect Nash Equilibrium. Section~\ref{sec:err_handling} provides the error handling method to deal with unexpected fluctuations in swapping demand. In Section~\ref{sec:experiment}, we conduct experiments to verify the effectiveness of our proposed game model. Finally, Section~\ref{sec:conclusion} concludes the paper.

\section{Related Work}
\label{sec:related_work}

Due to the popularity of BSSs, maximizing their profit has become a subject of interest among researchers. A line of works studies the optimal pricing \cite{liang2018battery,hu2023optimal} and optimal charging \cite{liang2018battery,sarker2014optimal,zhang2018monte,sun2017optimal, wu2017optimization} to maximize the profit of one single station. For instance, regarding pricing strategy, Hu \textit{et al.} compared the pay-per-swap and battery subscription models in \cite{hu2023optimal} for a system including one charging station, one BSS, and one battery renter. Two Stackelberg game models are adopted to solve the pricing equilibriums. As for the charging strategy, Wu \textit{et al.} formulated an optimization model in \cite{wu2017optimization} to reduce the single BSS charging cost considering the charging damage and the electricity cost while minimizing the number of batteries taken from stock to fulfill the swapping orders. While the mentioned and other researchers have significantly contributed to optimizing individual BSSs' profit, the results may not directly apply to competitive market scenarios.

Another line of work studies the profit of BSSs in the competitive market \cite{hu2023optimal,zhao2019closed,yan2022battery,ren2022game} using game theory. For instance, Zhao \textit{et al.} formulates a Stackelberg game to optimize the battery charging and swapping system in \cite{zhao2019closed} with guaranteed service quality. They studied the BSS centralized charging scenario where the swapping stations do not have the charging ability and have to send the batteries to the charging station to charge. In contrast, our work focuses on the decentralized battery charging scenario where the battery swapping stations could charge the batteries inside the station, which has been widely adopted by NIO in real life \cite{NIOAPP}. Our previous work \cite{yan2022battery} studied the BSS real-time optimal charging problem under the TOU pricing in one single time slot. In this paper, we further investigate the overall profit over multiple periods. The consideration of demand and charging strategy of multiple time slots can make BSSs charge more at low price periods to achieve higher profit, but the massive uncertainties and decision variables also bring more analysis challenges. Additionally, unlike \cite{zhao2019closed,Wu2012V2G,qian2021multi,yoon2015stackelberg,xiong2017optimal,zhao2017real} which prove the existence of the equilibrium only, we rigorously prove both the existence and the uniqueness of the equilibrium. The uniqueness property provides a theoretical assurance of the optimality of the proposed strategy under the competitive market, for it relieves BSSs of the need to consider under which equilibrium their profits are maximized.

\section{model formulation}
\label{sec:model}

\begin{figure}[t]
\centerline{\includegraphics[width=\linewidth]{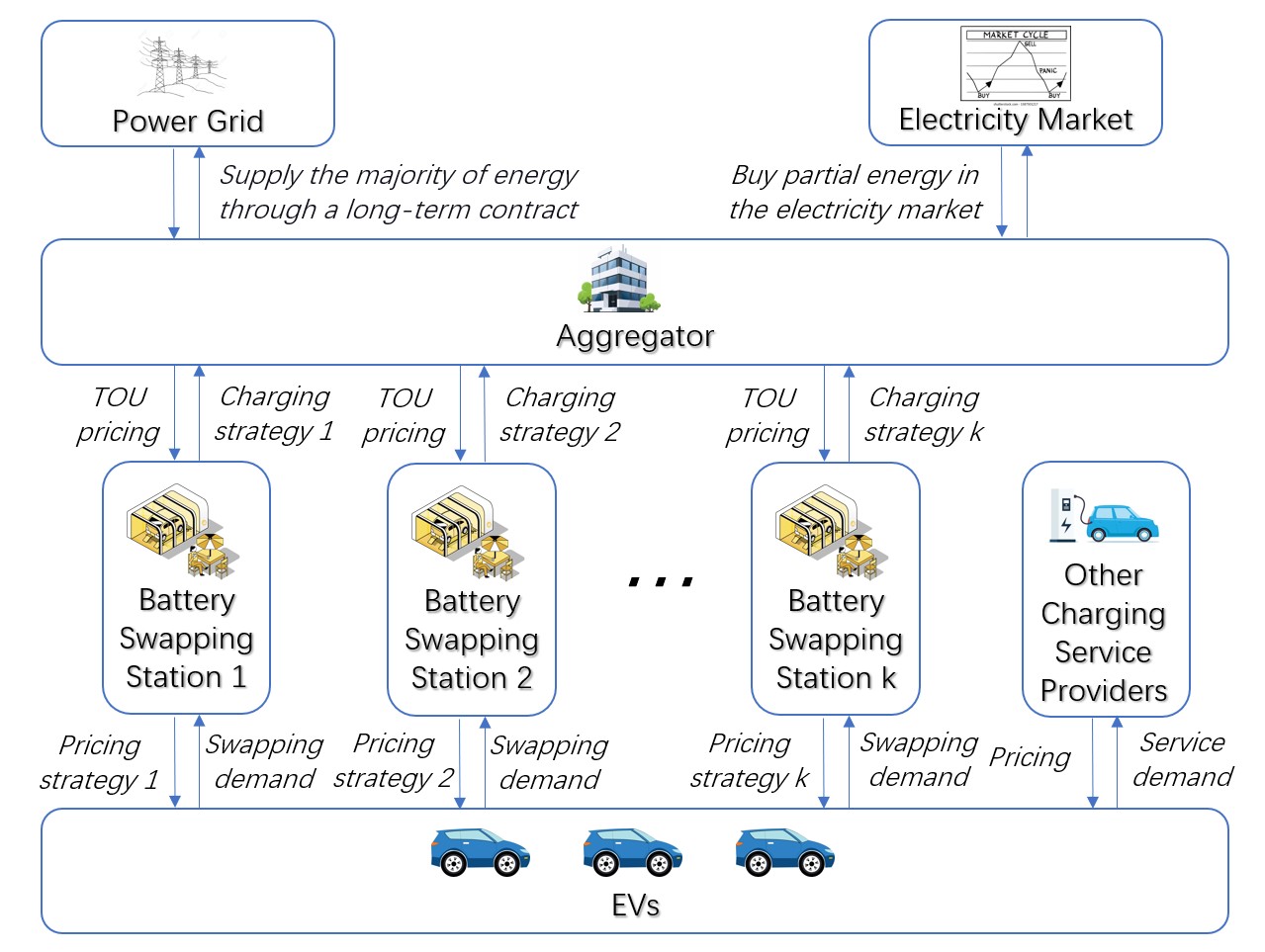}}
\caption{An aggregator-BSS-EV system. The aggregator provides energy and coordinates the charging strategy of BSSs with TOU pricing, and the BSSs compete for EV customers through pricing.}
\label{fig:system}
\end{figure}

We consider a set of competitive BSSs $\mathcal{K}$ in a region of a village or a small city, as depicted in Figure~\ref{fig:system}. The system comprises six major components: the power grid, the electricity market, the aggregator, BSSs, other charging service providers, and EVs. The aggregator coordinates all BSS charging behaviors and participates in the electricity market. Specifically, the aggregator purchases the majority of its energy through a long-term contract with the power grid to reduce the uncertainty \cite{zhang2018long}, and supplements it with partial energy purchases from the electricity market to achieve higher profit \cite{lu2020fundamentals}. In addition, the aggregator adopts TOU pricing as a well-known demand-side management tool to coordinate the charging behavior of all BSSs \cite{liu2017optimal, nogales2002forecasting, MediwaththeGameTheoretic2018}. BSSs purchase electricity from the aggregator to charge their batteries and offer swapping services to EVs. We assume all BSSs are operated by different entities, and complete with each other. Each BSS sets a swapping service price based on its own station features, and the EVs give a demand response to each station. To avoid all BSSs setting a high swapping price simultaneously, we introduce other charging service providers as a competitor. In the system, we adopt the same assumption in \cite{sun2017optimal,ni2020inventory} that all BSSs use batteries of the same type to ensure swapping compatibility.

In this paper, we consider the constant swapping price scenario, i.e., the swapping price is determined and announced to the public before the operation day, and cannot change during the operation day. The constant swapping price is more EV owners-friendly, for they do not have to consider the service fees when scheduling their swapping. Also, it is currently adopted by industry giant NIO \cite{NIOAPP}. The BSSs determine the swapping price on the operation day and announce it to the public, such as user APPs. Therefore, the pricing and charging involve both day-ahead and real-time markets. First, the BSSs determine the swapping price simultaneously in the day-ahead market, and announce the pricing result on websites or the APPs. Then, the EVs select the BSSs based on the swapping price, battery number, and station distance on the operation day. With the swapping demand estimated by pricing, each BSS determines the real-time charging amount to save the charging cost. To accommodate the dynamic nature of TOU pricing, we divide a day into several time slots $t\in[1,2,\dots,T]$, where the electricity price in each time slot is a function of the total aggregated load during the time slot. To capture this sequence, we model the operation among EVs, BSSs, and the aggregator as a hierarchical game. Game details are described below.

\subsection{EV Demand Response Model for BSS Pricing}

In this section, we propose a demand response model for BSSs based on their features. When selecting the BSSs, EVs prefer the station with more battery, low swapping price and short distance. We define these features as the station's \textit{attractions} {\cite{bell1975market}}, a positive value EV owners use to evaluate the station. The higher the attractions, the more EV owners will choose the station. We model the attractions of BSS $k$ as:
\begin{equation} \label{eq:BSS_attractions}
    U_k = C_bf(N_k) + C_p Ps_k + C_d d_k,
\end{equation}
where $C_b$ and $C_d$ are positive battery number and distance preference parameters respectively, $C_p$ is the (negative) pricing preference parameter. These parameters could obtained from the survey methods. Without loss of generality, we normalize these three parameters to let $C_p=-1$ for notation convenience \footnote{The model will first focus on homogeneous customers and will later discuss the heterogeneous setting in Sec~\ref{ssec:heter_EVs}}. $f(N)$ is the battery number advantage function. It maps the battery number inside BSS to the station's advantage. $Ps_k$ is the swapping price set by BSS $k$, and $d_k$ is the distance advantage based on all BSS locations and the distribution of EV owners in the region. In the demand response settings, we assume all BSSs will announce the price and other information on APPs or Website (which will be further introduced in Figure~\ref{fig:game_structure}), thus users have access to the full range of BSS information when making decisions, including swapping price, battery number and locations.

We first consider the battery number advantage function $f(N)$. Suppose an EV owner goes to a BSS with few batteries (for example, only one). There is a high probability that he will not be able to swap the battery immediately and will have to wait for a fully charged battery since other EVs might arrive at that BSS before him. Inspired by \cite{tushar2012economics,zhao2017real}, we find some properties function $f(N)$ must satisfy to model the function. The properties function $f(N)$ must satisfy is as follows.

(i) The function $f(N)$ should be increasing, as the more batteries BSS owns, the lower the possibility that EVs have to wait.
\begin{align}
    f(N_1) > f(N_2), \forall N_1 > N_2 \geq 1.
\end{align}

(ii) The increasing speed of function $f(N)$ should be decreasing, as the function value should change significantly when battery number $N$ is small, but change slightly when $N$ is large. For example, the difference of 1 battery with 5 batteries should be larger than the difference of 5 batteries with 10 batteries.
\begin{align}
    f(N_1) - f(N_2) < f(N_2) - f(N_3), \forall N_1 > N_2 > N_3 \geq 1.
\end{align}

(iii) The function $N-f(N)$ should increase. This property ensures that the BSS with more batteries does not attract an excessive number of EV owners that the station does not have sufficient batteries to serve.
\begin{align}
    N_1 - f(N_1) > N_2 - f(N_2), \forall N_1 > N_2 \geq 1.
\end{align}

(iv) The function $f(N)$ should approach a constant $C_0$ as the battery number goes to infinity. If a station has an infinite number of batteries, the attractions of BSS should not go to infinity, as the EV owner only ensures he will not wait for a fully-charged battery but has to consider other features such as swapping price and distance.
\begin{align}
    \lim_{N \to +\infty} f(N) = C_0.
\end{align}

Further, we let $f(0)$ be 0. Based on the aforementioned properties, the following function $f$ is proposed.
\begin{equation} \label{eq:batt_num_adv_function}
    f(N)= \frac{1}{\sqrt{C_b}} - \frac{1}{N+\sqrt{C_b}}.
\end{equation}
The proposed function in \eqref{eq:batt_num_adv_function} satisfies the above properties. It also possesses some mathematical features that allow us to design the game with excellent properties, such as the existence and uniqueness of the Nash Equilibrium, as will be discussed in Section~\ref{sec:property}. Then, we consider the distance advantage. The distance advantage of each BSS is the same as that of other competitive suppliers and is already profoundly studied by many works of literature \cite{hotbllino1929stability,fournier2014hotelling,brenner2005hotelling}. Since the BSSs' locations and the EV area's distribution can be obtained by all players, we assume $d_k$ is the common knowledge for all BSSs in the region.
% Try the result of the N-player hotelling game later.

The price for swapping should remain competitive to attract EVs. In this paper, we assume the price has an upper bound to ensure the attractions is positive, i.e.,
\begin{align} \label{eq:Ps_upper_constratints}
    Ps_k < C_b f(N_k) + C_d d_k, \forall k\in \mathcal{K}.
\end{align}

To avoid all BSSs simultaneously setting high swapping prices, we introduce a charging station as a backup choice for EVs. The attractions of the charging station is modeled as
\begin{align}
    U_c = -\lambda \hat{t} + C_p Pc + C_d d_c,
\end{align}
where $\hat{t}$ is the average serving time, including the queuing time and the charging time, $\lambda$ is the time cost, $Pc$ is the charging cost, and $d_c$ is the distance advantage. $C_p$ and $C_d$ are same in \eqref{eq:BSS_attractions}. In this paper, we focus mainly on the profit of BSS. Thus we treat the parameters of the charging station as the constant.

With the attractions of all stations, according to the market share theory in \cite{bell1975market,schuur2021explicit}, the market share $m_k$ of station $k$ is
\begin{equation}
    m_k = \frac{U_k}{\sum_{i\in \mathcal{K}} U_i + U_c}.
\end{equation}
We denote $R_{k,t}$ as the swapping demand of BSS $k$ at time $t$, and $\mathbf{R}=[R_1,\dots,R_T]$ is average hourly total swapping demand. Based on the total demand and the market share, we propose the demand response model as,
\begin{equation} \label{eq:demand_response}
    R_{k,t}=m_k \cdot R_t,
\end{equation}
In reality, $\mathbf{R}$ can be inferred from the number of EVs supporting battery swapping and the historical data of their swapping habits. Therefore, we treat $\mathbf{R}$ as a known parameter for all stations. Besides, we assume the market is sufficiently competitive \cite{david2000strategic} that the pricing of BSSs will not affect the $\mathbf{R}$.

\subsection{BSS Profit Model}
In this subsection, we model the BSS utility function, including battery charging fee, swapping revenue, and battery degradation cost.

\subsubsection{Swapping Revenue}
BSSs provide swapping services to EVs and generate revenue. We denote the swapping service price of BSS $k$ by $Ps_{k}$. Since we consider the constant pricing scenario, the swapping revenue could be modeled as
\begin{align} \label{pi revenue}
    \Pi_{revenue,k} = \sum_{t\in T} Ps_k \cdot R_{k,t}.
\end{align}

\subsubsection{Battery Charging Cost}

The BSSs need to charge their batteries to serve EVs. We denote BSS $k$ charging strategy as $\mathbf{X_k}=[X_{k,1},\dots,X_{k,T}]^\intercal$, where $X_{k,t}$ denotes the battery charging amount at time $t$. We further denote $Pe_t$ as the TOU pricing set by the aggregator, the battery charging cost of BSS $k$ can be modeled as
\begin{align} \label{pi charging}
    \Pi_{charging,k} = - \sum_{t\in T} Pe_t \cdot X_{k,t}.
\end{align}

Then we consider the constraints of $X_{k,t}$. We assume that all the EV swapping demands must be satisfied. Hence, the swapping demand sets a lower bound on the charging strategy. Considering energy loss during battery charging, we denote the battery charging efficiency as $\eta$, where $0<\eta<1$. This implies that if BSS adopts a charging strategy of $X_{k,t}$, only $\eta X_{k,t}$ energy will be stored in the battery. Since we assume the system's batteries are the same type, the battery charging efficiency should be the same for all BSSs. Furthermore, the energy stored in the BSS cannot exceed the total battery storage capacity. Thus, we obtain the demand constraints for $X_k$ as,
\begin{align} \label{eq:X_demand_constraints}
&R_{k,t} \cdot C \le E_{k,t} + \eta X_k \le N_k \cdot C, \\
&E_{k,t} = E_{k,0} + \sum_{\tau=0}^t (\eta X_{k,\tau} - C R_{k,\tau}), \forall k\in \mathcal{K}, t\in [0,T],
\end{align}
where $C$ denotes the battery volume, $N_k$ denotes the battery number of BSS $k$, $E_{k,t}$ denotes the total energy in batteries of BSS $k$ at time $t$ and $E_{k,0}$ denotes the initial energy storage in batteries at the beginning of the first time slot.

Moreover, the electricity BSS $k$ charge during one time slot is upper-bounded by physical transmission capacity, which we denoted by $X_{M,k}$. In this paper, we don't consider BSS injecting power back into the grid. Thus we get
\begin{align} \label{eq:X_transmission_constraints}
0 \le X_{k,t} \le X_{M,k}, \forall k\in K, t\in [0,T].
\end{align}

\subsubsection{Battery Degradation Cost}

In this paper, we consider a battery degradation cost model based on \cite{zhao2017real}, i.e., the degradation cost is linear to the sum of the battery charging and discharging rates. We consider all the batteries in the station as one huge storage. The storage charge rate in a time slot is the charging amount $X_{k,t}$, and the discharging rate is the swapping demand from EVs, i.e., $R_{k,t}$. Therefore, the battery degradation penalty is modeled as

\begin{align} \label{pi degradation}
\Pi_{degradation,k} &= -\sum_{t\in T} \mu(X_{k,t} + R_{k,t}),
\end{align}

where $\mu$ is the degradation rate. As we assume all BSSs use batteries of the same type in the system, parameter $\mu$ is the same for all BSSs.

Battery degradation gives BSSs a lower bound when setting the swapping price. The swapping price should be larger than the service cost due to the degradation, i.e.,

\begin{align} \label{eq:Ps_lower_constratints}
    Ps_k > \mu C, \quad \forall k \in \mathcal{K}.
\end{align}
With the above utility, we model the profit of BSS $k$ as:
\begin{align} \label{BSS utility}
\Pi_k = \Pi_{charging,k} + \Pi_{revenue,k} + \Pi_{degradation,k}.
\end{align}

\subsection{TOU Pricing Parameter for Aggregator}
\label{sec:model_aggregator}
In this subsection, we discuss the TOU pricing strategy for the aggregator. The aggregator purchases most of the energy from the power grid. Since the aggregator is a large electricity consumer, the aggregator is willing to sign a long-term contract with the power grid \cite{zhang2018long}. Long-term contracts are agreements between the power grid and consumers that specify the terms of electricity delivery and pricing \cite{song2002nash}. It provides a fixed price for electricity over a long period. This stability and predictability can help both energy producers and consumers better plan their budgets and reduce their exposure to market volatility, and therefore, is a preferred choice for large power consumers like BSS aggregators. We denote the contract price of the aggregator as $Pe_A$.

Apart from long-term contracts, aggregators will also buy partial power in the electricity market to reduce the cost \cite{lu2020fundamentals}. To better participate in the market, the aggregator prefers to coordinate the electric usage of all BSSs. TOU pricing is a powerful and widely used tool in real life to adjust the aggregated load \cite{yang2012game}. Both the aggregator and BSSs will all benefit from TOU pricing since it motivates consumers to shave the peak, saving the cost when purchasing power in the electricity market. That motivates us to investigate the region aggregator with TOU pricing.

We denote the TOU price the aggregator sets by $\mathbf{Pe}=[Pe_1,\dots,Pe_T]$. For the aggregator, $\mathbf{Pe}$ should not be too low. Otherwise, the aggregator may incur a loss. Moreover, $\mathbf{Pe}$ cannot be set excessively high, as this would encourage the BSS to trade directly with the power grid. We assume some large BSSs can also sign long-term contracts with the power grid and denote that price as $Pe_{BSS}$. Usually, the energy consumption of a single station is much smaller than that of the aggregator. Thus the energy price in the BSS contract should be higher than that in the aggregator contract, i.e., $Pe_A < Pe_{BSS}$.

Inspired by widely adopted linear TOU model in theoretical studies \cite{yang2012game,liu2017optimal,MediwaththeGameTheoretic2018} and real-life contexts \cite{MediwaththeGameTheoretic2018,LinearEPricePractice2,LinearEPricePractice3}, we adopt a TOU pricing strategy as follows,
\begin{align} \label{eq:Pe}
    Pe_t=\alpha D_t + \beta,
\end{align}
where $D_t$ is the aggregated demand of time slot $t$, parameters $\alpha$ and $\beta$ satisfy the following equations \eqref{eq:alpha_beta}, in which $D_{max}$, $D_{min}$ is the maximum and the minimum aggregated demand in the historical record.
\begin{subequations} \label{eq:alpha_beta}
\begin{align}
    Pe_A&=\alpha D_{min} + \beta, \\
    Pe_{BSS}&=\alpha D_{max} + \beta.
\end{align}
\end{subequations}

The TOU pricing in \eqref{eq:Pe},\eqref{eq:alpha_beta} has three benefits: 1) it ensures that the price will always be larger than $Pe_A$, which ensures the aggregator can make profits. 2) it ensures that the price will consistently be lower than $Pe_{BSS}$, which incentivizes BSSs to stay with the aggregator to reduce their charging costs, rather than signing a contract directly with the power grid. 3) The linear TOU electricity price will help the aggregator shave the peak load \cite{yang2012game,liu2017optimal}, which saves the cost when buying energy from the electricity market.

For the BSS aggregator, the aggregated load $D_t$ is the sum of all BSSs charging amounts and the base load in the aggregator, i.e.,
\begin{align} \label{eq:Dt}
    D_t = \sum_{k\in \mathcal{K}} X_{k,t} + D_{I,t},
\end{align}
where $D_{I,t}$ denotes the base load at time $t$. In this paper, we assume the base load $\mathbf{D_I}=[D_{I,1},\dots,D_{I,T}]$ is a constant. Additionally, we assume the grid will reveal appropriate $\mathbf{D_I}$ to BSSs at the start of operation day, to motivate them to shave the peak loads.

\subsection{Hierarchical Non-cooperative Game Model}

To study the optimal pricing and charging strategy among competitive BSSs, a hierarchical non-cooperative game model is introduced. The hierarchical non-cooperative game is a powerful tool to analyze the multi-stage decision-making process through selfish individuals \cite{fudenberg1991game}. The framework of the game model is shown in Figure~\ref{fig:game_structure}. First, BSSs simultaneously determine the swapping service price $Ps_k$ on the operation day in the day-ahead market. Then, the pricing result will be sent to the public, such as APPs or websites, and will no longer change. On the operation day, BSSs estimate the swapping demand at each time slot $R_{k,t}$ and determine the optimal battery charging strategy $X_{k,t}$. Besides, the prediction error handling is proposed to address the uncertainty of EV users.

\begin{figure}[t]
\centerline{\includegraphics[width=\linewidth]{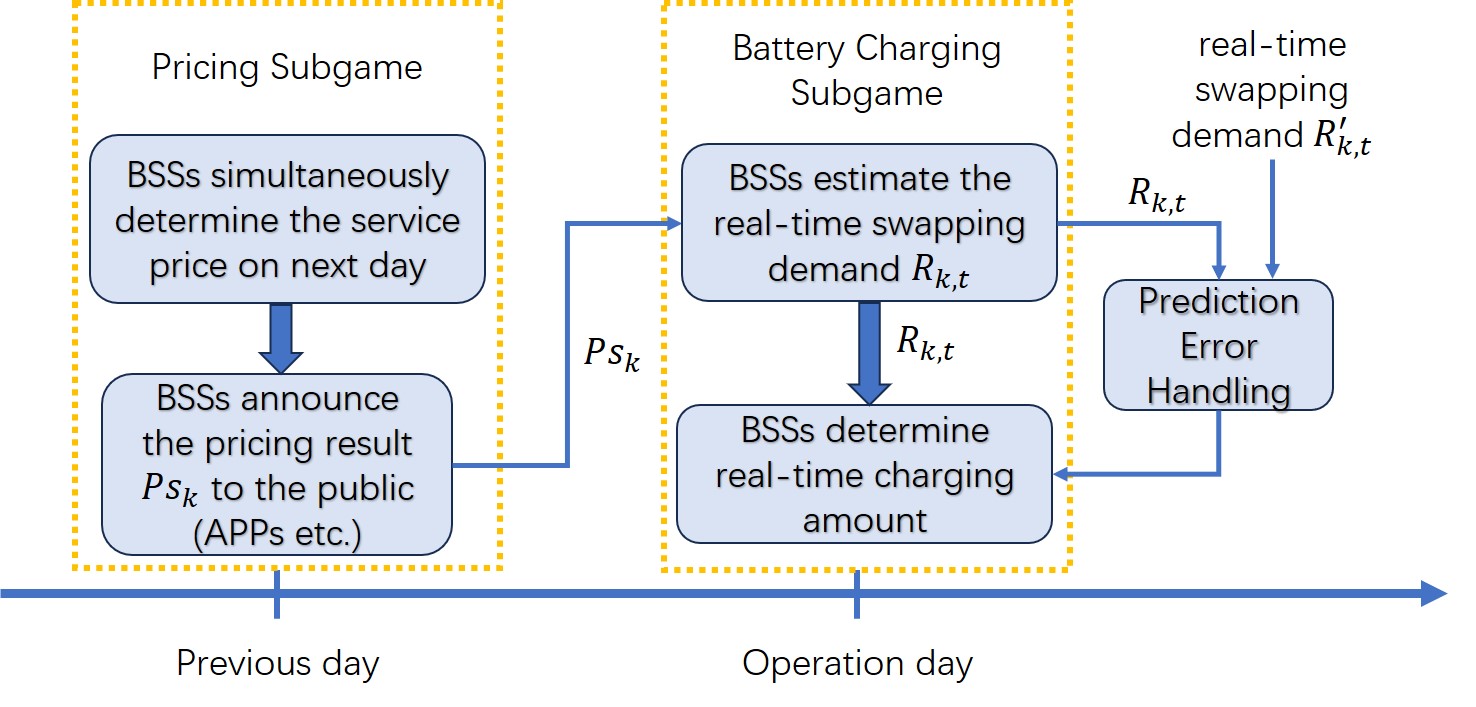}}
\caption{The hierarchical game structure. The BSS first determines the optimal pricing strategy in the Pricing Subgame before the operation day. Based on the pricing and EV swapping demand prediction, BSS determines the real-time battery charging strategy in the Battery Charging Subgame in the real-time market.}
\label{fig:game_structure}
\end{figure}

\subsubsection{Pricing Subgame}

In the pricing subgame, the BSSs determine the optimal pricing strategy to maximize the payoff. As shown in the demand response model \eqref{eq:demand_response}, the swapping demand of each BSS is determined by all station pricing strategy $Ps_k,\forall k \in \mathcal{K}$. The swapping demand will directly infect the profit. Therefore when a BSS considers the pricing strategy, he not only needs to consider his own pricing strategy, but also estimate the pricing strategies of other BSSs. We formulate the optimal pricing game to model the optimal pricing problem as follows.

\textit{\textbf{Pricing Subgame}}
\begin{itemize}
\item \textit{Players}: All BSSs in set $\mathcal{K}$.
\item \textit{Strategy}: The pricing strategy $Ps_k$ for $k \in \mathcal{K}$.
\item \textit{Payoffs}: The payoff of each BSS $k$ as defined in \eqref{BSS utility}.
\end{itemize}

\subsubsection{Battery Charging Subgame}

In the battery charging subgame, BSSs solve the optimal charging strategy. Under the TOU pricing in \eqref{eq:Pe}, the electricity price is affected by all BSS charging strategies. Thus, when determining the charging strategy, the BSS needs to predict the charging strategies of other players. The optimal charging problem is formulated as follows.

\textit{\textbf{Battery Charging Subgame}}
\begin{itemize}
\item \textit{Players}: All BSSs in set $\mathcal{K}$.
\item \textit{Strategy}: The charging amount of BSS $k$ at each time slot $\mathbf{X_k}=\{X_{k,1},X_{k,2},\dots,X_{k,T}\}$.
\item \textit{Payoffs}: The payoff of each BSS $k$ as defined in \eqref{BSS utility}.
\end{itemize}

%% complete information assumptions!
In this paper, we adopt the complete information assumption, i.e., the BSSs parameters are common knowledge for every BSSs in the system. The complete information assumption is practical for real-life scenarios since some parameters, such as the location $d_k$ and total swapping demand in the system $R_t$, are common knowledge for all BSSs. Besides, some parameters such as battery number $N_k$ will be public on APPs and websites to convince EV owners. The only private parameter, the transmission capacity $X_{M,k}$, can be estimated with the station scale. Furthermore, since $X_{M,k}$ cannot be easily changed after construction, it will eventually become known to all players during long-term operation.

\section{Existence and uniqueness of Sub-game Perfect Nash Equilibrium}
\label{sec:property}
In the game content, the most important solution concept is the Nash Equilibrium. The Nash Equilibrium provides coalitional stability that once the game fills into the Nash Equilibrium, no player can increase his profit if he unilaterally deviates from the equilibrium. For the hierarchical game, we consider the Subgame Perfect Nash Equilibrium (SPNE) \cite{fudenberg1991game}. A strategy profile is an SPNE if it represents a Nash equilibrium for every subgame of the hierarchical game. Mathematically, the SPNE for the BSS pricing and charging game is the strategy vector $[(\mathbf{X_1}^*,Ps_1^*),\dots,(\mathbf{X_K}^*,Ps_K^*)] $ with the following properties:
\begin{subequations}
\begin{align}
    &\Pi(\mathbf{X_k}^*,\mathbf{X_{-k}^*}) \geq \Pi(\mathbf{X_k},\mathbf{X_{-k}^*}), \label{eq:Xk_NE_definition} \\
    &\Pi(Ps_k^*,\mathbf{Ps_{-k}^*}) \geq \Pi(Ps_k,\mathbf{Ps_{-k}^*}), \label{eq:Ps_NE_definition} \quad \forall k \in \mathcal{K}.
\end{align}
\end{subequations}
In the non-cooperative game, the existence of the Nash Equilibrium is not always guaranteed \cite{bacsar1998dynamic}. In this section, we will investigate the existence and uniqueness properties of the SPNE in our proposed game.

\subsection{Existence of Subgame Perfect Nash Equilibrium}
\label{subsec:existence}

For notation convenience, we define matrix $\mathbf{H}=diag(\alpha,\dots,\alpha)$ with size $T \times T$, matrix  \begin{center} $\mathbf{f}_k=\begin{bmatrix}
    \mu + \alpha (\sum_{i\in \mathcal{K},i\neq k} X_{i,1} + D_{I,1}) + \beta, \\
    \mu + \alpha (\sum_{i\in \mathcal{K},i\neq k} X_{i,2} + D_{I,2}) + \beta, \\
    \vdots \\
    \mu + \alpha (\sum_{i\in \mathcal{K},i\neq k} X_{i,T} + D_{I,T}) + \beta \\
\end{bmatrix}$. \end{center}

\begin{theorem}
A pure strategy subgame perfect equilibrium exists in the BSS pricing and charging game.
\label{thm:existance}
\end{theorem}

\begin{proof}
We will prove the existence of a pure strategy Nash Equilibrium in both layers of the BSS pricing and charging game.

First, we will show that a pure strategy Nash Equilibrium exists in the charging subgame.

The payoff function \eqref{BSS utility} could be written as follows:

\begin{align} \label{eq:payoff_vector}
    \Pi_{k} &= \sum_t (Ps_k \cdot R_{k,t}- Pe_t \cdot X_{k,t} - \mu (X_{k,t}+R_{k,t})) \nonumber \\
    % &= -\sum_t ((\mu + Pe_t)X_{k,t}) + \sum_t (Ps_k-\mu C)R_{k,t} \nonumber \\
    % &= -\sum_t ((\mu + \alpha D_t + \beta)X_{k,t}) + \sum_t (Ps_k-\mu C)R_{k,t} \nonumber \\
    &= -\sum_t ((\mu + \alpha (\sum_k X_{k,t} + D_{I,t}) + \beta)X_{k,t}) \nonumber \\
    & \quad + \sum_t (Ps_k-\mu)R_{k,t} \nonumber \\
    &= -\sum_t (\alpha X_{k,t}^2+(\mu + \alpha (\sum_{i\in \mathcal{K},i\neq k} X_{i,t} + D_{I,t}) + \beta)X_{k,t}) \nonumber \\
    & \quad + \sum_t (Ps_k-\mu)R_{k,t} \nonumber \\
    &=-\mathbf{X_k}^\intercal \mathbf{H} \mathbf{X_k} - \mathbf{f_k}^\intercal \mathbf{X_k} + \sum_t (Ps_k-\mu)R_{k,t}.
\end{align}

Since $\mathbf{H}$ is positive definite, $\Pi_k$ is concave in $\mathbf{X_k}$. Therefore, the charging subgame is a concave K-player Game\cite{rosen1965existence}. According to Rosen's condition\cite{rosen1965existence}, a Nash Equilibrium in pure strategy exists in the charging subgame.

Next, we will show that the pure strategy Nash Equilibrium exists in the pricing subgame.
Consider the payoff of at time slot $t$,

\begin{equation} \label{eq:pi_kt}
    \Pi_{k,t}= Ps_k R_{k,t} - X_{k,t} Pe_t - \mu (X_{k,t} + R_{k,t}).
\end{equation}

Take the second order partial derivative on $Ps_k$, we could obtain

\begin{align}
    \frac{\partial^2 \Pi_{k,t}}{(\partial Ps_k)^2}=-2(U_k + a_k)^{-3}(Rt\cdot a)(Ps_k-\mu + U_k + a_k),
\end{align}

where $a_k=\sum_{i\in \mathcal{K},i\neq k}U_i + U_c$. Due to the constraint \eqref{eq:Ps_lower_constratints}, we have $Ps_k-\mu C > 0$. Thus the utility function $\Pi_{k,t}$ is concave in $Ps_k$. The payoff function $\Pi_k$, which is the sum of $\Pi_{k,t}$, is still concave in $Ps_k$. Therefore the pricing subgame is a concave K-player Game\cite{rosen1965existence}. According to Rosen's condition\cite{rosen1965existence}, a pure strategy Nash Equilibrium exists in the pricing subgame.

In conclusion, a pure strategy Nash Equilibrium exists in both layers of the BSS pricing and charging game. Therefore, a pure strategy sub-game perfect equilibrium.
\end{proof}

The existence of the equilibrium ensures at least one coalitional stable point. Once the game falls into equilibrium, no BSS can increase his profit by unilaterally deviating from the current strategy. Next, we will investigate the uniqueness of the equilibrium to ensure that the strategy under equilibrium is optimal.

\subsection{The Equilibrium Uniqueness in Pricing Subgame}

\begin{theorem} \label{thm:uniqueness_psk}
In the pricing subgame, the equilibrium is unique in pure strategy.
\end{theorem}

\begin{proof}
We prove the uniqueness of the Nash equilibrium by verifying that the game satisfies Theorem 2 in\cite{rosen1965existence}, i.e., the weighted sum of the utility function $\sigma(\mathbf{X},\mathbf{r})$ is diagonally strictly concave given a constant vector $\mathbf{r}=\frac{(\sum_{k\in \mathcal{K}} U_k+U_c)^2}{\sum_{t=1}^T R_t}\mathbf{1}$. In this proof, we abbreviate $\sum_{k\in \mathcal{K}} U_k$ by $\sum U$.

One sufficient condition for $\sigma(\mathbf{X},\mathbf{r})$ is diagonally strictly concave is that the symmetric matrix $[G(\mathbf{X},\mathbf{r})+G(\mathbf{X},\mathbf{r})^\intercal]$ be negative definite, where $G(\mathbf{X},\mathbf{r})$ is the Jacobian of the pseudogradient of $\sigma(\mathbf{X},\mathbf{r})$ \cite{rosen1965existence}. We define $\sigma(\mathbf{X},\mathbf{r}):=\sum r_i \Pi_i(\mathbf{X})$, and $g(\mathbf{X},\mathbf{r})$ is the pseudogradient \cite{rosen1965existence} of $\sigma(\mathbf{X},\mathbf{r})$. When $\mathbf{r}=\frac{(\sum U+U_c)^2}{\sum_{t=1}^T R_t}\mathbf{1}$, $g(\mathbf{X},\mathbf{r}) =$

\begin{align}
\begin{bmatrix}
    (u_1-p_1+\mu)(\sum U+U_c) + \sum_{i\in \mathcal{K}} (p_i-\mu)U_i \\
    (u_2-p_2+\mu)(\sum U+U_c) + \sum_{i\in \mathcal{K}} (p_i-\mu)U_i \\
    \vdots \\
    (u_k-p_k+\mu)(\sum U+U_c) + \sum_{i\in \mathcal{K}} (p_i-\mu)U_i
\end{bmatrix}.
\end{align}

Define $E_i=U_i-P_i$. The Jacobian of $g(\mathbf{X},\mathbf{r})$ is $G(\mathbf{X},\mathbf{r}) =$
\begin{align}
\begin{bmatrix}
    -2(\sum U+U_c) & E_2-E_1 & \cdots & E_k-E_1 \\
    E_1-E_2 & -2(\sum U+U_c) & \cdots & E_k-E_2 \\
    \vdots & \vdots & \ddots & \vdots \\
    E_1-E_k & E_2-E_k & \cdots & -2(\sum U+U_c)
\end{bmatrix}.
\end{align}
Thus matrix $[G(\mathbf{X},\mathbf{r})+G(\mathbf{X},\mathbf{r})^\intercal]=$
\begin{align}
\begin{bmatrix}
    -4(\sum U+U_c) & 0 & \cdots & 0 \\
    0 & -4(\sum U+U_c) & \cdots & 0 \\
    \vdots & \vdots & \ddots & \vdots \\
    0 & 0 & \cdots & -4(\sum U+U_c)
\end{bmatrix}.
\end{align}
is negative definite.

Above all, since the $[G(\mathbf{X},\mathbf{r})+G(\mathbf{X},\mathbf{r})^\intercal]$ is negative definite, the weighted sum of utility $\sigma(\mathbf{X},\mathbf{r})$ is diagonally strictly concave. According to Theorem 2 in \cite{rosen1965existence}, the equilibrium point is unique.
\end{proof}

The uniqueness of equilibrium in the pricing subgame ensures that the pricing strategy under equilibrium is coalitional optimal guaranteed. Because if there is no such result, the BSSs still need to consider under which equilibrium his profit is maximized.

\subsection{The Equilibrium Uniqueness in Battery Charging Subgame}

% Except the existence, we also find the Nash Equilibrium in the battery charging subgame is unique in pure strategy.

\begin{theorem} \label{thm:uniqueness_xk}
In the charging subgame, the equilibrium is unique in pure strategy.
\end{theorem}

\begin{proof}
We notice that the BSS profit is the sum of the profit in each time slot. We use the trick in \cite{krawczyk2005coupled} that the summation of a finite number of utility functions of the static game does not affect the concavity and strictly positive definiteness. We focus on the utility at one time slot as

\begin{align}
    \Pi_{k,t} =& - \alpha X_{x,t}^2-\alpha(\sum_{i\in \mathcal{K},i\neq k} X_{i,t}+D_{I,t} + \beta+\mu) X_{k,t} \nonumber \\
    &\qquad + (Ps_k-\mu)R_{k,t}.
\end{align}

Next, we turn to the $\Pi_{k,t}$ property. Similar to the pricing game, we verify whether the weighted sum of utility is diagonally strictly concave. When $\mathbf{r}=\mathbf{1}$, $\sigma_t(\mathbf{X},\mathbf{r})=\sum r_i \Pi_i(\mathbf{X}) =$

\begin{align}
\sum_{i=1}^k (-\alpha X_i^2 - (\mu+\alpha (\sum X_{-i} + D_I)+\beta)X_i + (Ps_k-\mu)R_i).
\end{align}

The gradient $g(\mathbf{X},\mathbf{r}) =$
\begin{align}
\begin{bmatrix}
    -2\alpha X_1^2-\mu+\alpha (\sum X_{-1} + D_I)+\beta \\
    -2\alpha X_2^2-\mu+\alpha (\sum X_{-2} + D_I)+\beta \\
    \vdots \\
    -2\alpha X_k^2-\mu+\alpha (\sum X_{-k} + D_I)+\beta
\end{bmatrix}.
\end{align}

Then the Jacobian $G(\mathbf{X},\mathbf{r}) =$
\begin{align}
\begin{bmatrix}
    -2\alpha & -\alpha & \cdots & -\alpha \\
    -\alpha & -2\alpha & \cdots & -\alpha \\
    \vdots & \vdots & \ddots & \vdots \\
    -\alpha & -\alpha & \cdots & -2\alpha
\end{bmatrix}.
\end{align}
is negative definite.

Above all, since the $[G(\mathbf{X},\mathbf{r})+G(\mathbf{X},\mathbf{r})^\intercal]$ is negative definite, the weighted sum of utility $\sigma_t(\mathbf{X},\mathbf{r})$ is diagonally strictly concave. Since summing up a finite number of utility functions of the static game does not change the concavity and strictly positive definiteness \cite{krawczyk2005coupled}, according to Rosen's condition \cite{rosen1965existence}, the equilibrium in the battery charging subgame is unique.
\end{proof}

The equilibrium uniqueness in both subgames provides a theoretical guarantee for the optimal pricing and battery charging strategy. Because if there is no such result, the BSSs still need to consider under which equilibrium his profit is maximized. Besides, the uniqueness of equilibrium guarantees that any algorithm that converges to a Nash equilibrium converges to the optimal solution, which is helpful when designing the solving method.

\section{Solution Approach for Proposed Hierarchical Game}
\label{sec:solution}

Having established the existence and uniqueness of equilibrium in the hierarchical game, we design a solution method using the backward induction \cite{fudenberg1991game}. Specifically, we first solve the equilibrium of the pricing subgame, and then use this result to solve the charging subgame.

\subsection{Best Response Dynamics Algorithm for Price Subgame}

In this subsection, we solve the equilibrium swapping price in the subgame. We introduce a concept called the best response strategy \cite{fudenberg1991game}. The best response strategy is the strategy that maximizes the payoff of one player if he knows the strategies of other players, i.e.,
\begin{align}
    Ps_k^*(\mathbf{Ps_{-k}}) = \arg \max_{Ps_k} \Pi_k(Ps_k,\mathbf{Ps_{-k}}).
\end{align}

Solving the best response usually serves as the first step in solving the equilibrium. We solve the best response pricing strategy as follows.

\begin{lemma}\label{lmm:BS_Ps}
The best response pricing strategy $Ps_k^*(\mathbf{Ps_{-k}})$ for each BSS is as follows,

\begin{align} \label{eq:BS_Ps}
Ps_k^*= &C_b f(N_k) + C_d d_k + a_k - \nonumber \\
&\quad \sqrt{a_k(a_k+ C_b f(N_k) + C_d d_k - \mu)},
\end{align}
where $a_k=\sum_{i\in \mathcal{K},i\neq k}U_i + U_c$.
\end{lemma}

\begin{proof}
We calculate the best response pricing strategy by taking the first order partial derivative on $Ps_k$,

\begin{equation}
\frac{\partial \Pi_{k,t}}{\partial Ps_k}=R_t - \frac{a_k}{U_k+a_k}R_t-\frac{P_k-\mu}{(U_k+a_k)^2}CR_t.
\end{equation}

Let the first-order partial derivative equal to zero, we obtain
\begin{align}
&\frac{\partial \Pi_{k}}{\partial Ps_k}=0, \\
&\iff 1-\frac{a_k}{U_k+a_k}-\frac{P_k-\mu}{(U_k+a_k)^2}a=0, \\
&\iff (U_k+a_k)^2=a_k(a_k-\mu C+C_b f(N_k)+C_d d_k).
\end{align}
Since $a_k>0$,$U_k>0$,$C_b f(N_k)+C_d d_k > Ps_k > \mu C$, therefore
\begin{align}
&U_k+a_k = \sqrt{a_k(a_k-\mu+C_b f(N_k)+C_d d_k)}.
\end{align}
which is equivalent to \eqref{eq:BS_Ps}.

To sum up, the pricing strategy in Lemma~\ref{lmm:BS_Ps} confirms the first order partial derivative on $Ps_k$ equals zero. Since the utility function is concave in $Ps_k$ as proved in Theorem~\ref{thm:existance}, the pricing strategy in Lemma~\ref{lmm:BS_Ps} maximizes the utility function.
\end{proof}

Next, we design an algorithm to solve the Nash Equilibrium based on the Best Response Dynamics (BRD). The algorithm details are described in Algorithm~\ref{alg:IBR_Solving_Ps_NE}. Firstly, all $Ps_k$ are initialized to the zeros. Then, we calculate the best response strategy of each BSS to update their pricing strategy from the previous iteration. Once all BSSs' strategies are updated, we obtain the new iteration and calculate the infinity norm of the two iteration vectors. If the difference exceeds the tolerance, we assign the $\mathbf{Ps^{old}}$ to the updated strategy and repeat the iteration. Otherwise, each player's strategy is the best response strategy of other players within the tolerance. By definition in \eqref{eq:Ps_NE_definition}, this vector represents the Nash Equilibrium for the swapping subgame.

\begin{algorithm}[t]
    \small
    \begin{algorithmic}[1]
        \REQUIRE $d_k$, $N_k$ for all $k \in \mathcal{K}$, $C_d$, $C_b$, $\lambda$, $\hat{t}$, $\mu$, $C$, tolerance $\epsilon$.
        \STATE Initialize $\mathbf{Ps^{old}}$ to $\mathbf{0}$, $d=+\infty$.
        \WHILE{$d > \epsilon$}
        \FOR{$k$ = [1:$K$]}
        \STATE Calculate $k$-th element of new iteration $\mathbf{Ps^{new}}$ with best response of the previous iteration $\mathbf{Ps^{old}}$ using Lemma~\ref{lmm:BS_Ps}.  \label{algstep:update_pricing}
        \ENDFOR
        \STATE Calculate $d=\Vert \mathbf{Ps^{new}} - \mathbf{Ps^{old}} \Vert_\infty$.
        \STATE Assign $\mathbf{Ps^{old}}=\mathbf{Ps^{new}}$.
        \ENDWHILE
        \STATE Output $\mathbf{Ps^{new}}$ as Nash Equilibrium.
    \end{algorithmic}
    \caption{\strut BRD for the Swapping Price Equilibrium}
    \label{alg:IBR_Solving_Ps_NE}
\end{algorithm}

\begin{theorem}
    The Algorithm~\ref{alg:IBR_Solving_Ps_NE} will converge to the unique pure strategy Nash Equilibrium in the pricing subgame.
\end{theorem}

\begin{proof}
Observing the utility function of BSS $k$. In terms of pricing, we find that the payoff of each BSS can be divided into two parts: the function of the player's own pricing (term $Ps_k$ and $U_k$), and the function of the sum of the pricing of all players (term $\sum U+U_c$). Therefore, the pricing subgame is an aggregative game \cite{dindovs2006better}.

The proof idea is from \cite{dindovs2006better}. Since the utility function only contains two parts, we define the derivative of utility $\Pi_k$ as
\begin{align}
    D_k(Ps_k,\sum) = \frac{\partial \Pi_k}{\partial Ps_k}= \frac{\partial \Pi_k^1}{\partial Ps_k} (Ps_k,\sum) + \frac{\partial \Pi_k^1}{\partial \sum} (Ps_k,\sum).
\end{align}

We consider the best response function $Ps_k^*(\mathbf{Ps_{-k}})$. As the pricing subgame is an aggregative game, the best response only depends on the sum of the opponent's strategy $\sum_{-k}=-\sum Ps_k$. That is, there is a function of sum $B_k(\sum_{-k})$ such that $Ps_k^*(\mathbf{Ps_{-k}})=B_k(\sum_{-k})$.
When $Ps_k$ equals the best response strategy $Ps_k^*(\mathbf{Ps_{-k}})$, the derivative of utility equals to zero, i.e.,
\begin{align} \label{eq:derivative_sum_ps}
    D_k(Ps_k,\sum) = D_k(B_i(\sum_{-k}),B_i(\sum_{-k}) + \sum_{-k})=0.
\end{align}
Since $\Pi_k$ is concave in $Ps_k$ as shown in Theorem~\ref{thm:existance}, we obtain \\
(1) $B_k$ is continuous and single-valued in $\sum_{-k}$, since equation~\eqref{eq:derivative_sum_ps} has one unique solution for all $\sum_{-k}$, where the concave $\Pi_k$ reaches its maximum. And \\
(2) the payoff of BSS $k$ decreases as $Ps_k$ moves away from $B_k(\sum_{-k})$.

Taking the derivative of $B_k(\sum_{-k})$ on $\sum_{-k}$, we get
\begin{align}
    &\frac{B_k(\sum_{-k})}{\sum_{-k}}=\frac{B_k(\sum_{-k})}{d a_k} \cdot \frac{d a_k}{d \sum_{-k}} \nonumber \\
    &= (1-\frac{2a_k+C_bf(N_k)+C_d d_k-\mu C}{2\sqrt{a_k(a_k+C_bf(N_k)+C_d d_k-\mu)}}) (-1) \nonumber \\
    &= \frac{2a_k+C_bf(N_k)+C_d d_k-\mu C}{2\sqrt{a_k(a_k+C_bf(N_k)+C_d d_k-\mu)}}-1.
\end{align}
Since $2a_k+C_bf(N_k)+C_d d_k-\mu \geq 2\sqrt{a_k(a_k+C_bf(N_k)+C_d d_k-\mu)}$, all derivative is positive and thus have the same sign. According to \cite{dindovs2006better}, any better reply dynamics in the pricing subgame will converge to equilibrium. In the proposed algorithm~\ref{alg:IBR_Solving_Ps_NE}, step~\ref{algstep:update_pricing} improves the profit and thus is a better reply to other BSSs' actions. Therefore Algorithm~\ref{alg:IBR_Solving_Ps_NE} will converge to the Nash Equilibrium.
\end{proof}
Intuitively, the BRD either converges to equilibrium or falls into a loop (see acyclic game in \cite{fabrikant2013structure}). Since the best response functions of the pricing subgame have the same sign on $\sum_{-k}$, players will always move in the same direction toward $B_k(\sum_{-k})$ and thus will converge to the equilibrium.

\subsection{Best Response Dynamics Algorithm for Battery Charging Subgame}

With the equilibrium in the pricing subgame, we solve the charging subgame using backward induction. In this subsection, we solve the charging equilibrium based on the pricing equilibrium.

We first solve the best response charging strategy. For notation convenience, we define some matrices as follows,
\small
$\mathbf{A}=\eta \begin{bmatrix}
    1 &   &   &  \\
    1 & 1 &   &  \\
    \vdots & \vdots & \ddots  &  \\
    1 & 1 & \cdots & 1 \\
    -1 &   &   &  \\
    -1 & -1 &   &  \\
    \vdots & \vdots & \ddots  &  \\
    -1 & -1 & \cdots & -1
\end{bmatrix}$,
$\mathbf{b_k}=\begin{bmatrix}
    N_k C-E_{k,0} \\
    N_k C-E_{k,0} + R_{k,1} \\
    \vdots \\
    N_k C-E_{k,0} + \sum_{t=1}^{T-1} R_{k,t} \\
    E_{k,0}-R_{k,1} \\
    E_{k,0}-R_{k,1}-R_{k,2}\\
    \vdots \\
    E_{k,0} - \sum_t R_{k,t} \\
\end{bmatrix}$,
\normalsize
$\mathbf{lb}=\mathbf{0}^\intercal$, $\mathbf{ub_k}=X_{M,k} \mathbf{1}^\intercal$.
%$\mathbf{ub}=[X_{M,1},\dots,X_{M,k}]^\intercal$.

\begin{lemma}\label{lmm:BS_Xk}
The best response battery charging strategy $X_k^{best}(\mathbf{X_{-i}})$ for each BSS, is the optimal solution of the following quadratic programming problem,
\begin{subequations} \label{eq:X_k_optimization}
\begin{align}
&\min_{\mathbf{X_k}} \quad &\mathbf{X_k}^\intercal \mathbf{H} \mathbf{X_k} + \mathbf{f_k}^\intercal \mathbf{X_k}, \\
&s.t. &\mathbf{lb} \leq \mathbf{X_k} \leq \mathbf{ub_k}, \label{seq:X_k_optimization_bound_contraint} \\
&     &\mathbf{A} \mathbf{X_k} \leq \mathbf{b_k}. \label{seq:X_k_optimization_inequality_constraint}
\end{align}
\end{subequations}
\end{lemma}

\begin{algorithm}[t]
    \small
    \begin{algorithmic}[1]
        \REQUIRE $N_k$, $X_{M,k}$ for all $k \in \mathcal{K}$, $\alpha$, $\beta$, $D_I$, $C$, $\mu$, tolerance $\epsilon$, pricing equilibrium $P_{s,k}$.
        \STATE Initialize $\mathbf{X^{old}}$ to $\mathbf{0}$, $d=+\infty$.
        \STATE Calculate the hourly demand $R_{k,t}$ using the equilibrium pricing strategy $P_{s,k}$ and Equation~\eqref{eq:demand_response}.
        \WHILE{$d>\epsilon$}
        \FOR{$k$ = [1:$K$]}
        \STATE Calculate $k$th row of new iteration $\mathbf{X^{new}}$ with best response of the previous iteration $\mathbf{X^{old}}$ using Lemma~\ref{lmm:BS_Xk}.\label{algstep:update_charging}
        \ENDFOR
        \FOR{$k$ = [1:$K$]}
        \FOR{$t$ = [1:$T$]}
        \IF{$d > \lvert \mathbf{X^{new}}(k,t) - \mathbf{X^{old}}(k,t) \rvert$}
        \STATE Assign $d=\lvert \mathbf{X^{new}}(k,t) - \mathbf{X^{old}}(k,t) \rvert$.
        \ENDIF
        \ENDFOR
        \ENDFOR
        \STATE Assign $\mathbf{X^{old}}=\mathbf{X^{new}}$.
        \ENDWHILE
        \STATE Output $\mathbf{X}$ as Nash Equilibrium.
    \end{algorithmic}
    \caption{\strut BRD for the Charging Nash Equilibrium}
    \label{alg:IBR_Solving_X_NE}
\end{algorithm}

The proof of Lemma~\ref{lmm:BS_Xk} is presented in Appendix~\ref{apdx:proof_lmm_BS_XK}. Once the best response charging strategy, as shown in Lemma~\ref{lmm:BS_Xk}, has been obtained, we design Algorithm~\ref{alg:IBR_Solving_X_NE} based on the BRD to solve the equilibrium. Since we have proved the uniqueness of equilibrium in Theorem~\ref{thm:uniqueness_xk}, we initialize all $X$ values to $0$. The corresponding hourly demand $R_{k,t}$ is computed using the equilibrium pricing strategy $P_{s,k}$ and Equation~\eqref{eq:demand_response}. Next, we run the iteration process and calculate the best response charging strategy for each BSS based on the previous iteration. After updating the charging strategies for all BSSs, we obtain a new iteration and calculate the maximum difference of each element in the two iteration matrices. If the maximum difference exceeds the tolerance, we assign the old iteration to $\mathbf{X^{new}}$ and repeat the iteration. Otherwise, the result is that each player's strategy is the best response strategy within the tolerance. According to the definition in \eqref{eq:Xk_NE_definition}, this result is the Nash Equilibrium of the charging subgame.

\begin{theorem}
    The Algorithm~\ref{alg:IBR_Solving_X_NE} will converge to the unique strategy equilibrium in the battery charging subgame.
\end{theorem}

\begin{proof}
The idea of this proof is similar to the convergence proof of the pricing subgame. Observing the utility function of BSS $k$. The charging subgame is somehow like the aggregative game since the payoff only contains the player's own charging $\mathbf{X_k}$, and the vector of the sum of the charging amount of other BSSs, denoted by $\mathbf{\sum}=(\sum_1,\dots,\sum_T)$, where each $\sum_t=X_{1,t}+X_{2,t}+\dots+X_{K,t}$. We define the derivative of utility $\Pi_k$ on $\mathbf{X}$ as
\begin{align}
    D_k(\mathbf{X_k},\mathbf{\sum}) = \frac{\partial \Pi_k}{\partial \mathbf{X_k}}= \frac{\partial \Pi_k^1}{\partial \mathbf{X_k}} (\mathbf{X_k},\mathbf{\sum}) + \frac{\partial \Pi_k^1}{\partial \mathbf{\sum}} (\mathbf{X_k},\mathbf{\sum}).
\end{align}

We consider the best response function $\mathbf{X_k}^*(\mathbf{X_{-k}})$. As shown in Lemma~\ref{lmm:BS_Ps}, the best response depends only on the sum vector $\mathbf{\sum_{-k}}$. That is, there is a function of the sum $B_k(\mathbf{\sum_{-k}})$ such that $\mathbf{X_k}^*(\mathbf{X_{-k}})=B_k(\mathbf{\sum_{-k}})$.
When $\mathbf{X_k}$ equals the best response strategy $\mathbf{X_k}^*(\mathbf{X_{-k}})$, the derivative of utility equals to zero, i.e.,
\begin{align} \label{eq:derivative_sum_xk}
    D_k(\mathbf{X_k},\mathbf{\sum}) = D_k(B_i(\mathbf{\sum_{-k}}),B_i(\mathbf{\sum_{-k}}) + \mathbf{\sum_{-k}})=0.
\end{align}
Since $\Pi_k$ is concave in $\mathbf{X_k}$ as shown in Theorem~\ref{thm:existance}, we obtain \\
(1) $B_k$ is continuous and single-valued in $\mathbf{\sum_{-k}}$, since equation~\eqref{eq:derivative_sum_xk} has one unique solution for all $\mathbf{\sum_{-k}}$, where the concave $\Pi_k$ reaches its maximum. And \\
(2) the payoff of BSS $k$ decreases as $\mathbf{X_k}$ moves away from $B_k(\mathbf{\sum_{-k}})$.

Next, we consider the derivative of $B_k(\mathbf{\sum_{-k}})$ on $\mathbf{\sum_{-k}}$. The best response $\mathbf{X_k}^*(\mathbf{X_{-k}})$ is the optimum of optimization problem \eqref{eq:X_k_optimization}, it satisfies the KKT conditions of \eqref{eq:X_k_optimization}. We combine the constraints \eqref{seq:X_k_optimization_bound_contraint} and \eqref{seq:X_k_optimization_inequality_constraint} to $\mathbf{A_{KKT}} \mathbf{X_k} \leq \mathbf{b_{KKT}}$. The KKT conditions of \eqref{eq:X_k_optimization} is
\begin{subequations}
\begin{align}
2\mathbf{H}\mathbf{X_k}+\mathbf{f_k}+\lambda \mathbf{A_{KKT}}=0, \\
\mathbf{A_{KKT}} \mathbf{X_k} - \mathbf{b_k} \leq 0, \\
\lambda \geq 0,\\
\lambda (\mathbf{A_{KKT}} \mathbf{X_k} - \mathbf{b_k}) =0.
\end{align}
\end{subequations}
Since we are interested in term $\mathbf{\sum_{-k}}$, we focus on the gradient condition that
\begin{align}
    &2\mathbf{H}\mathbf{X_k} + \mathbf{f_k} + \lambda \mathbf{A_{KKT}}= \nonumber \\
    &2\mathbf{H}\mathbf{X_k} + (\mu+\beta) \mathbf{1} + \alpha \mathbf{\sum_{-k}} + \alpha \mathbf{D_I} + \lambda \mathbf{A_{KKT}}=0.
\end{align}

The derivatives on $\mathbf{\sum_{-k}}$ stay a constant $\alpha$ and thus have the same sign. According to \cite{dindovs2006better}, any better reply dynamics in the pricing subgame will converge to equilibrium. In the proposed algorithm~\ref{alg:IBR_Solving_X_NE}, step~\ref{algstep:update_charging} improves the profit and thus is a better reply to other BSSs' actions. Therefore, Algorithm~\ref{alg:IBR_Solving_X_NE} will converge to the Nash Equilibrium.
\end{proof}

Intuitively, the battery charging subgame is somehow like the aggregator game. Thus, we use a similar idea to prove the convergence. We write the profit in term $B_k(\mathbf{\sum_{-k}})$, and prove that the best response function of the pricing subgame has the same sign on $\mathbf{\sum_{-k}}$. Therefore, players will always move in the same direction toward $\mathbf{\sum_{-k}}$ during BDR and will converge to the equilibrium.

\section{Prediction error handling}
\label{sec:err_handling}
% including charging cost and the degradation

In our hierarchical game model, the swapping demand is predicted in the day-ahead market and will impact the real-time battery charging market result. Although the prediction is based on historical data and a demand response model, sometimes the demand may have an unexpected fluctuation. In this section, we propose a prediction error-handling method to find the optimal battery charging strategy when the demand fluctuates unexpectedly.

Suppose BSS $k$ finds that his demand at time $t_0$ differs from the prediction. Denote the energy left in the station $k$ by $E_{k,t_0}^{\prime}$, rather than the prediction $E_{k,t_0}$. The BSS needs to find a replacement charging strategy $\mathbf{X_k}^{\prime}$ to minimize the charging and degradation cost while ensuring the swapping demand of the incoming time slots. The objective function is
\begin{subequations} \label{eq:err_handling}
\begin{align}
    &min_{\mathbf{X_k}^{\prime}} \quad &\sum_{t=t_0}^T X_{k,t}^{\prime} Pe_t + \mu(X_{k,t} + R_{k,t}). \label{seq:err_handling_tou} \\
    &s.t.  &Pe_t = \alpha (X_{k,t}^{\prime} + \sum X_{-k,t}^{NE} + D_{I,t}) + \beta, \\
    &      &E_{k,t} = E_{k,t_0}^{\prime} + \eta X_{k,t} - C R_{k,t}, \\
    &      &C\cdot N_k \ge E_{k,t} \ge C \cdot R_{k,t} \forall t\in T,  \label{seq:err_handling_total_amont} \\
    &      &0 \le X_{k,t} \le X_{M,k}, \forall t\in T \label{seq:err_handling_transmission}.
\end{align}
\end{subequations}
The constraint \eqref{seq:err_handling_tou} is the TOU pricing. The term $\sum X_{-k,t}^{NE}$ denotes the equilibrium charging strategy of other BSSs except for BSS $k$, which could be obtained by Algorithm~\ref{alg:IBR_Solving_X_NE}. The inequality \eqref{seq:err_handling_total_amont} ensures the energy BSS left in the station could meet the swapping demand while not exceeding the storage capacity. The inequality \eqref{seq:err_handling_transmission} is the transmission constraint. For convenience of notation, we define matrix $\mathbf{H}^{\prime}=diag(\alpha,\dots,\alpha)$ with size $T-t_0 \times T-t_0$, matrix \begin{center}
$\mathbf{f_k}^{\prime}=\begin{bmatrix}
    \mu + \alpha (\sum_{i\in \mathcal{K},i\neq k} X_{i,t_0} + D_{I,t_0}) + \beta \\
    \mu + \alpha (\sum_{i\in \mathcal{K},i\neq k} X_{i,t_0+1} + D_{I,t_0+1}) + \beta \\
    \vdots \\
    \mu + \alpha (\sum_{i\in \mathcal{K},i\neq k} X_{i,T} + D_{I,T}) + \beta \\
\end{bmatrix}$,
$\mathbf{A}^{\prime}= \eta \begin{bmatrix}
    1 &   &   &  \\
    1 & 1 &   &  \\
    \vdots & \vdots & \ddots  &  \\
    1 & 1 & \cdots & 1 \\
    -1 &   &   &  \\
    -1 & -1 &   &  \\
    \vdots & \vdots & \ddots  &  \\
    -1 & -1 & \cdots & -1
\end{bmatrix}$,\\
$\mathbf{b_k}^{\prime}=\begin{bmatrix}
    N_k C-E_{k,t_0}^{\prime} \\
    N_k C-E_{k,t_0}^{\prime} + R_{k,t_0+1} \\
    \vdots \\
    N_k C-E_{k,t_0}^{\prime} + \sum_{t=t_0}^{T-1} R_{k,t} \\
    E_{k,t_0}^{\prime}-R_{k,t_0} \\
    E_{k,t_0}^{\prime}-R_{k,t_0}-R_{k,t_0+1} \\
    \vdots \\
    E_{k,t_0}^{\prime} - \sum_{t=t_0}^T R_{k,t} \\
\end{bmatrix}$,
\end{center}
$\mathbf{lb}^{\prime}=\mathbf{0}^\intercal$, $\mathbf{ub_k}^{\prime}=X_{M,k} \mathbf{1}^\intercal$ with size $T-t_0$. The error-handling problem \eqref{eq:err_handling} can be written in quadratic form. Thus, we obtain the following proposition.

\begin{proposition} \label{pps:err_handling}
The optimal error-handling strategy is the solution of the following Quadratic Optimization Problem,
\begin{subequations} \label{eq:error_quadratic}
\begin{align}
&\min_{\mathbf{X_k}^{\prime}} \quad &\mathbf{X_k}^{\prime \intercal} \mathbf{H}^{\prime} \mathbf{X_k}^{\prime} + \mathbf{f_k}^{\prime \intercal} \mathbf{X_k}^{\prime}, \\
&s.t. &\mathbf{lb}^{\prime} \leq \mathbf{X_k}^{\prime} \leq \mathbf{ub_k}^{\prime}, \\
&     &\mathbf{A}^{\prime} \mathbf{X_k}^{\prime} \leq \mathbf{b_k}^{\prime}.
\end{align}
\end{subequations}
\end{proposition}

In summary, the optimal charging strategy in the real-time market depends on the EV swapping demand, which is predicted by our day-ahead pricing model. Proposition~\ref{pps:err_handling} gives an optimal adjustment scheme when the swapping demand has an unexpected fluctuation in the real-time market.

\section{Simulations}
\label{sec:experiment}
In this section, we conduct simulations to demonstrate the advantages of the proposed pricing and charging strategy with real-life parameters.

\subsection{Simulation Settings}

\subsubsection{Parameter settings}

\begin{figure*} % 因排版移动
\begin{subfigure}[t]{0.5\textwidth}
   \includegraphics[width=\linewidth]{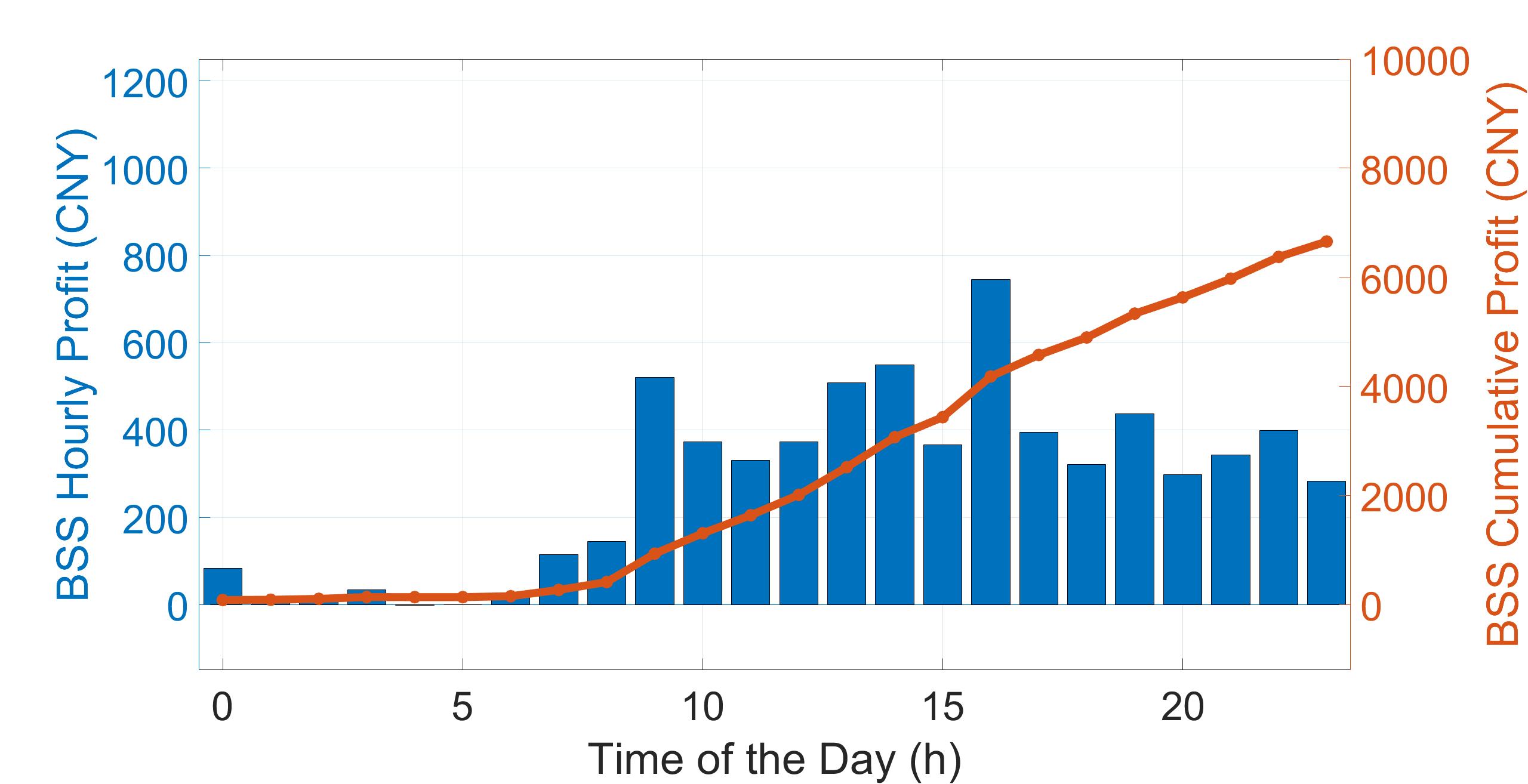}
   \caption{Uncoordinated charging, uncoordinated pricing.}
   \label{fig:BSS7_hourly_profit_curcur}
\end{subfigure}
\hfill
\begin{subfigure}[t]{0.5\textwidth}
   \includegraphics[width=\linewidth]{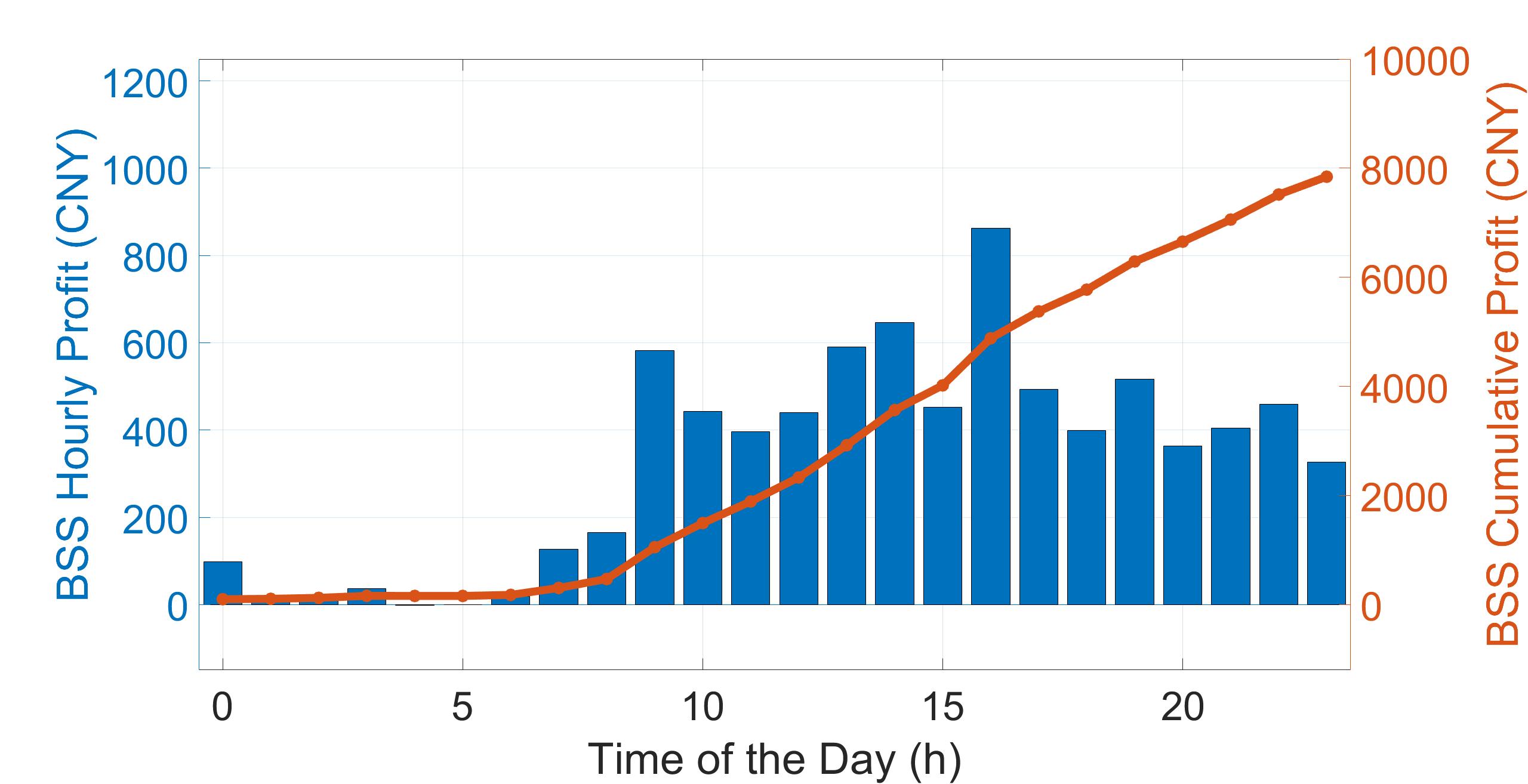}
   \caption{Uncoordinated charging, optimal pricing.}
   \label{fig:BSS7_hourly_profit_curopt}
\end{subfigure}
\hfill
\begin{subfigure}[t]{0.5\textwidth}
   \includegraphics[width=\linewidth]{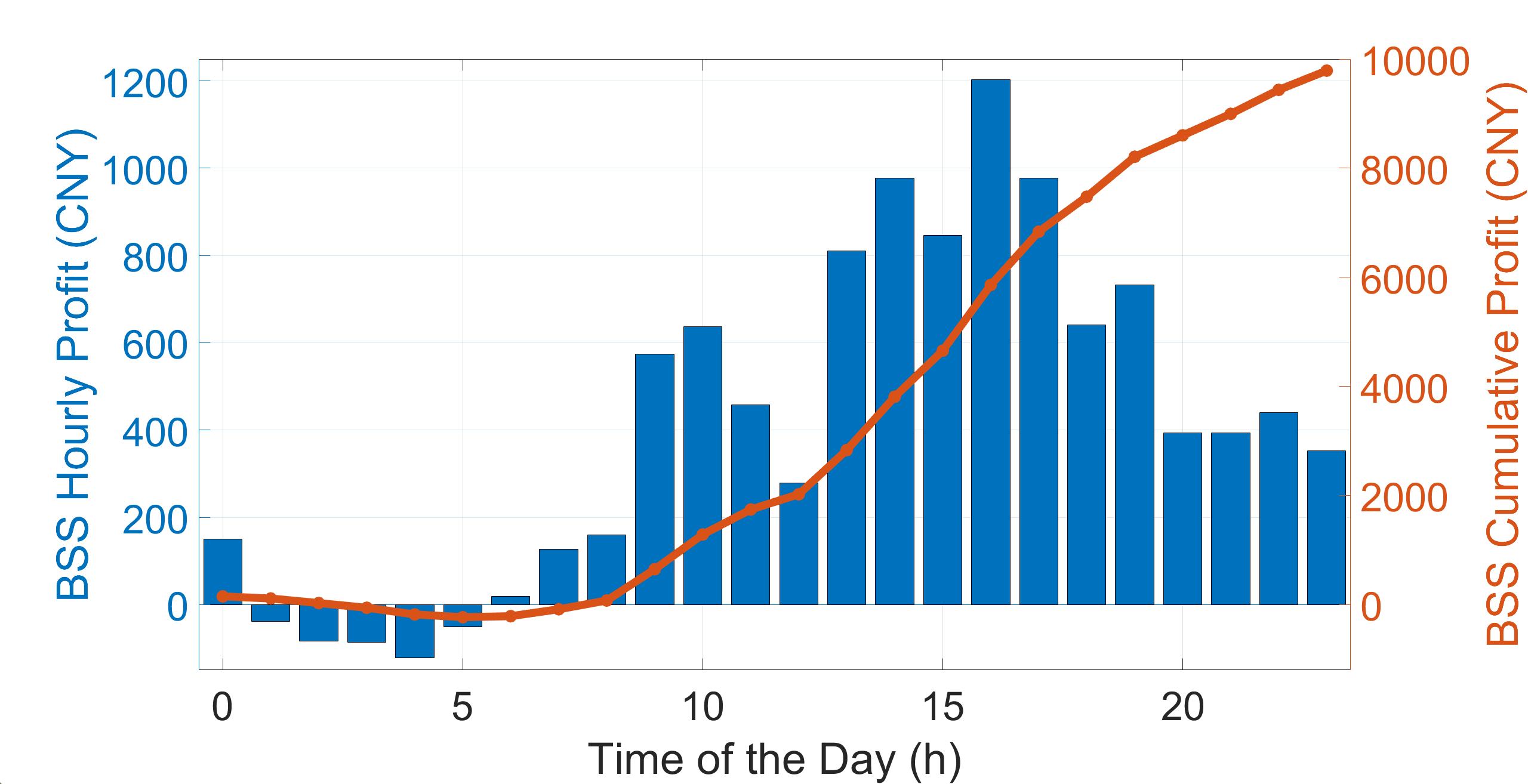}
   \caption{Optimal charging, optimal pricing.}
   \label{fig:BSS7_hourly_profit_optopt}
\end{subfigure}
\hfill
\begin{subfigure}[t]{0.5\textwidth}
   \includegraphics[width=\linewidth]{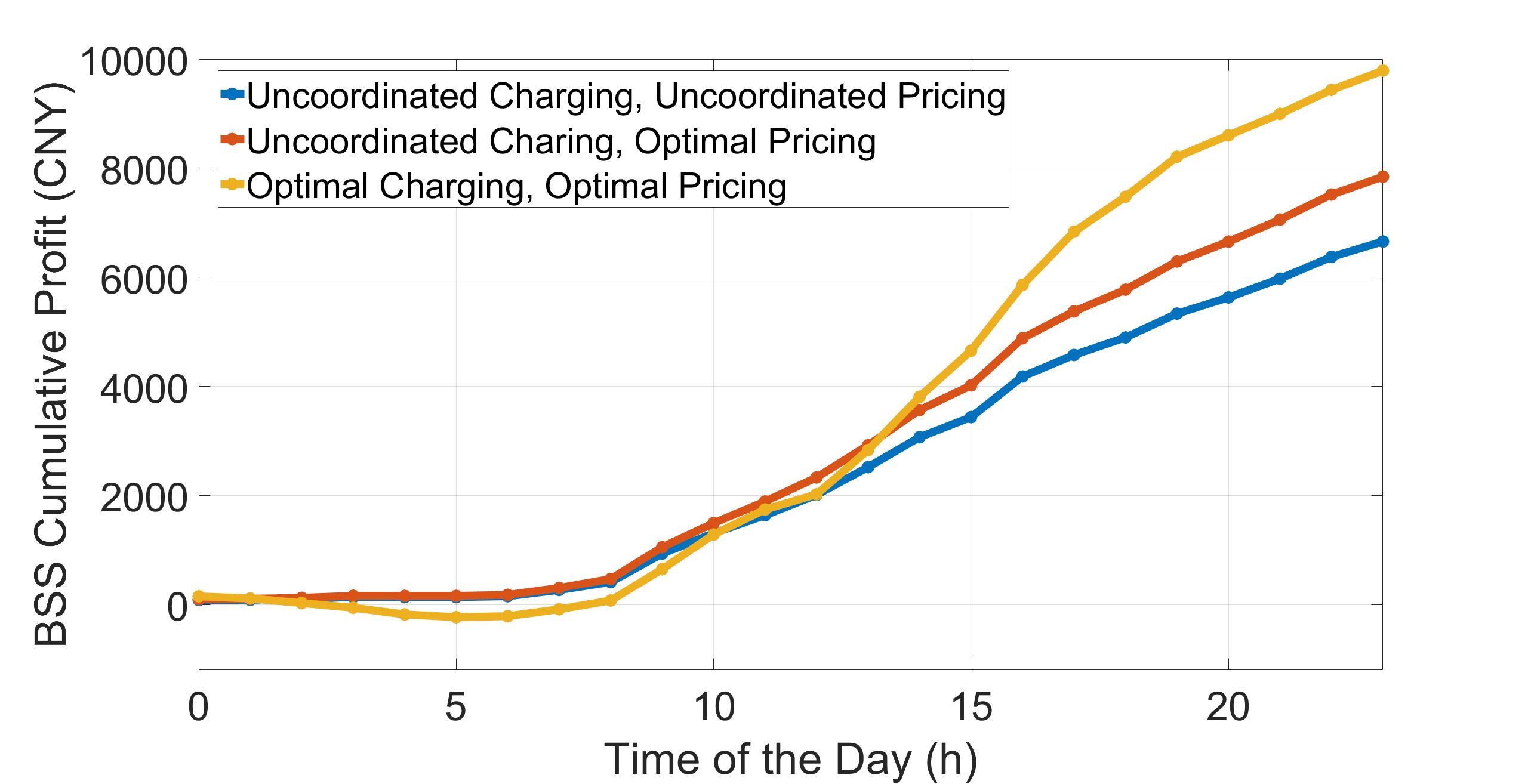}
   \caption{Cumulative profit comparison.}
   \label{fig:BSS7_profit_comparision}
\end{subfigure}
\hfill
\caption{BSS 7 profit under different pricing and charging strategies in 12-BSS system.}
\label{fig:BSS7_profit}
\end{figure*}

\begin{table}[t]
    \centering
    \caption{Parameters Settings}  % 表格标题
    \label{tab:Exp_parameters}  % 用于索引表格的标签
    \begin{tabular}{cc}
   \toprule
   Parameters & Values \\
   \midrule
		$K$ & 12 \\
		$T$ & 24 \\
        $Pe_A$ & $0.1$ CNY/kWh \\
		$Pe_{BSS}$ & $1.5$ CNY/kWh \\
        $\mathbf{D_I}$ & Scaled hourly load of Singapore \cite{SingaporeLoad} \\
  	$C$ & 75kWh \\
        $\mathbf{N}$ & [26,10,9,14,20,20,24,10,26,10,22,24] \\
        $\mathbf{X_M}$ & $\frac{1}{2}C \cdot \mathbf{N}$ \\
        $\mu$ &  0.005 CNY/kW \\
 	$\mathbf{R_t}$ & Swapping demand data in Xi'an \cite{NIOAPP} \\
        $\mathbf{d_k}$ & [9,7,9,13,15,22,11,7,9,8,14,9] \\
        $d_c$ & 24 \\
        $C_b$ & $30^2$ \\
        $C_d$ & 50 \\
        $\hat{t}$ & 2.5 hours \\
        $\lambda$ & 30 CNY/hour \\
        $Pc$ & 2.5 CNY/kWh \\
   \bottomrule
   \end{tabular}
\end{table}

The main parameters are presented in Table~\ref{tab:Exp_parameters}. We set up a 12-BSS system based on the BSS operations in Xi'an City, Shaanxi Province, China. The base load of the aggregator $D_I$ is set to the scaled average hourly load of Singapore \cite{SingaporeLoad}. Besides, we set $\mathbf{X_{M}}=\frac{1}{2}C\mathbf{N}$ to approximate the situation where all BSSs can charge all batteries to full within two hours. The hourly average total swapping demand data in the region is obtained from \cite{NIOAPP}. We set the charging price at the charging station to 2.5 CNY/kWh, and the average serving time (including charging time and queuing time) to 2.5 hours. The average time cost for EV owners is 30 CNY/hour. To solve the optimization problem \eqref{eq:X_k_optimization}, we use the \textit{quadprog} function from Matlab$^\circledR$.

\subsubsection{Benchmark settings}

We consider benchmark cases as follows:
\begin{enumerate}
    \item \textbf{Uncoordinated Pricing Strategy (Cur)}: We set the unordered pricing strategy as all BSS set 262.5 CNY for one swapping service (which is 3.5 CNY/kWh for 75 kWh batteries) as a benchmark case.
    \item \textbf{TOU pricing (TOU):} This pricing strategy is outlined in \cite{liang2018battery}. It contemplates a BSS system comprising a single BSS, a power grid, and generator set dispatch under TOU pricing. Parameters from C1S1 in \cite{liang2018battery} are utilized, with $\omega(t)$ set to 0.4 \cite{Coal-fired}.
    \item \textbf{Uncoordinated Charging Strategy (Cur):} The operator will charge the battery to full immediately as soon as he gets a new depleted battery. This battery charging strategy is currently widely used in Xi'an BSS operations\cite{NIOAPP} and is therefore chosen as the benchmark battery charging strategy.
    \item \textbf{Solely Optimization Charging (O):} This approach involves the operator optimizing their own charging strategies, disregarding competition among BSSs. The strategy of non-B2G BSSs from Figure 5b in \cite{ding2022joint} is adopted.
\end{enumerate}

\subsection{BSS Profit under Different Pricing and Charging Strategy}
In this subsection, we calculate the BSS average daily profit under different pricing and charging strategies. Due to the space limit, we provide a detailed analysis of BSS 7, for both $N_7$ and $d_7$ are near the median of $\mathbf{N}$ and $\mathbf{d_k}$. The results of the other stations are shown in Appendix~\ref{apdx-BSS_hourly_all}.

\subsubsection{BSS 7 average hourly profit under different strategies}
\label{ssec:BSS7_hourly_profit}

In this experiment, we calculate the average hourly profit of BSS 7. Specifically, we simulate a 12-BSS system based on real-life data for a week, and calculate the average hourly profit of BSS 7 under our proposed strategy and the benchmark cases. The profits results are shown in Figure~\ref{fig:BSS7_hourly_profit_curcur}, ~\ref{fig:BSS7_hourly_profit_curopt} and ~\ref{fig:BSS7_hourly_profit_optopt}. The optimal pricing and the average optimal charging strategy per day are shown in Figure~\ref{fig:opt_pricing_charging}. Besides, we draw the cumulative average daily profit at every hour and compare the results in Figure~\ref{fig:BSS7_profit_comparision}.

As depicted in Figure~\ref{fig:BSS7_profit}, the pricing and charging strategies lead to an improvement in the BSS average daily profit. For example, the average daily profit under the uncoordinated charging and uncoordinated pricing strategy is around 6656 CNY. In comparison, the profit under the uncoordinated charging and optimal pricing strategy is around 7843 CNY, and the profit under optimal charging and optimal pricing strategy is around 9788 CNY. Figure~\ref{fig:opt_pricing_charging} show the optimal pricing and average daily charging strategy for each BSS. As shown in the graph, all of the BSS services are priced higher than the charging stations, with BSS 6 being the most expensive because its greatest distance advantage.

Additionally, the hourly profit gives an intuition into how the proposed pricing and charging strategy increases the profit. As shown in Figure~\ref{fig:BSS7_hourly_profit_curcur} and ~\ref{fig:BSS7_hourly_profit_curopt}, the pricing strategy increases the bar length without affecting the profit distribution over time, indicating that the pricing strategy enhances the profit by setting a more reasonable price based on the station advantage compared to other BSSs in the system. Conversely, the charging strategy changes the profit distribution over time, as demonstrated in Figures~\ref{fig:BSS7_hourly_profit_curopt} and \ref{fig:BSS7_hourly_profit_optopt}, indicating that the proposed charging strategy encourages the BSS to charge at lower electricity price and reduces the charging cost.

\begin{figure*}
\begin{subfigure}[t]{0.3\textwidth}
   \includegraphics[width=\linewidth]{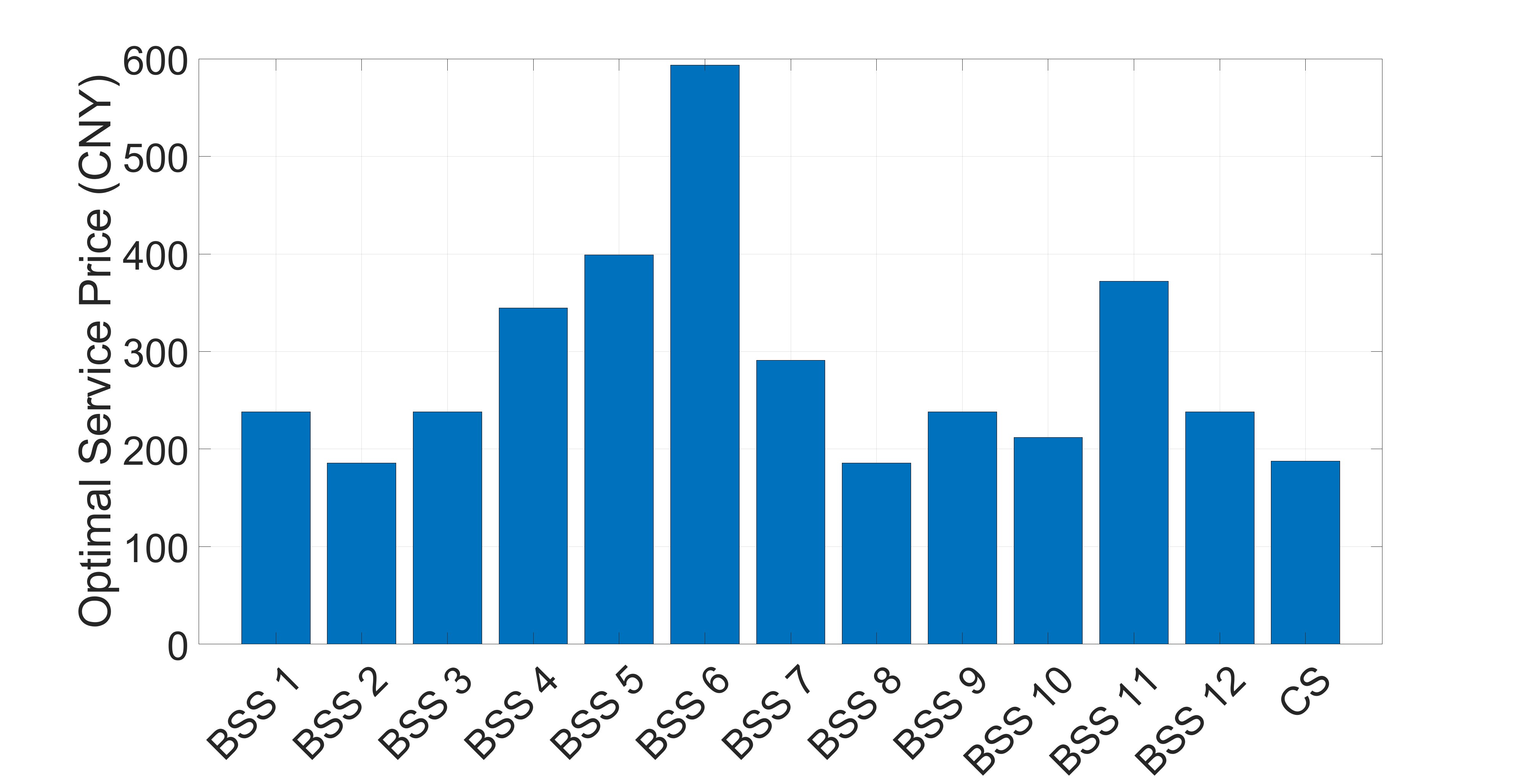}
   \caption{Optimal pricing.}
   \label{fig:opt_pricing}
\end{subfigure}
\hfill
\begin{subfigure}[t]{0.3\textwidth}
   \includegraphics[width=\linewidth]{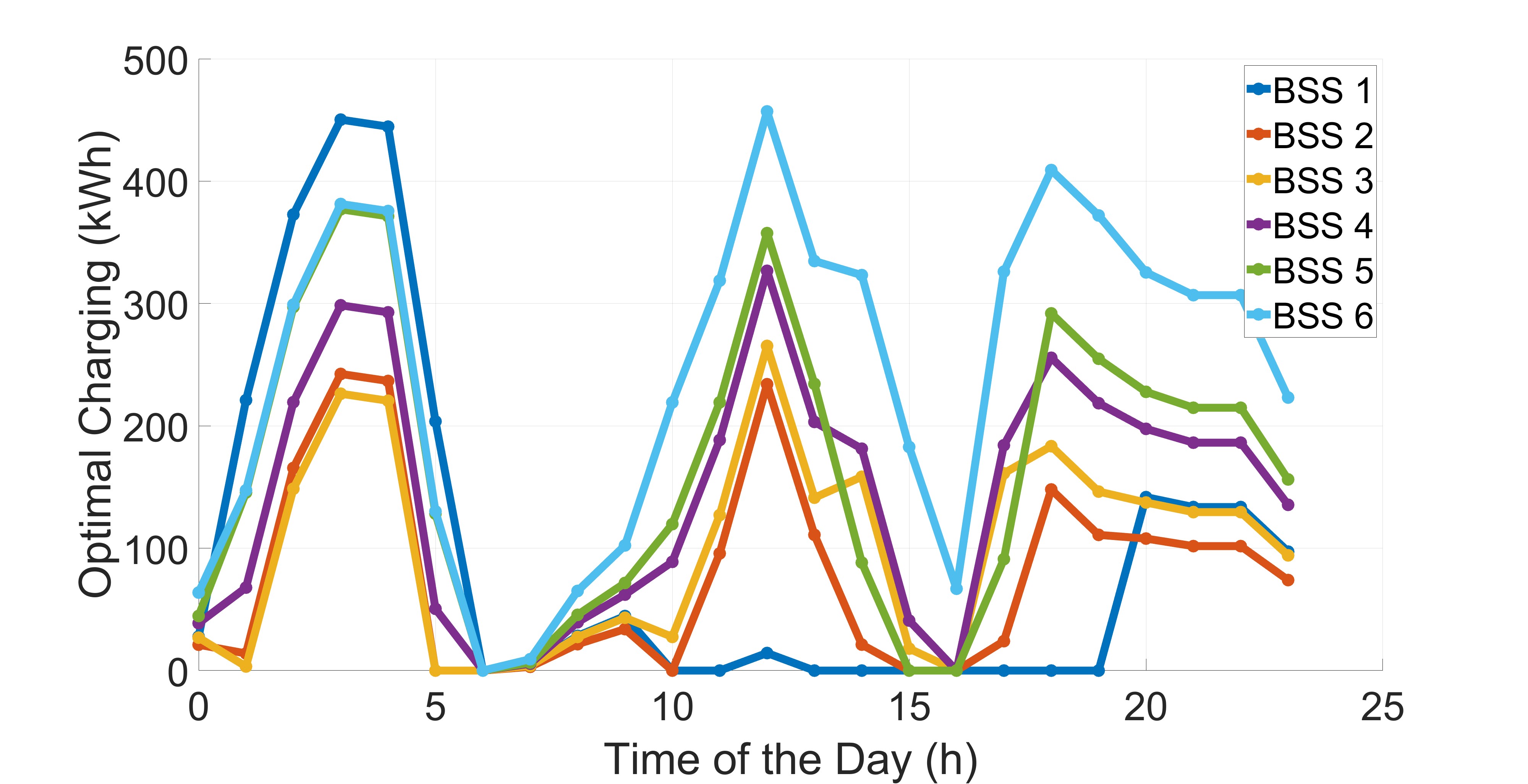}
   \caption{Optimal charging BSS 1-6.}
   \label{fig:opt_charging_1_6}
\end{subfigure}
\hfill
\begin{subfigure}[t]{0.3\textwidth}
   \includegraphics[width=\linewidth]{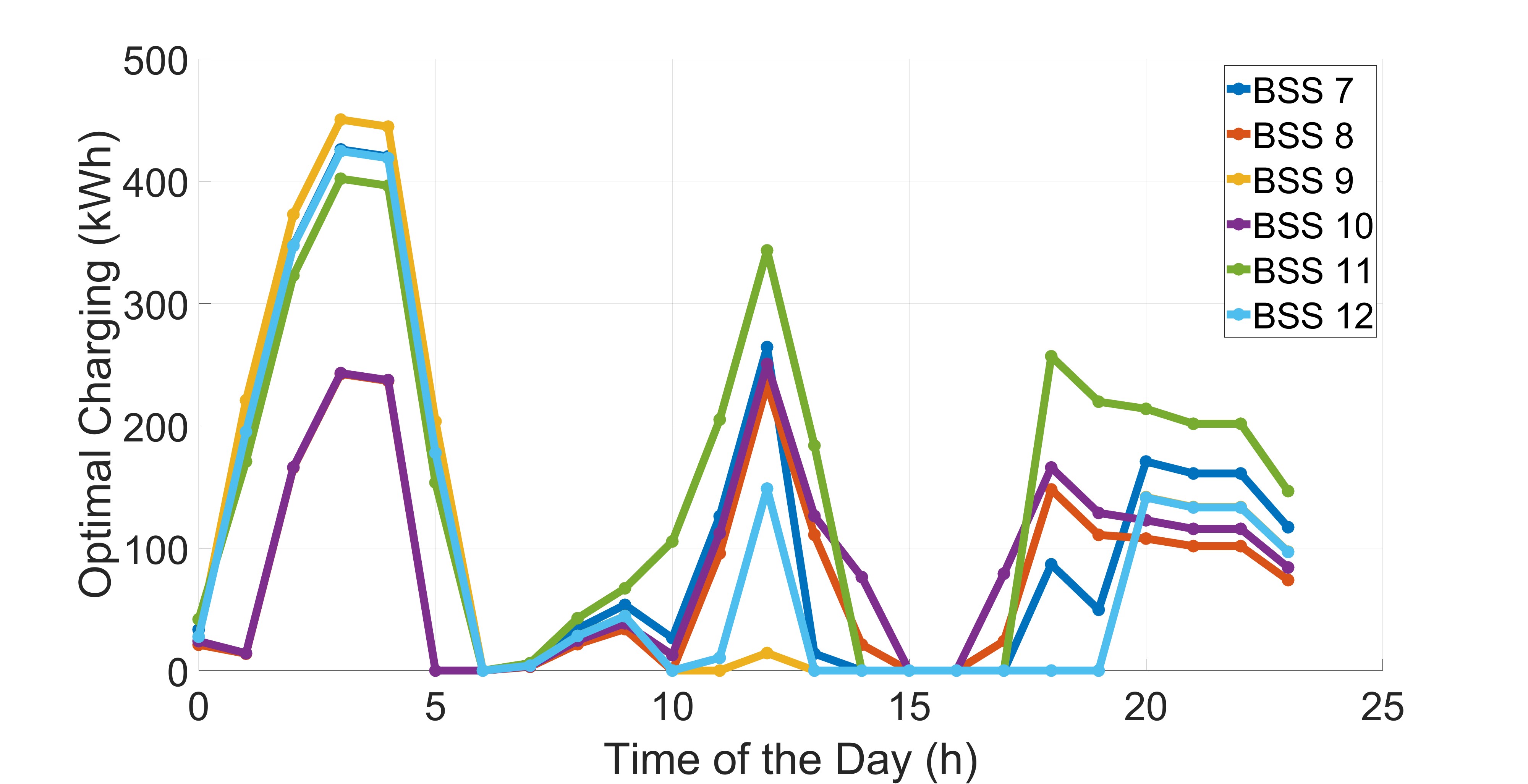}
   \caption{Optimal charging BSS 7-12.}
   \label{fig:opt_charging_7_12}
\end{subfigure}
\hfill
\caption{The optimal pricing and average daily optimal charging strategy for 12-BSS system. The shown is the average value because the optimal result is different since the remaining power at 0:00 differs for each day.}

\label{fig:opt_pricing_charging}
\end{figure*}

\subsubsection{Comparison with benchmarks}
\label{ssec:benchmark}

In this section, we evaluate the profits under various pricing and charging strategies. The system was simulated over a period of 7 days, during which we computed the profit for all BSSs and the system’s total profit (social welfare). To incorporate more randomness, we assumed that the demand for the next day is the previous day’s demand multiplied by a random number between 0.8 and 1.2. The comparative results are depicted in Figure~\ref{fig:bm}. Due to space constraints, we present the results for the BSS with median $N$ and $d_k$ (BSS 7), as well as the social welfare.

\begin{figure*}
\begin{subfigure}[t]{0.3\textwidth}
   \includegraphics[width=\linewidth]{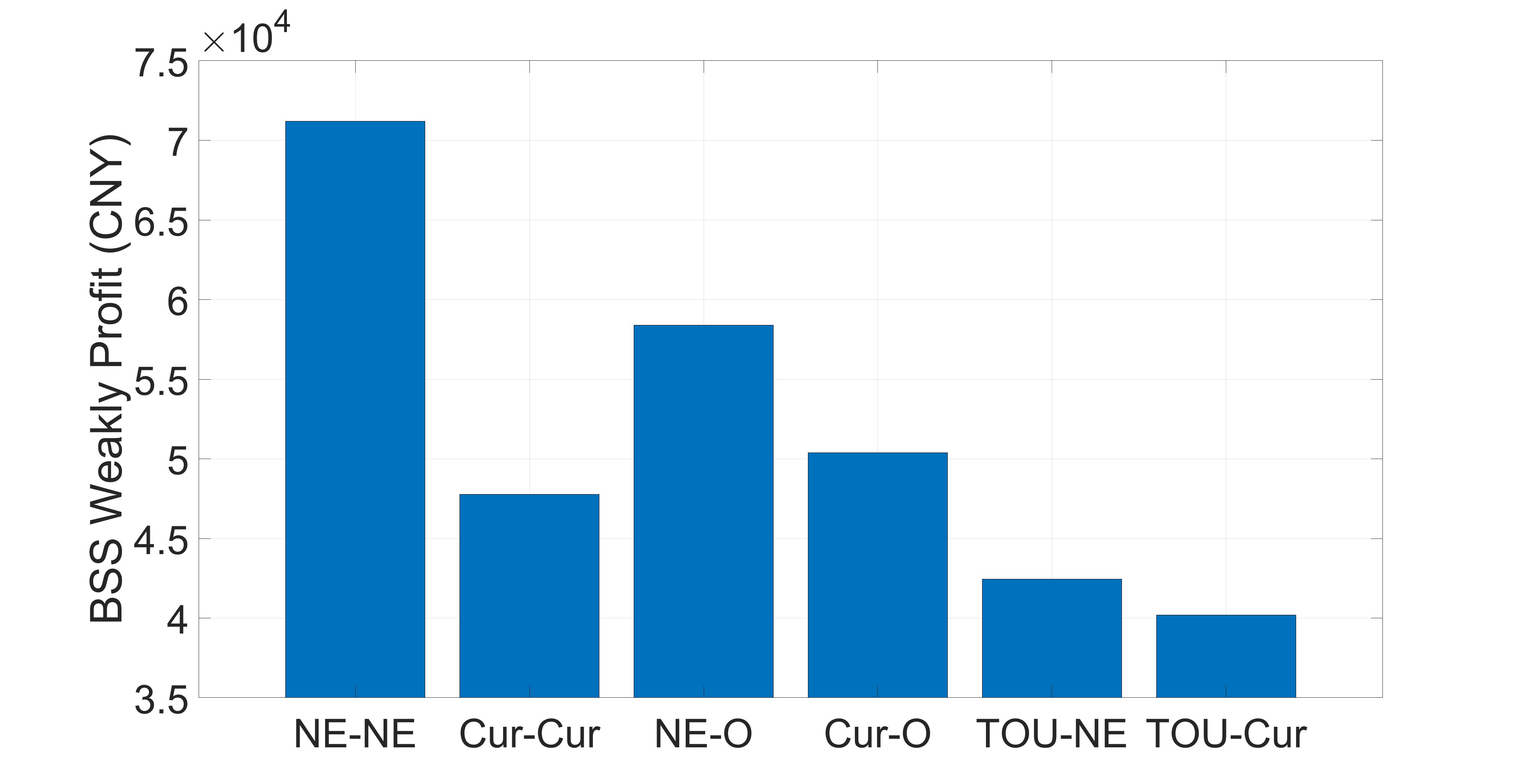}
   \caption{BSS 7 Weakly profits.}
   \label{fig:bm_BSS7}
\end{subfigure}
\hfill
\begin{subfigure}[t]{0.3\textwidth}
   \includegraphics[width=\linewidth]{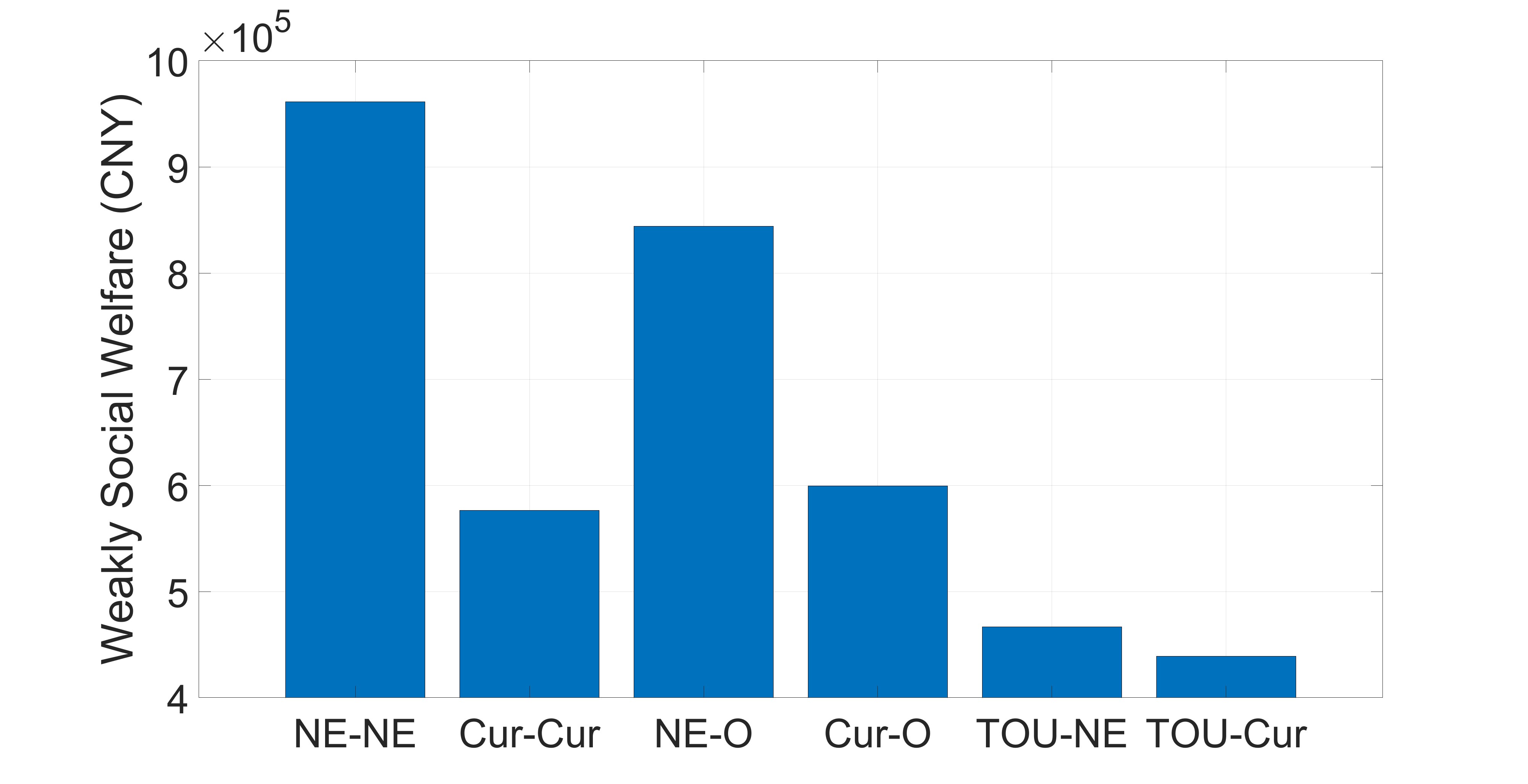}
   \caption{Social welfare.}
   \label{fig:bm_sw}
\end{subfigure}
\hfill
\begin{subfigure}[t]{0.3\textwidth}
   \includegraphics[width=\linewidth]{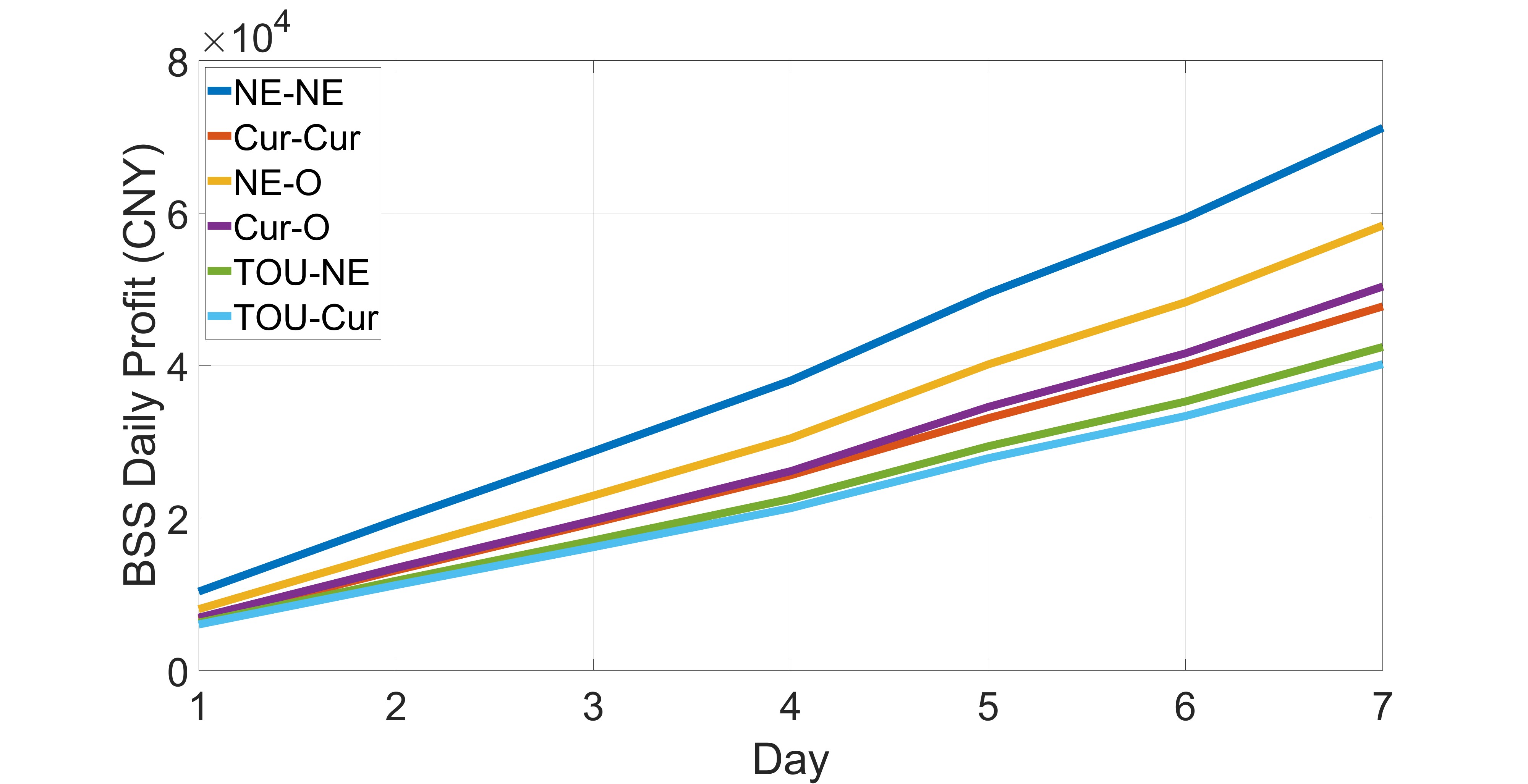}
   \caption{BSS 7 cumulative profit.}
   \label{fig:bm_day}
\end{subfigure}
\hfill
\caption{Profits comparison with benchmarks in 12-BSS system. The first abbreviation is short for pricing and the second is short for charging, e.g., NE-O denotes proposed Nash equilibrium (NE) pricing strategy and solely optimization charging (O) strategy.}
\label{fig:bm}
\end{figure*}

As depicted in Figure~\ref{fig:bm}, the proposed NE-NE strategy surpasses all benchmarks. For instance, under the proposed strategy, BSS 7’s weekly profit is approximately 71,000 CNY, a 22\% increase compared to NE-O (around 58,400 CNY). The weekly social welfare under the proposed strategy is nearly 961k CNY, almost double that of TOU-NE (around 467k CNY). The performance of NE-O and TOU-NE significantly trails that of NE-NE, underscoring the importance of considering competition among BSSs in both pricing and charging strategies.

\subsubsection{Optimal profit under different battery number}

In this section, we investigate how the number of batteries, $N_k$, affects the BSS's optimal profit in the system. In this experiment, we focus on the potential for earning more profit, without considering the construction cost of adding more batteries. We select one BSS and increase the number of batteries while keeping other parameters constant. Additionally, we adjust the transmission capacity to $X_M=\frac{1}{2}CN$ to ensure that the stations have the ability to charge the batteries they own. We then simulate the system for a week and calculate the optimal profit of that BSS as well as its profit share in the system. The results are shown in Figure~\ref{fig:BSS_profit_diff_battnum}.

\begin{figure*}
\begin{subfigure}[t]{0.25\textwidth}
  \includegraphics[width=\linewidth]{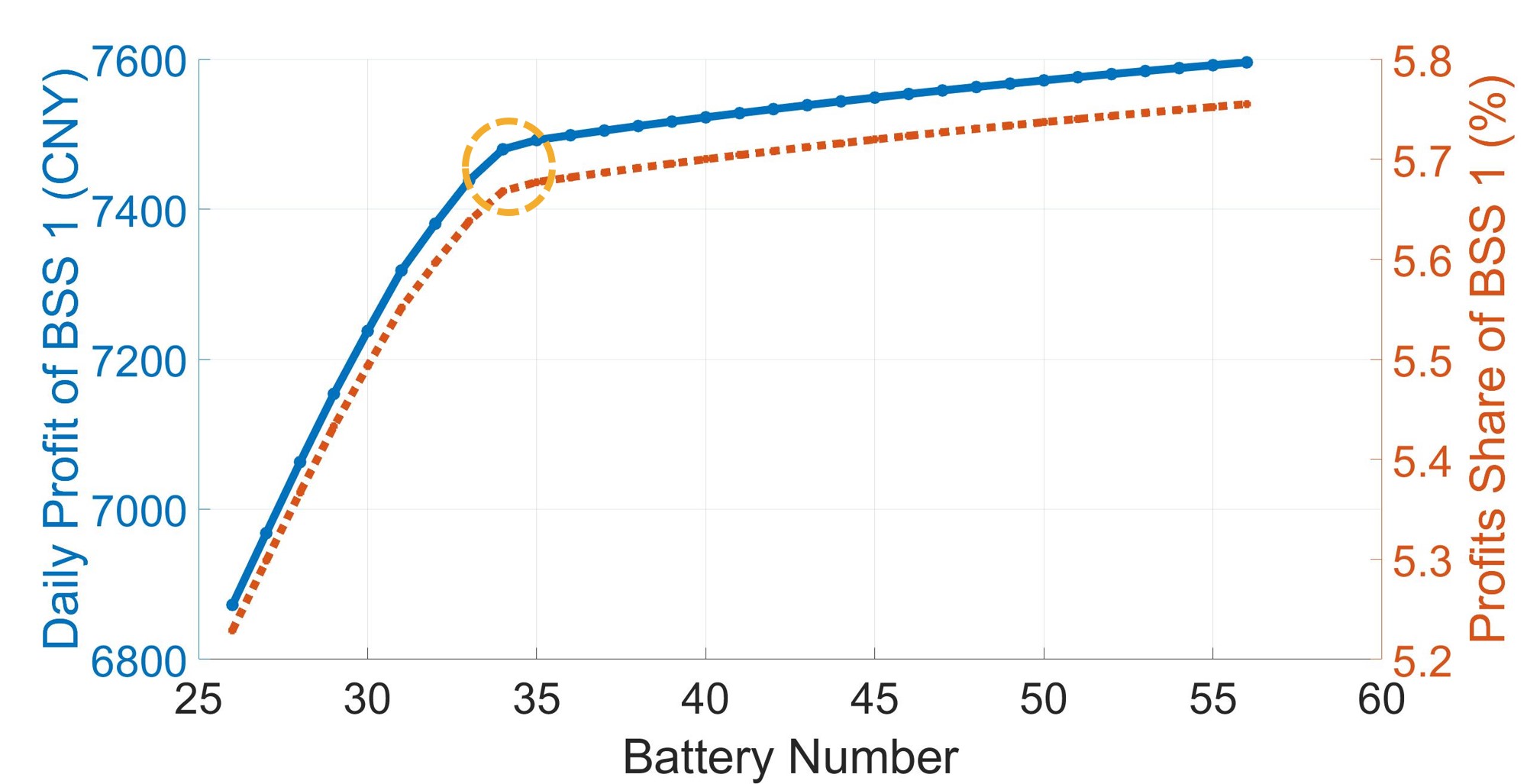}
  \caption{BSS 1}
  \label{fig:BSS1_profit_diff_battnum}
\end{subfigure}%
\hfill % maximize the horizontal separation
\begin{subfigure}[t]{0.25\textwidth}
  \includegraphics[width=\linewidth]{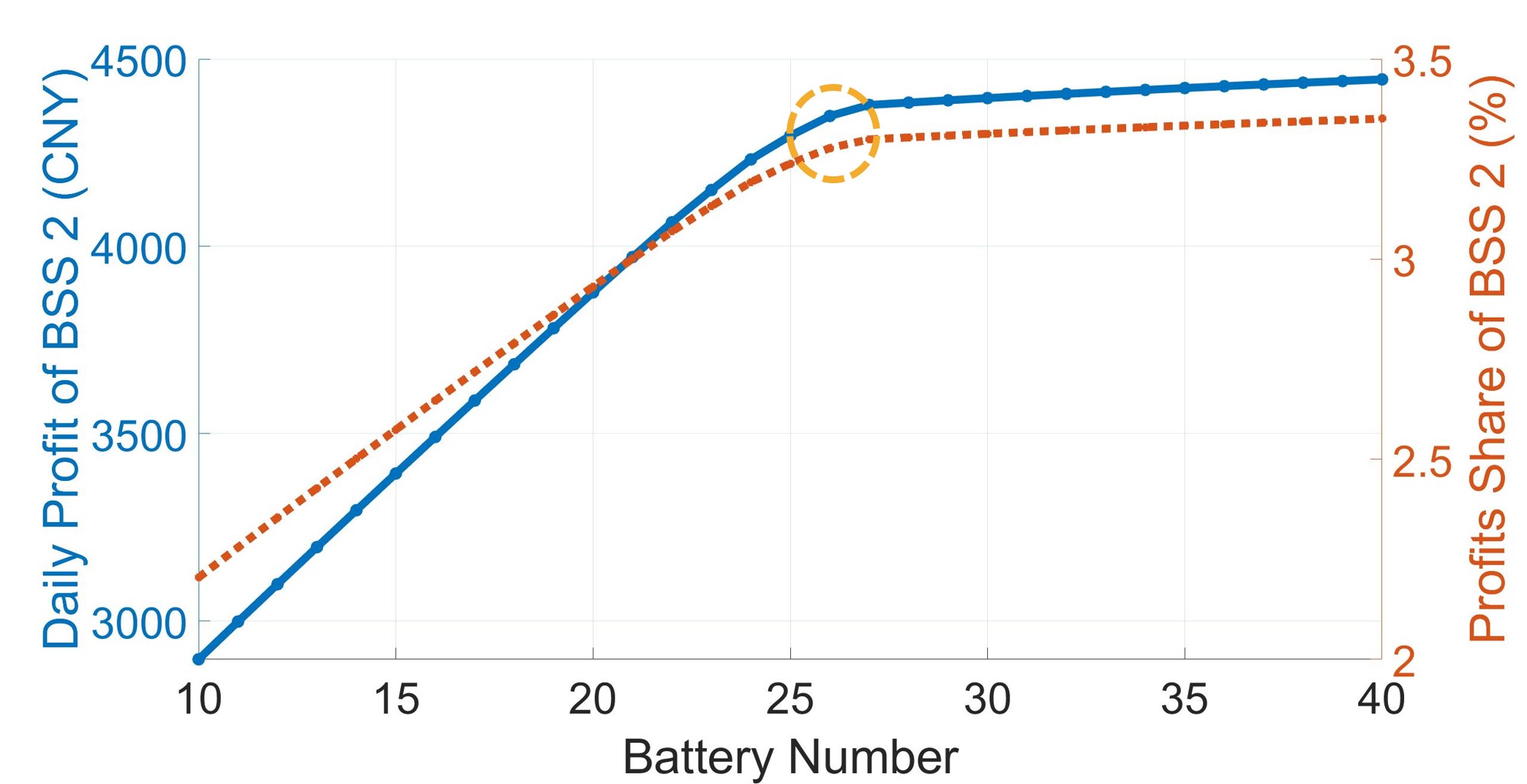}
  \caption{BSS 2}
  \label{fig:BSS2_profit_diff_battnum}
\end{subfigure}%
\hfill
\begin{subfigure}[t]{0.25\textwidth}
  \includegraphics[width=\linewidth]{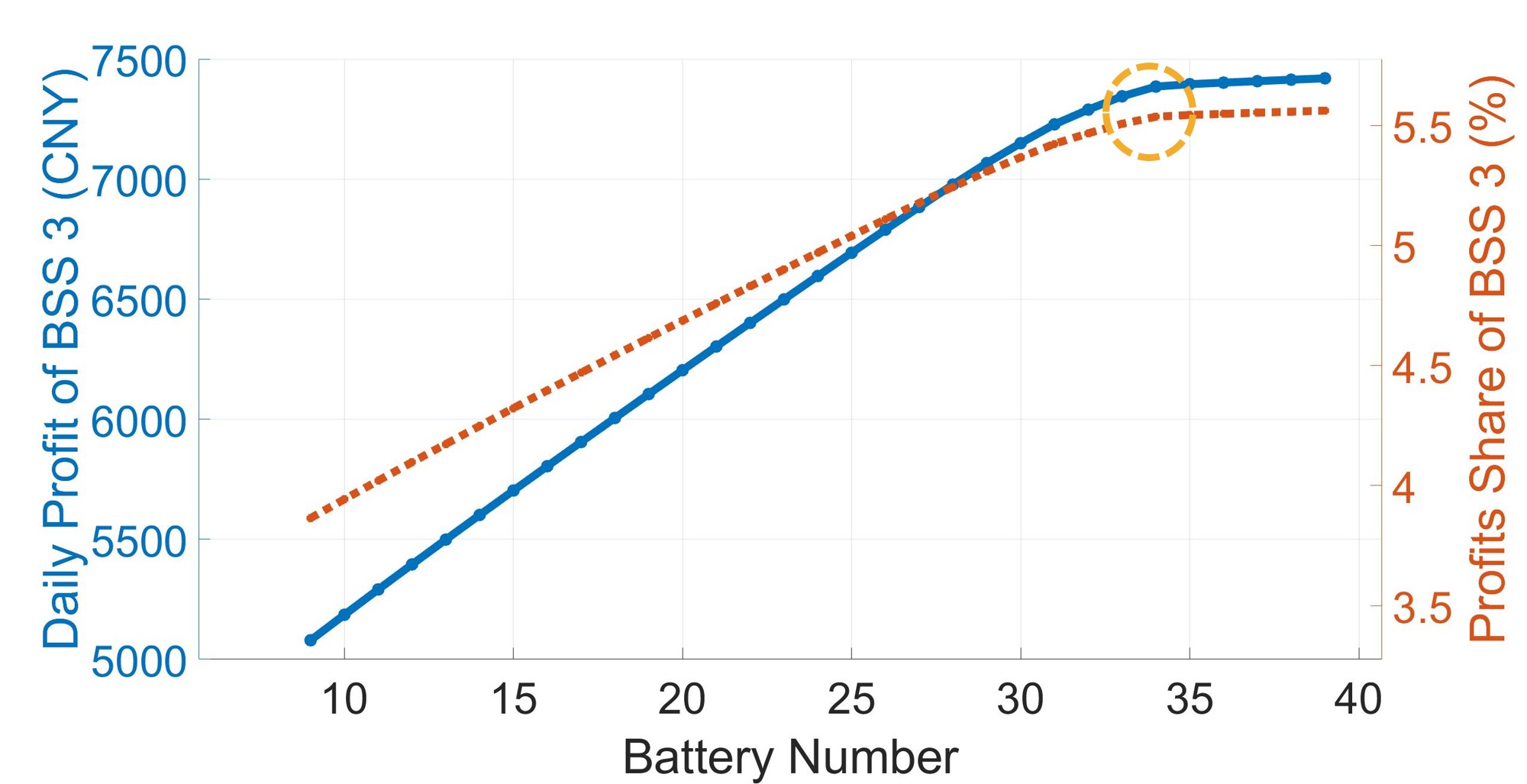}
  \caption{BSS 3}
  \label{fig:BSS3_profit_diff_battnum}
\end{subfigure}%
\hfill
\begin{subfigure}[t]{0.25\textwidth}
  \includegraphics[width=\linewidth]{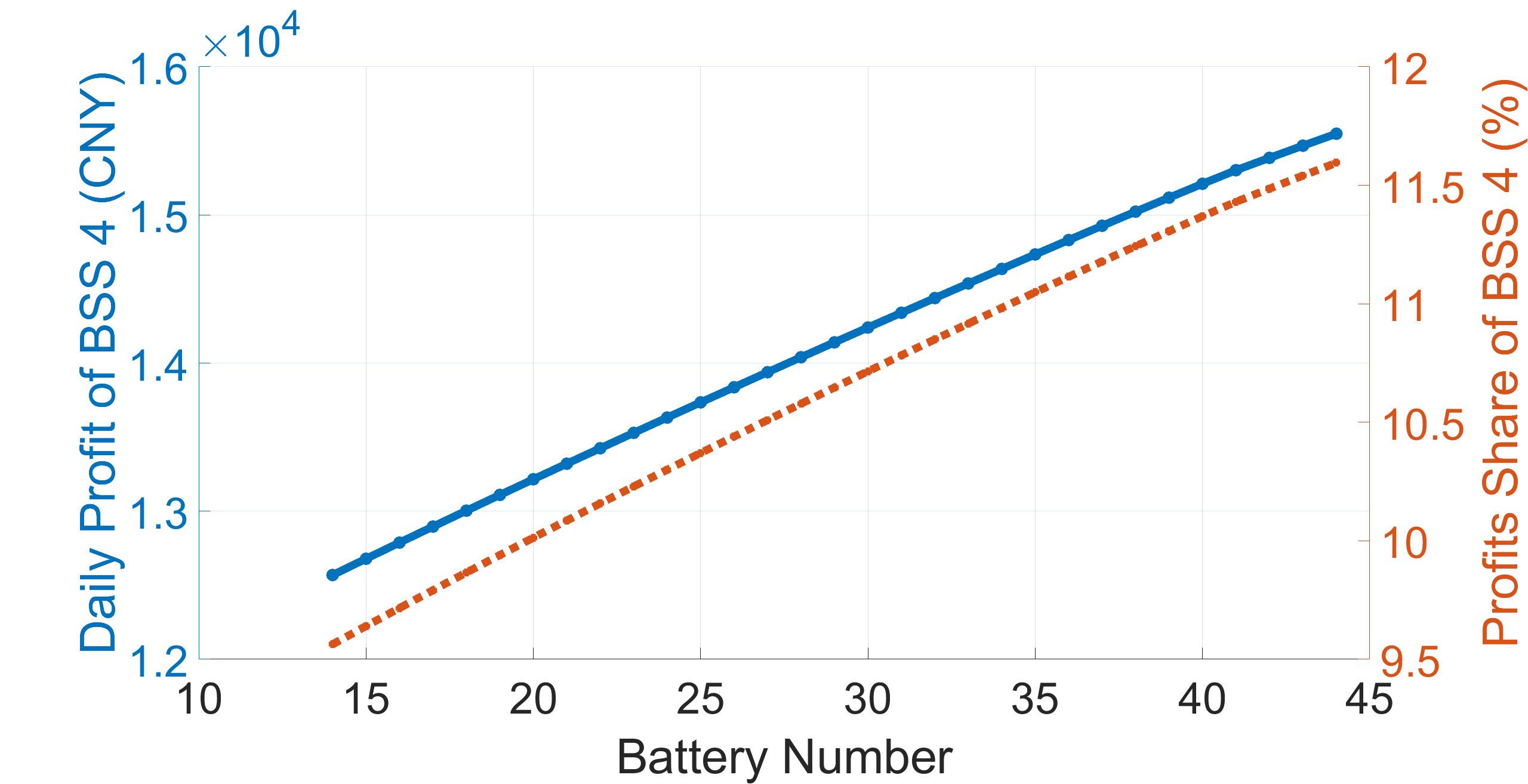}
  \caption{BSS 4}
  \label{fig:BSS4_profit_diff_battnum}
\end{subfigure}%
\hfill % maximize the horizontal separation
\caption{BSS daily profit and profit share under different battery numbers in 12-BSS system.}
\label{fig:BSS_profit_diff_battnum}
\end{figure*}

\begin{figure*}
\begin{subfigure}[t]{0.3\textwidth}
   \includegraphics[width=\linewidth]{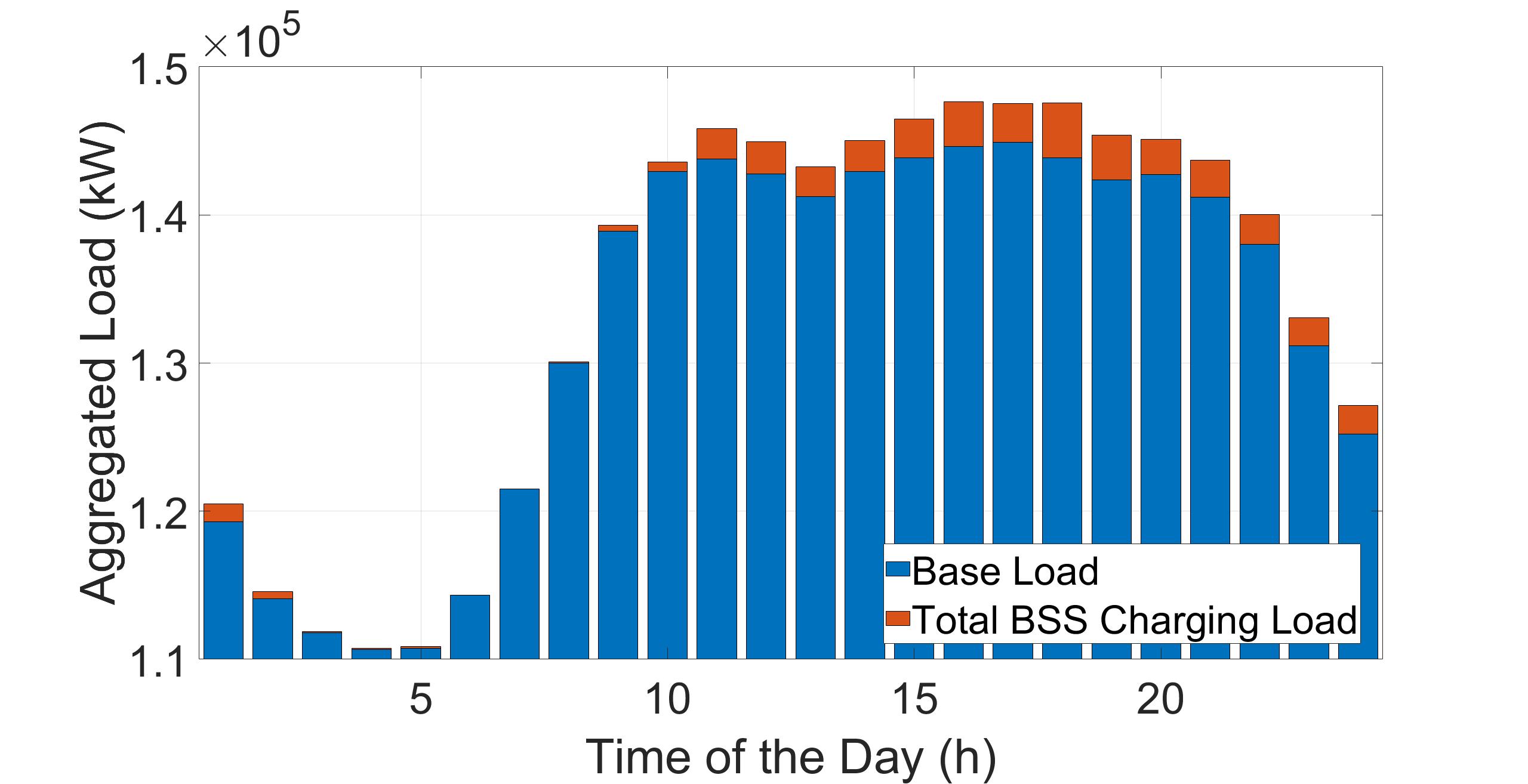}
   \caption{Uncoordinated charging, optimal pricing.}
   \label{fig:aggregated_load_share_cur}
\end{subfigure}
\hfill
\begin{subfigure}[t]{0.3\textwidth}
   \includegraphics[width=\linewidth]{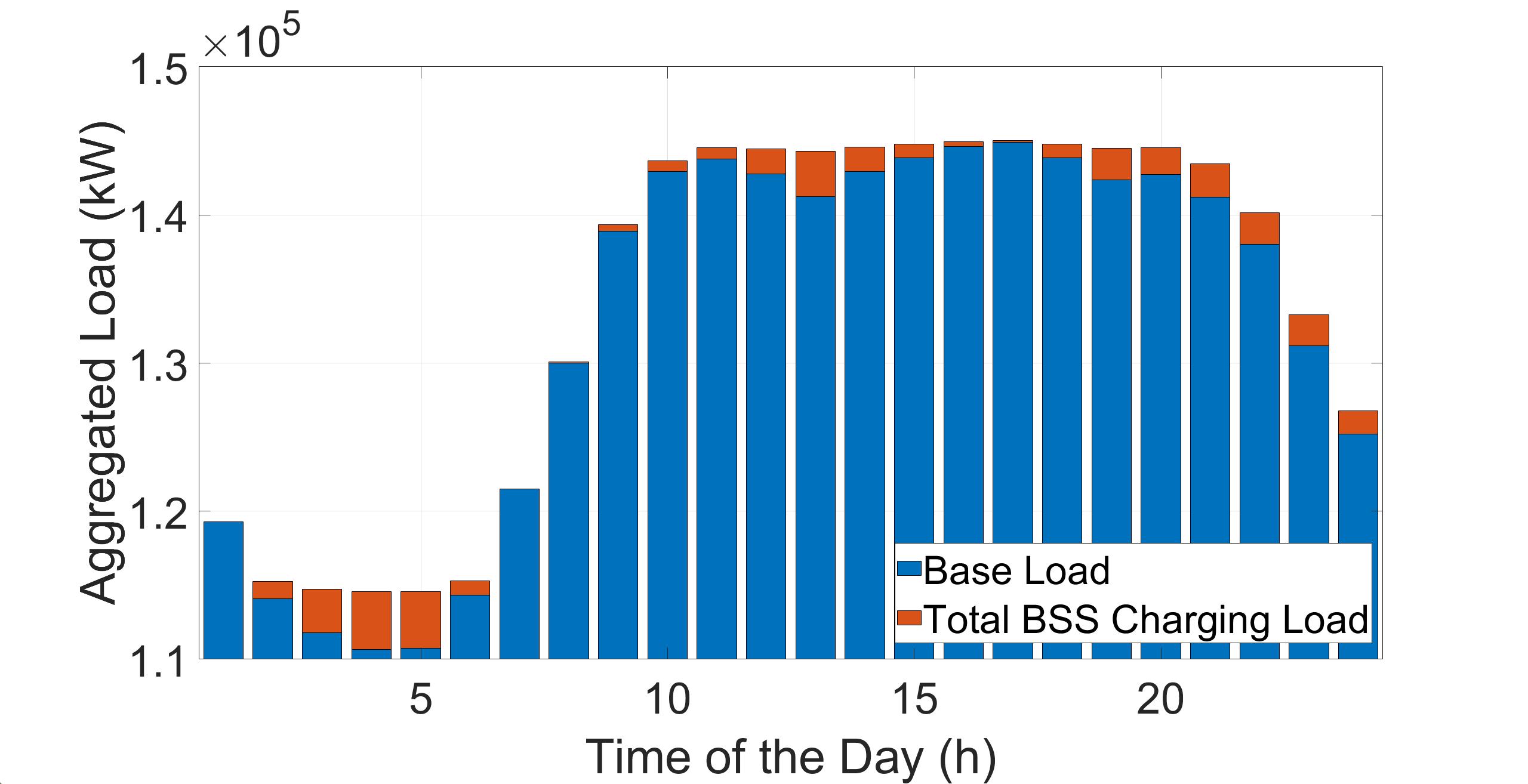}
   \caption{Optimal charging, optimal pricing.}
   \label{fig:aggregated_load_share_opt}
\end{subfigure}
\hfill
\begin{subfigure}[t]{0.3\textwidth}
   \includegraphics[width=\linewidth]{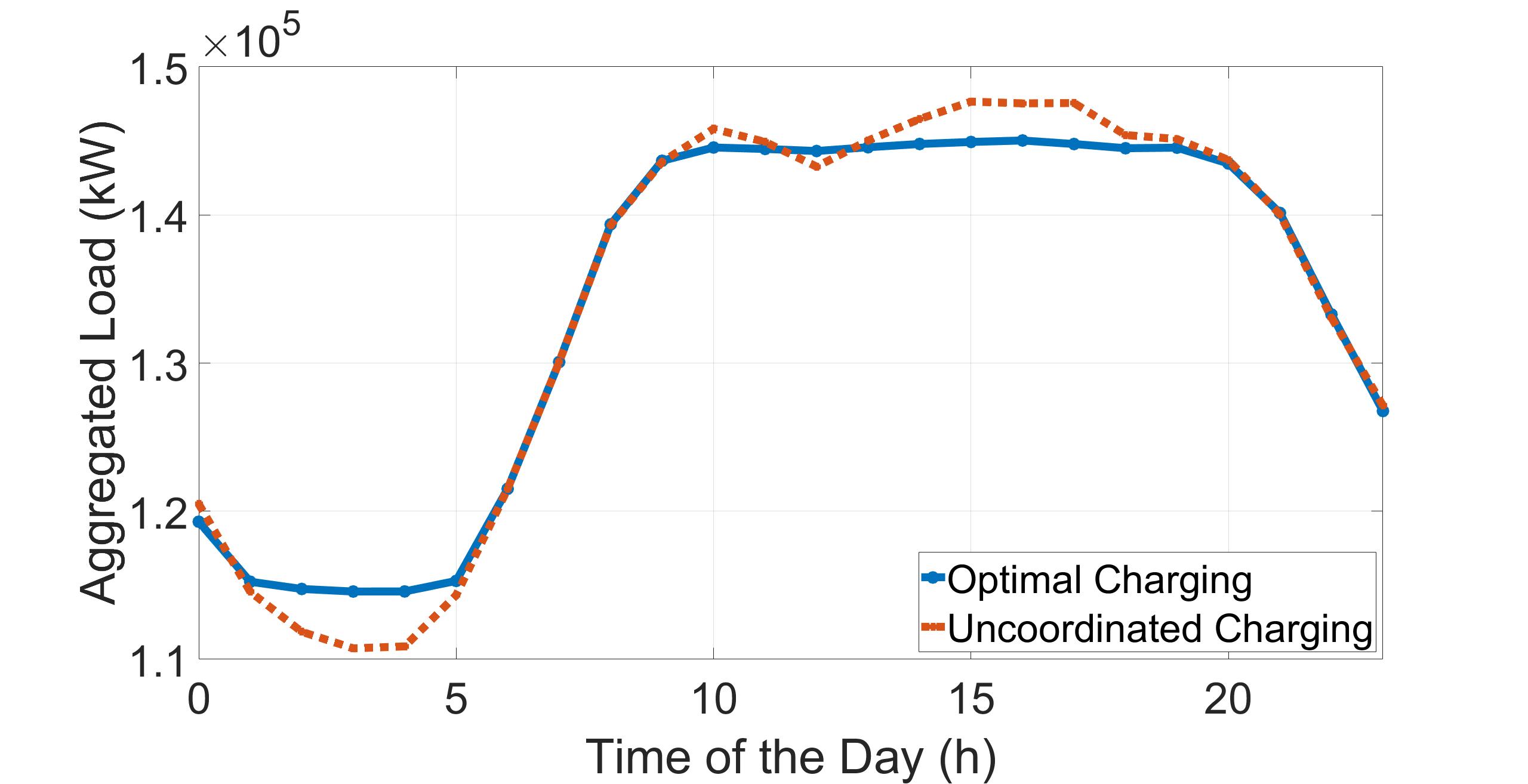}
   \caption{Aggregated load comparison.}
   \label{fig:aggregated_load_peak_shaving}
\end{subfigure}
\hfill
\caption{The aggregated load under different BSS strategies in 12-BSS system.}
\label{fig:peak_shaving_benefits}
\end{figure*}

As depicted in Figure~\ref{fig:BSS_profit_diff_battnum}, both the daily profit and profit share increase with an increase in the number of batteries. For instance, BSS 1 earns an average daily profit of around 6872 CNY and has a share of around 5.23\% in the system with 26 batteries. When the number of batteries increases to 30, the profit increases to around 7237 CNY, and the share increases to around 5.49\%. These results indicate that the BSSs with more batteries are more competitive in the system.

One interesting observation is that the increasing speed differs significantly at certain points, as circled in Figure~\ref{fig:BSS_profit_diff_battnum}. For example, the increasing speed of BSS 1 profit around the 34 is different. This is because an increase in the number of batteries brings two benefits to the station: higher station attractions in the pricing, and higher peak-shaving capacity to reduce the charging cost. However, as the number of batteries increases, the peak-shaving effect is limited by the swapping demand and cannot reduce more charging costs, leading to a change in the rate of profit increase.

\subsection{Peak Shaving Benefits for Aggregator}
In this section, we illustrate that the proposed aggregator pricing strategy shaves the peak load. We first demonstrate the total aggregated load, then present the detail load component and shares. Finally, we validate the peak shaving effect under electricity price difference $Pe_{BSS}-Pe_A$. In this experiment, we adopt the optimal pricing and battery charging strategies to illustrate the effect.
% this experiments also shows how aggregator benefit from adopting the time-of-use strategy shown in \eqref{eq:Pe}.

\subsubsection{Aggregated load under different BSS charging strategies}

In this experiment, we can observe that the proposed charging strategy shaves the total aggregated load. To illustrate this effect, we draw the aggregated load under different battery charging strategies, as shown in Figure~\ref{fig:peak_shaving_benefits}.

As depicted in Figure~\ref{fig:peak_shaving_benefits}, the TOU pricing in \eqref{eq:Pe} shaves the peak load. For example, during 17:00-18:00, the total aggregated load under optimal charging is 144788 kWh, which is only 98.12\% of the total load under unordered charging (147556 kWh).
% At midnight 3:00-4:00, the total aggregated load is 114576 kWh, which is 101.03\% of the unordered charging (110739 kWh).
\subsubsection{Variance of aggregated load with different contract price difference}

In Section~\ref{sec:model_aggregator}, we design a TOU pricing strategy based on $Pe_{BSS}$ and $Pe_A$. In this experiment, we investigate whether the peak shaving effect differs with the price difference $Pe_{BSS}-Pe_A$ in the long-term contract of the aggregator and the BSSs. To examine the effect, we calculate the variance of the total aggregated load of every time slot under different contract price differences $Pe_{BSS}-Pe_A$, and the result is shown in Figure~\ref{fig:var_diff_PeA}.

\begin{figure}[t]
\centerline{\includegraphics[width=\linewidth]{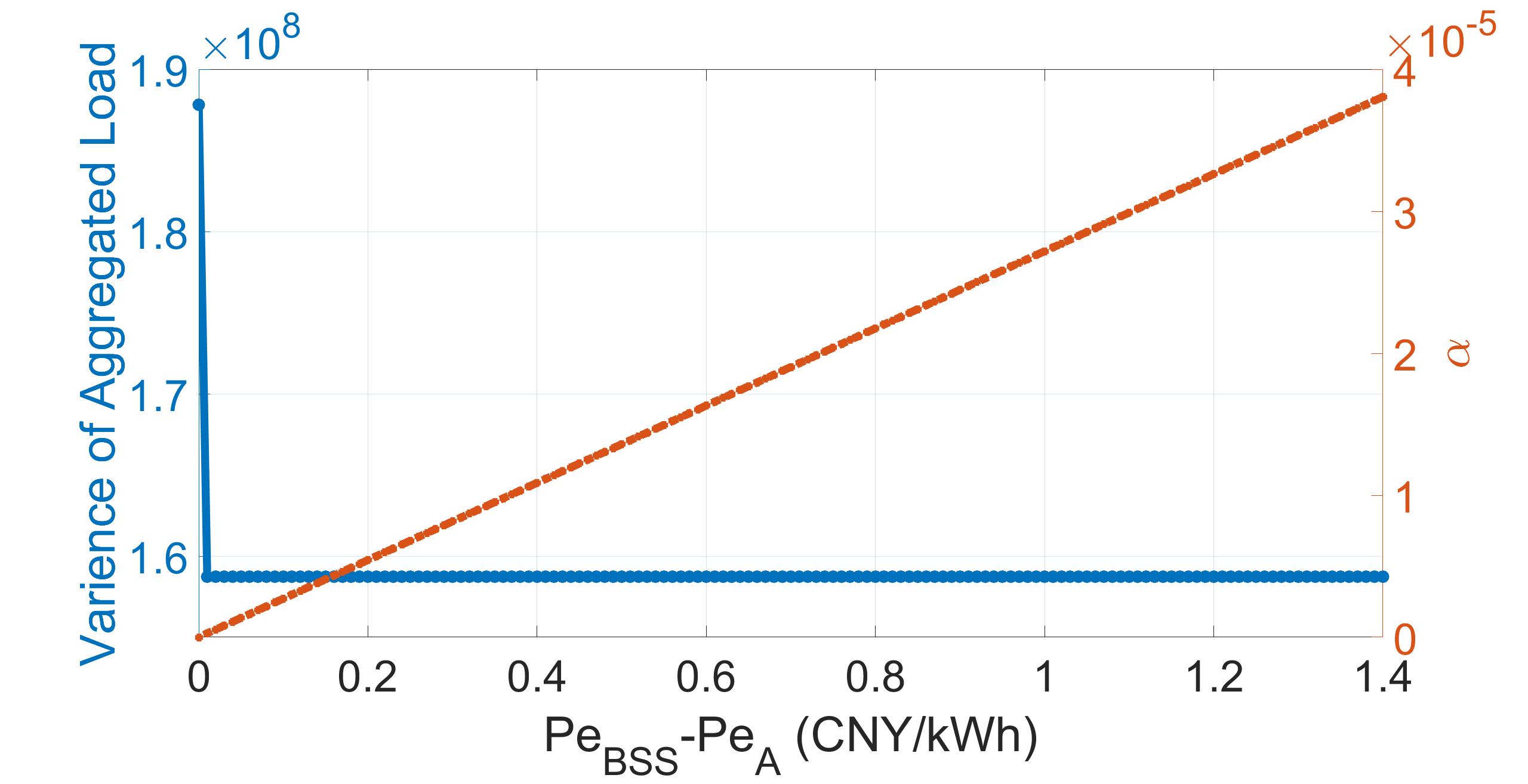}}
\caption{Variance of the aggregated load and $\alpha$ value under multiple electricity contract price difference $Pe_{BSS}-Pe_A$ in 12-BSS system.}
\label{fig:var_diff_PeA}
\end{figure}

As depicted in Figure~\ref{fig:var_diff_PeA}, when the contract price of the aggregator $Pe_A$ is smaller than that of the BSS $Pe_{BSS}$, the aggregator will achieve the same variance in load. However, as the difference approaches zero, the TOU pricing strategy becomes a fixed price strategy, since the parameter $\alpha$ becomes zero, leading to a significant increase in variance. This result indicates that the TOU pricing strategy achieves a better peak-shaving effect than the fixed price strategy. Moreover, to achieve the peak shaving effect, the aggregator does not need to negotiate an extremely low price in the long-term contract with the power grid.

\section{Discussion}
\label{sec:dis}
\subsection{Equilibrium Stability}

In this subsection, we examine the stability of the equilibrium with varying initial points. Initial points for $Ps_i^{old}$ and $X_{k,t}^{old}$ are randomly generated 50 times within the range $[1,10^4]$, and the pricing and charging equilibrium is calculated. The statistics of the initial values and algorithm results are presented in Figure~\ref{fig:stable_p} and Figure~\ref{fig:stable_x}. Due to space constraints, we display the charging result for BSS 1. As depicted in the figures, the equilibrium points remain consistent, validating the initial points.

\begin{figure}[t]
\centerline{\includegraphics[width=\linewidth]{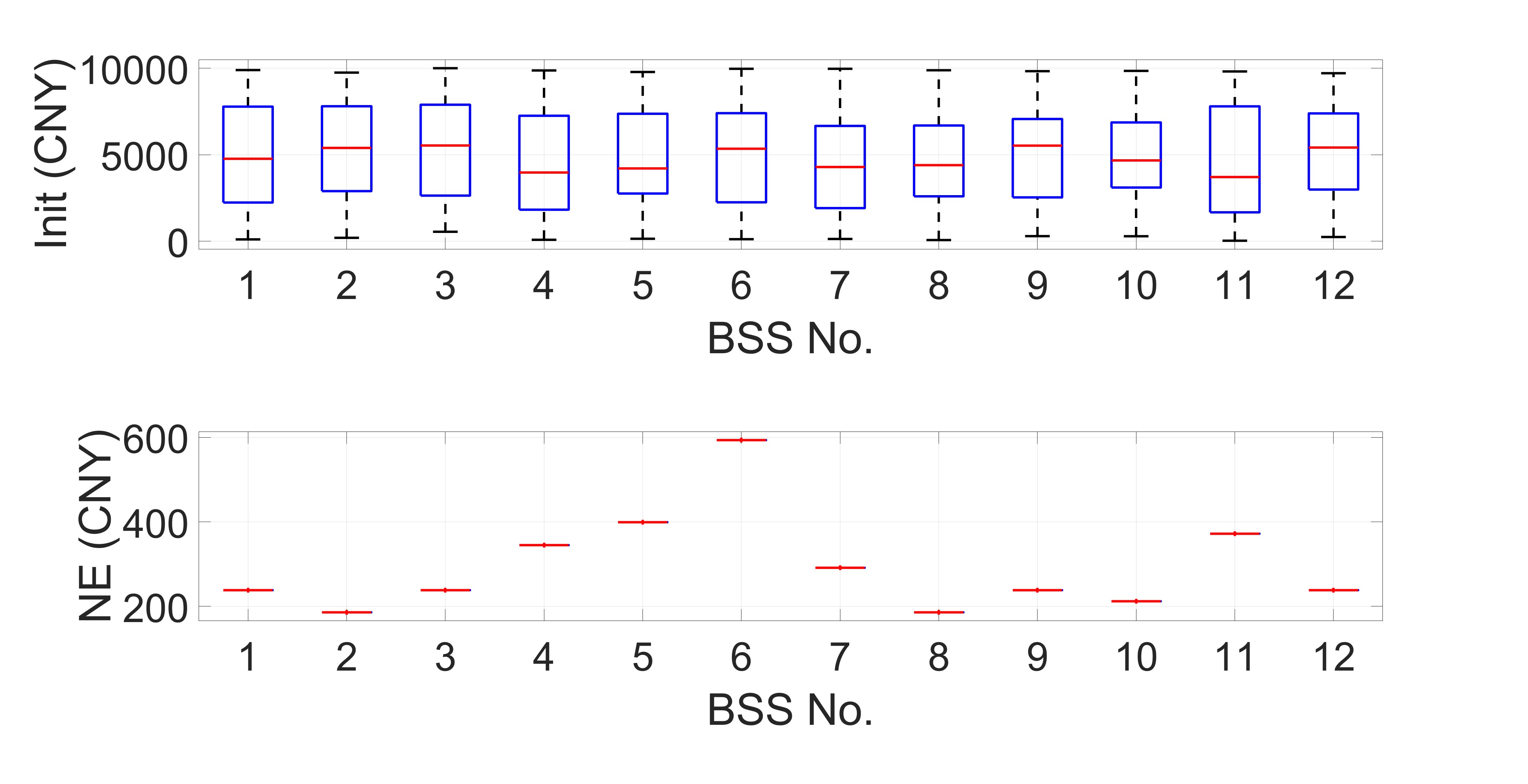}}
\caption{Pricing Equilibrium Stability in 12-BSS system.}
\label{fig:stable_p}
\end{figure}

\begin{figure}[t]
\centerline{\includegraphics[width=\linewidth]{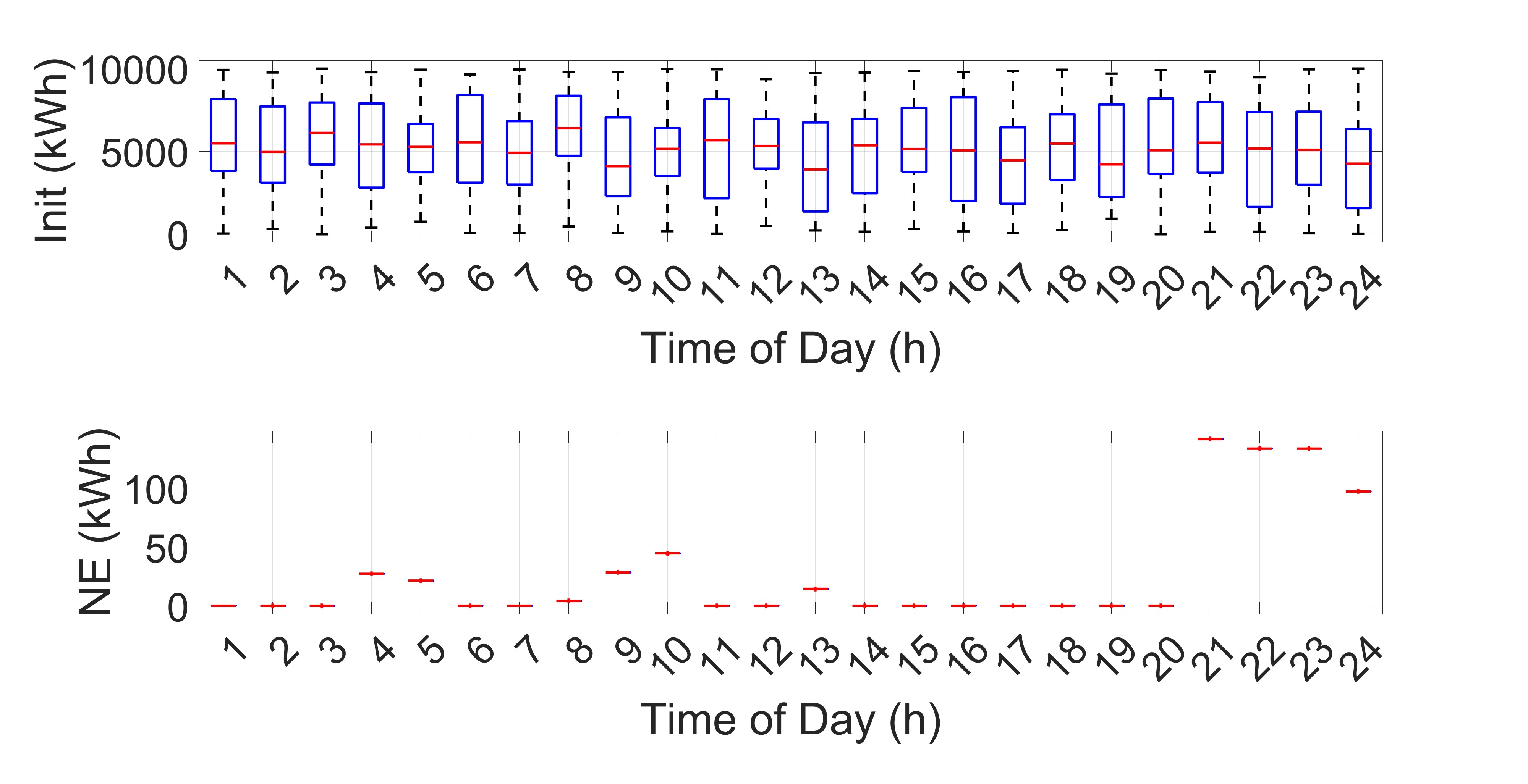}}
\caption{Charging Equilibrium Stability in 12-BSS system. The shown is the charging equilibrium of BSS 1.}
\label{fig:stable_x}
\end{figure}

\subsection{Effectiveness under a Larger System}

To demonstrate the effectiveness of our algorithm in larger systems, we extend the scenario to a 36-BSS system. We calculate the profit of BSS 7 and the cumulative profit of all BSSs, representing the system’s social welfare. In this system, we set $C_d = 0.5$ to validate the effectiveness under varying weighting parameters. BSS parameters, including the battery number (randomly generated within $[10,40]$) and the distance advantage (within $[5,25]$), are randomly generated. The charging station price is set to the average real-life price of 1.5 CNY/kWh. The swapping demand is tripled compared to the 12-BSS system. All other parameters remain consistent with the previous 12-BSS system. The optimal pricing and average daily charging strategies are depicted in Figure~\ref{fig:opt_pricing_charging_large}. Due to space constraints, we only display the charging strategies of BSS 1-12. The social welfare under different pricing and charging strategies is presented in Figure~\ref{fig:sw_larger_system}. As shown in Figure~\ref{fig:sw_larger_system}, our proposed strategy improves the social welfare.

\begin{figure*}
\begin{subfigure}[t]{0.3\textwidth}
   \includegraphics[width=\linewidth]{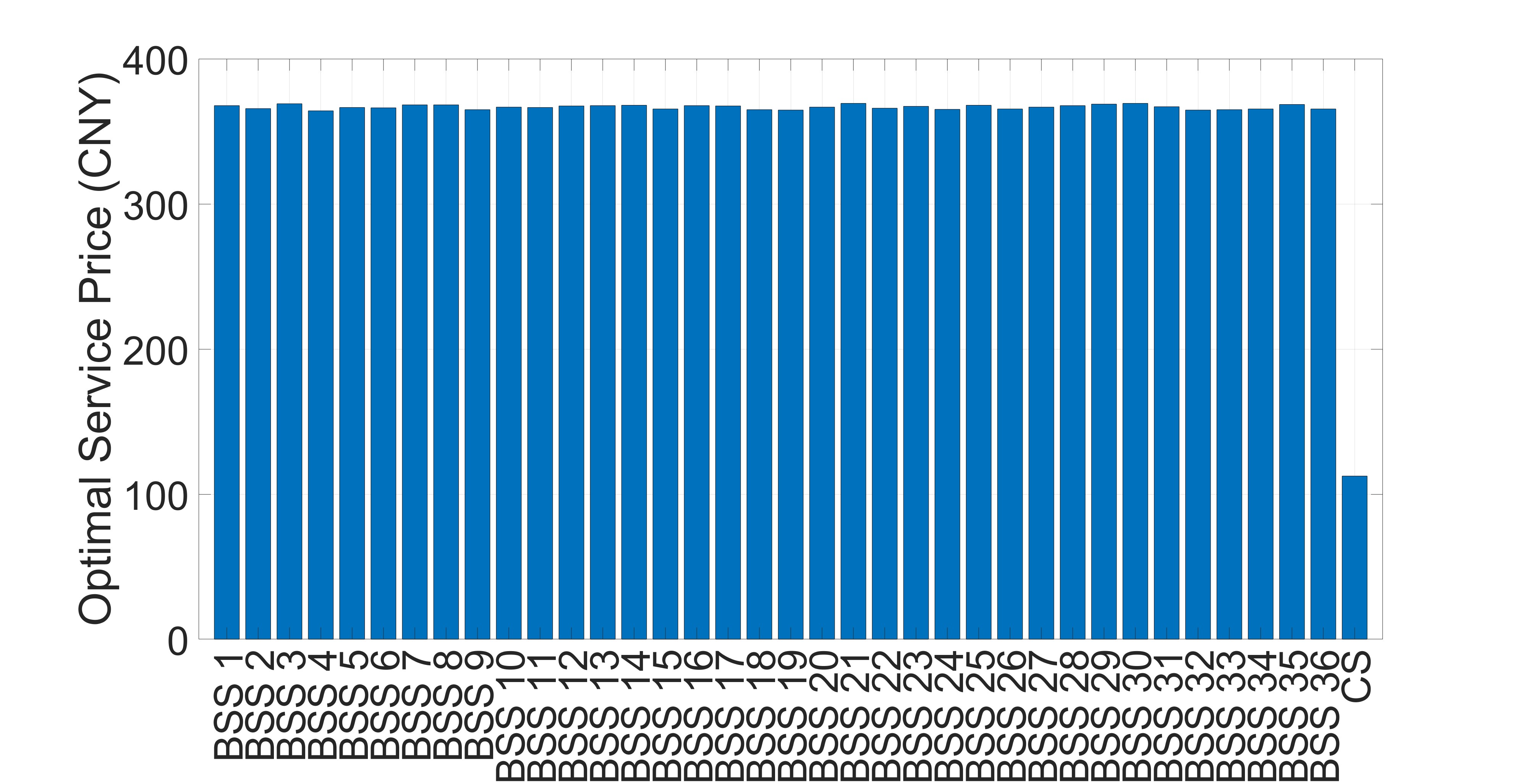}
   \caption{Optimal pricing.}
   \label{fig:opt_pricing_36}
\end{subfigure}
\hfill
\begin{subfigure}[t]{0.3\textwidth}
   \includegraphics[width=\linewidth]{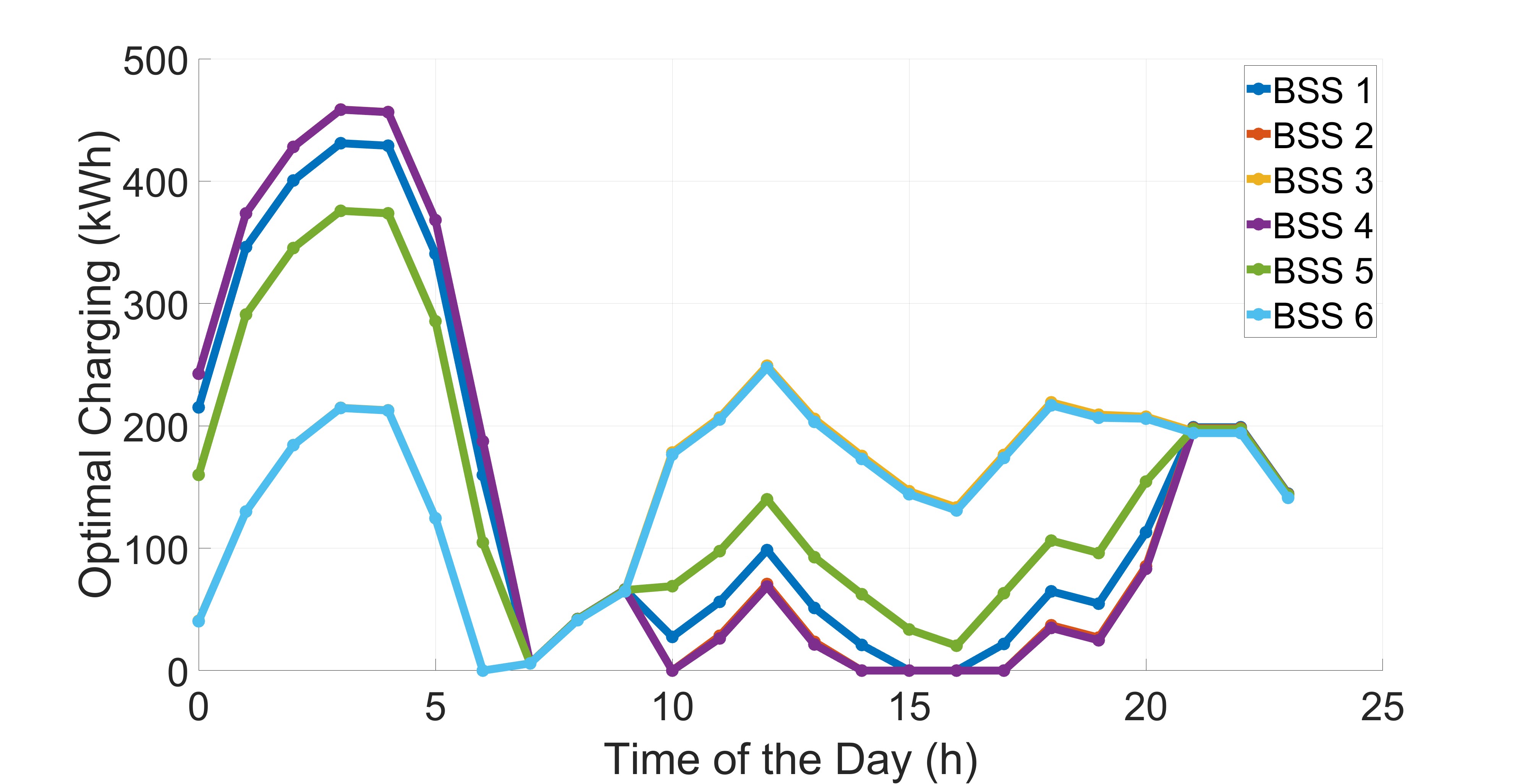}
   \caption{Optimal charging of BSS 1-6.}
   \label{fig:opt_charging_1_6_large}
\end{subfigure}
\hfill
\begin{subfigure}[t]{0.3\textwidth}
   \includegraphics[width=\linewidth]{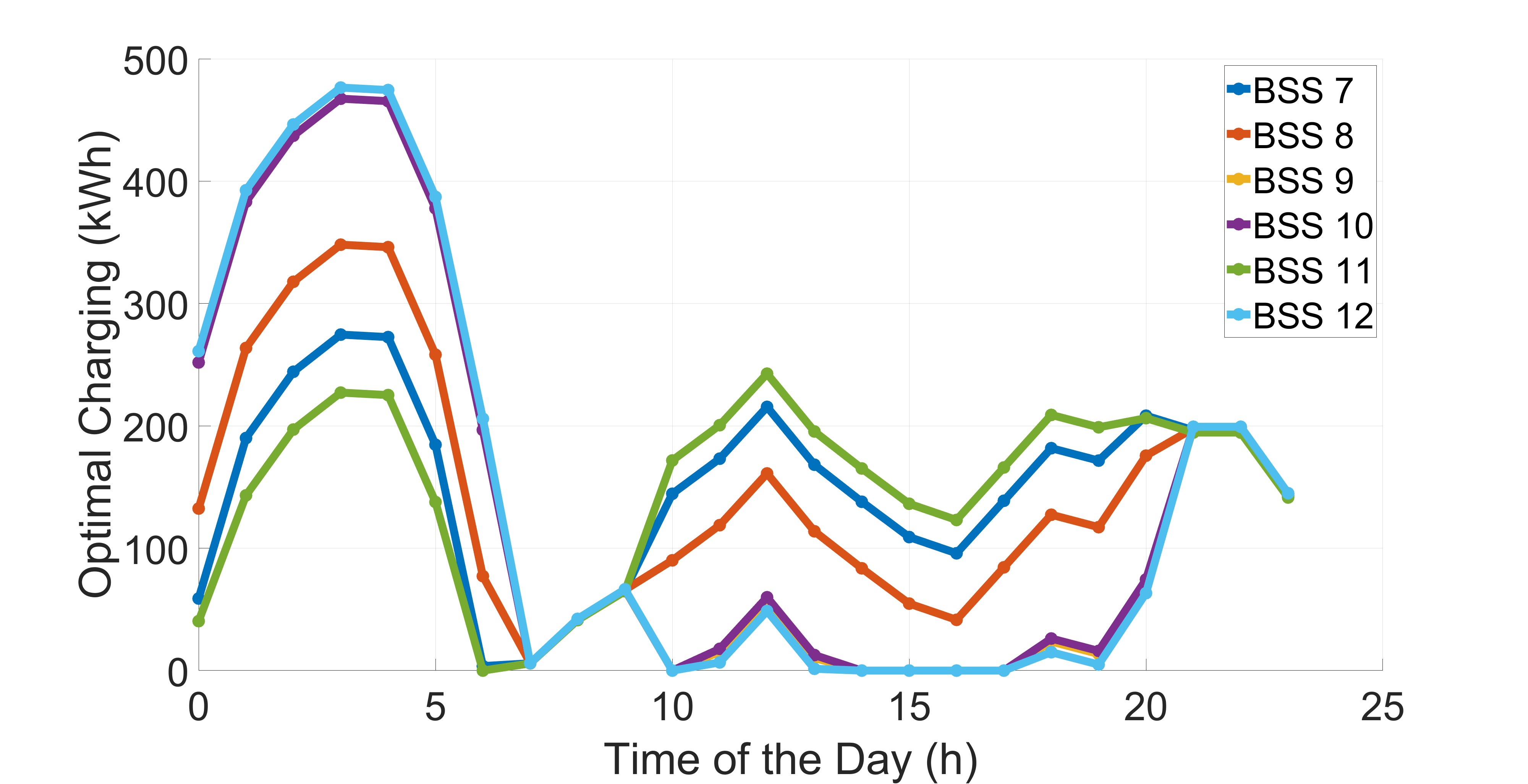}
   \caption{Optimal charging of BSS 7-12.}
   \label{fig:opt_charging_7_12_large}
\end{subfigure}
\hfill
\caption{The optimal pricing and average daily optimal charging strategy for 36-BSS system. The results are quite similar since all BSS parameters are generated from the same distribution.}
\label{fig:opt_pricing_charging_large}
\end{figure*}

\begin{figure*}
\begin{subfigure}[t]{0.3\textwidth}
   \includegraphics[width=\linewidth]{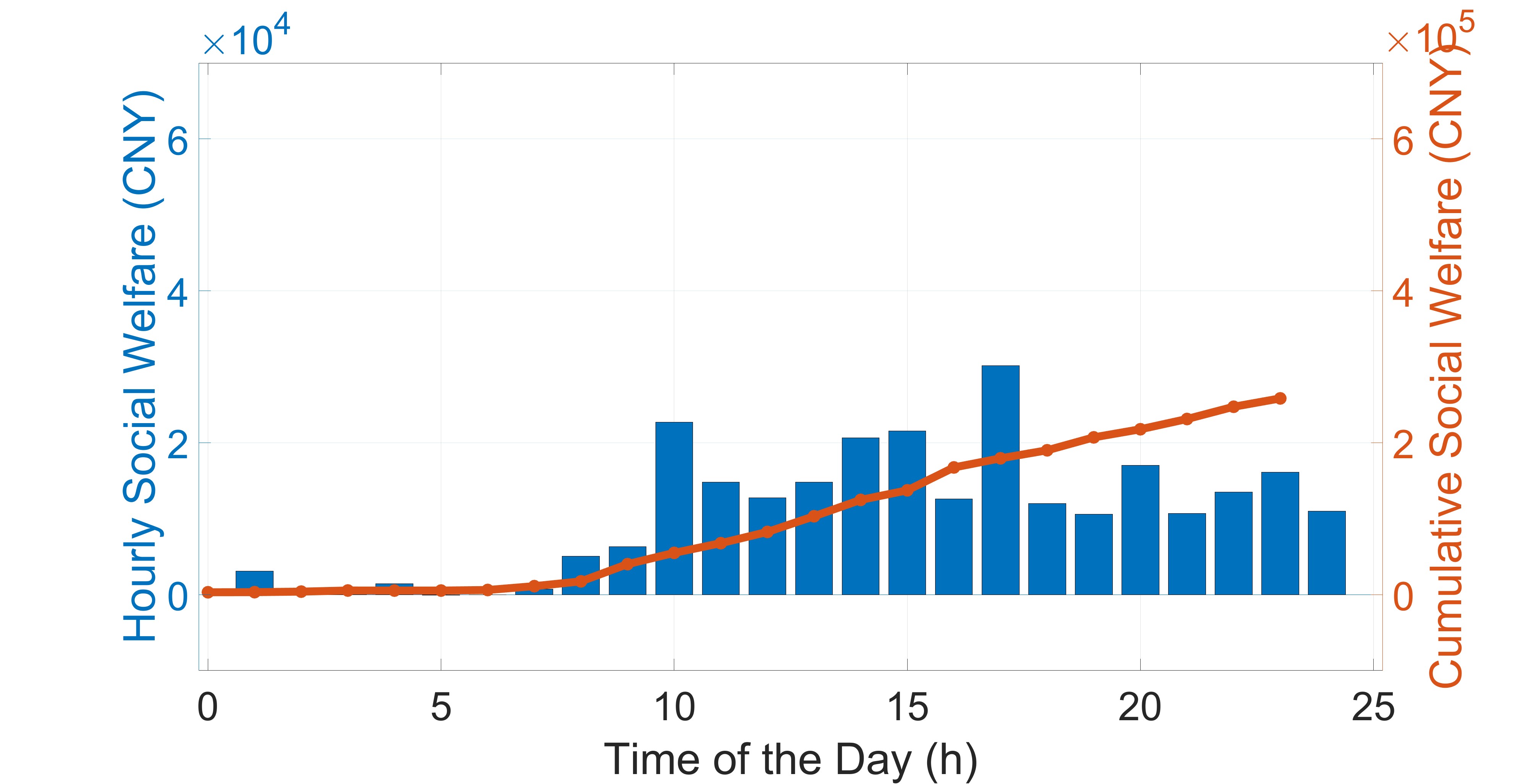}
   \caption{Uncoordinated charging, optimal pricing.}
   \label{fig:sw_larger_system_cur_cur}
\end{subfigure}
\hfill
\begin{subfigure}[t]{0.3\textwidth}
   \includegraphics[width=\linewidth]{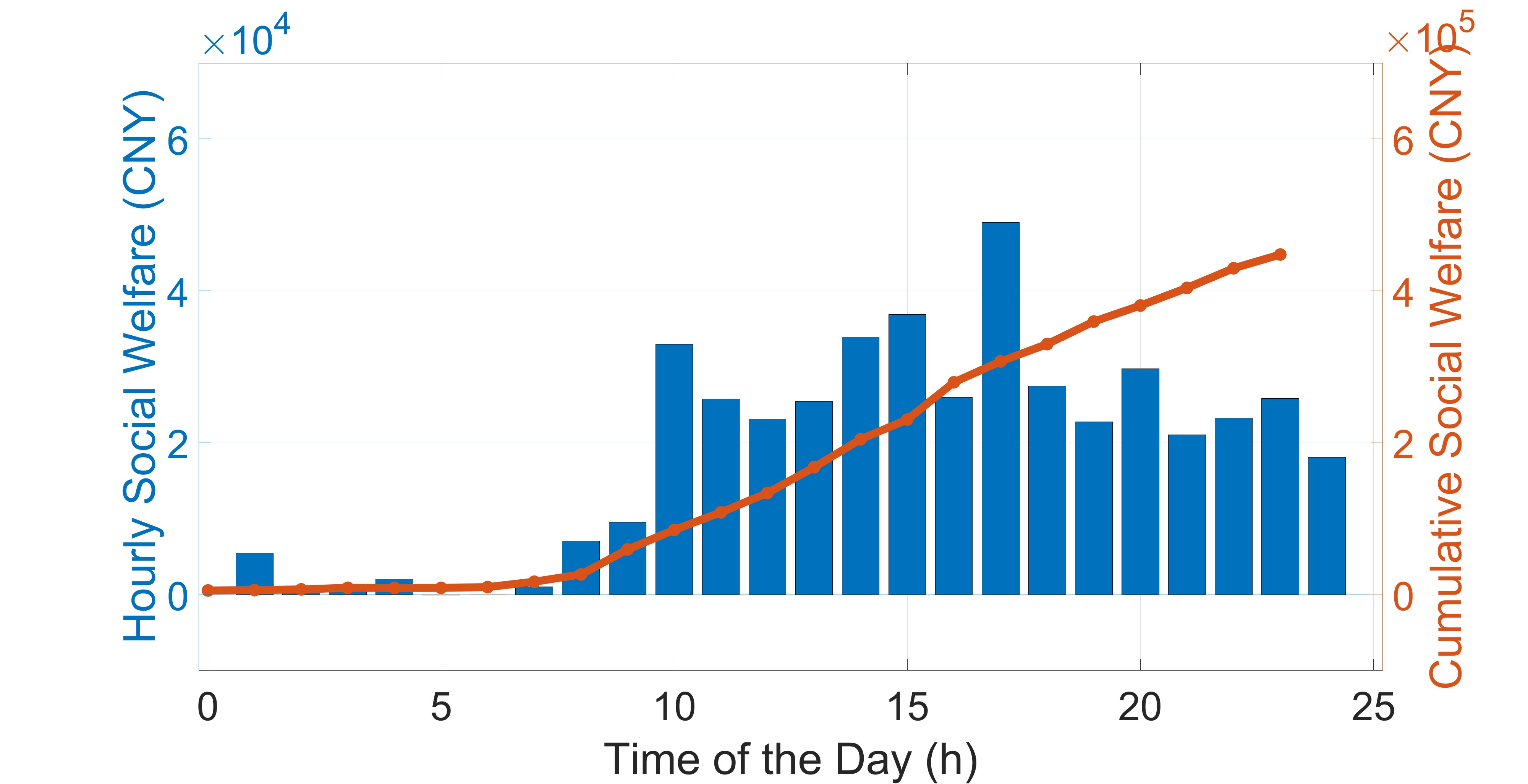}
   \caption{Optimal charging, optimal pricing.}
   \label{fig:sw_larger_system_cur_opt}
\end{subfigure}
\hfill
\begin{subfigure}[t]{0.3\textwidth}
   \includegraphics[width=\linewidth]{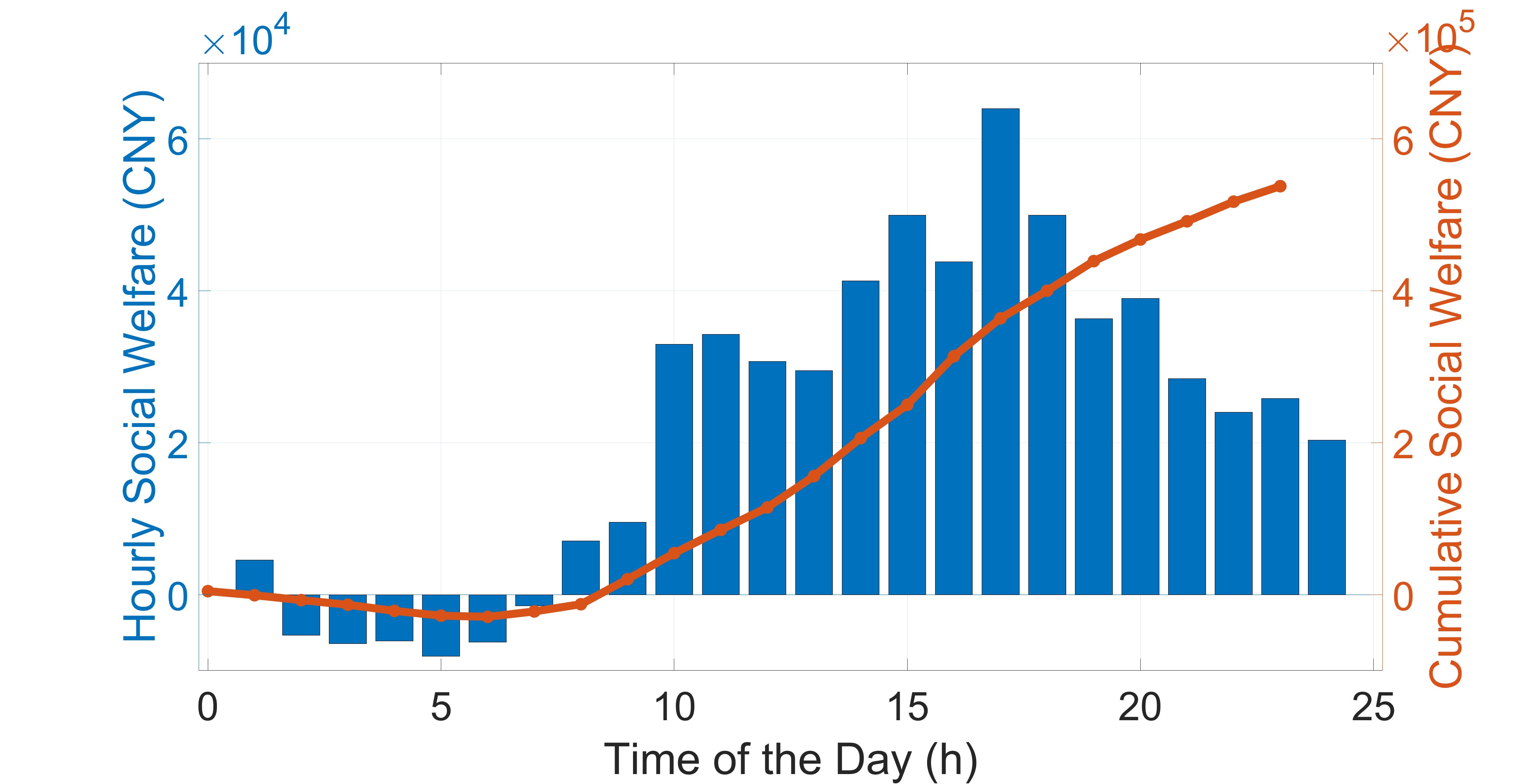}
   \caption{Aggregated load comparison.}
   \label{fig:sw_larger_system_opt_opt}
\end{subfigure}
\hfill
\caption{The social welfare under different BSS strategies in 36-BSS system.}
\label{fig:sw_larger_system}
\end{figure*}

\subsection{Extension to Heterogeneous EVs}
\label{ssec:heter_EVs}
In Section~\ref{sec:model}, we present a demand response model defined by isomorphic parameters $C_p$, $C_b$, and $C_d$. This subsection broadens our analysis to two heterogeneous EV settings. We posit two EV groups with distinct preferences: $C_p^1$, $C_b^1$, and $C_d^1$ for the first group, and $C_p^2$, $C_b^2$, and $C_d^2$ for the second. The total demand at each time slot $t$ for the groups is represented as $R_t^1$ and $R_t^2$. Each group's demand response model is the same as the isomorphic parameters. In this context, the BSS utility function is defined as follows:
\begin{align}
    \Pi_k &= \Pi_k^1 + \Pi_k^2 \nonumber \\
    &= Ps_k R_{k,t}^1 + Ps_k R_{k,t}^2 - X_{k,t} Pe_t - \mu (X_{k,t}  + R_{k,t}^1 + R_{k,t}^2) \nonumber \\
    &= (R_t^1\frac{U_k^1}{\sum U_j^1 + U_c^1} + R_t^2\frac{U_k^2}{\sum U_j^2 + U_c^2}) (Ps_k - \mu) -  X_{k,t} Pe_t.
\end{align}
From the perspective of charging, the introduction of the new utility function does not alter the ratio of $X_{k,t}$. As a result, the previously charging results, including the existence and uniqueness in the charging subgame, continue to be valid. When considering the pricing aspect, this utility function is an affine transformation of the isomorphic one, thereby preserving its concavity with respect to $Ps_k$. Consequently, the pricing equilibrium is guaranteed to exist. The proposed algorithm method remains a viable approach for obtaining an equilibrium from a given initial point. However, it is important to note that the uniqueness of the equilibrium may not be assured. In such instances, an algorithm akin to the particle swarm optimization technique can be employed. This involves dispersing a multitude of starting points throughout the solution space in an attempt to identify all potential equilibria.

\section{Conclusions and Future Directions}
\label{sec:conclusion}
This paper proposes a hierarchical game to study the BSS optimal pricing and charging strategy under an aggregator that adopts TOU pricing. We prove the existence and uniqueness of the SPNE and propose an algorithmic solution for SPNE with convergence proved. Simulations with real-life data show that our pricing and charging significantly improve the stations' profit compared to the strategy used in real life. In particular, 2/3 of stations improve the profit by more than 40\%, and the slightest improvement is 18.1\%. As for further directions, it is interesting to investigate the time-verifying swapping service price. Besides, one may consider the scenario where the station can inject energy back into the grid (V2G).

\ifCLASSOPTIONcaptionsoff
  \newpage
\fi

\bibliographystyle{IEEEtran}

\bibliography{ref.bib}

\begin{IEEEbiography}[{\includegraphics[width=1in,height=1.25in,clip,keepaspectratio]{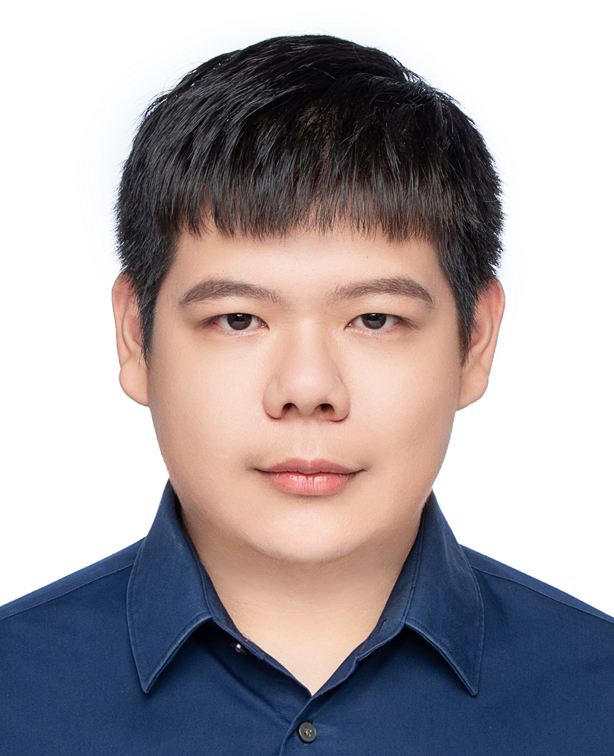}}]{Huanyu Yan}
Huanyu Yan obtained the B.Eng. degree in Computer Science and Technology from the Northwestern Polytechnical University, Shaanxi, China, in 2020. He is currently pursuing the Ph.D. degree in Computer and Information Engineering with the Chinese University of Hong Kong (Shenzhen), Shenzhen, China.

His research interests include electric vehicle charging scheduling, charging station pricing and optimizations, applications of game theory and economic mechanisms.
\end{IEEEbiography}

\begin{IEEEbiography}[{\includegraphics[width=1in,height=1.25in,clip,keepaspectratio]{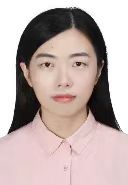}}]{Chenxi Sun}
Chenxi Sun received the B.Eng. degree and the Ph.D. degree in Electrical and Electronic Engineering from The University in Hong Kong in 2015 and in 2020 respectively. She is currently a Research Associate at the Shenzhen Institute of Artificial Intelligence and Robotics for Society, Shenzhen, China. Her research interests include optimization and machine learning in smart cities.
\end{IEEEbiography}

% % if you will not have a photo at all:
\begin{IEEEbiography}[{\includegraphics[width=1in,height=1.25in,clip,keepaspectratio]{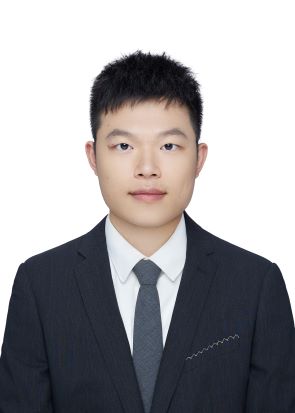}}]{Huanxin Liao}
Huanxin Liao obtained the B.Eng. degree in Electric Power Engineering and Automation from the Shanghai Jiao Tong University, Shanghai, China, in 2021. He is currently pursuing the Ph.D. degree in Computer and Information Engineering with the Chinese University of Hong Kong (Shenzhen), Shenzhen, China.

His research interests include electricity market, energy storage systems, carbon market and artificial intelligence.
\end{IEEEbiography}

\begin{IEEEbiography}[{\includegraphics[width=1in,height=1.25in,clip,keepaspectratio]{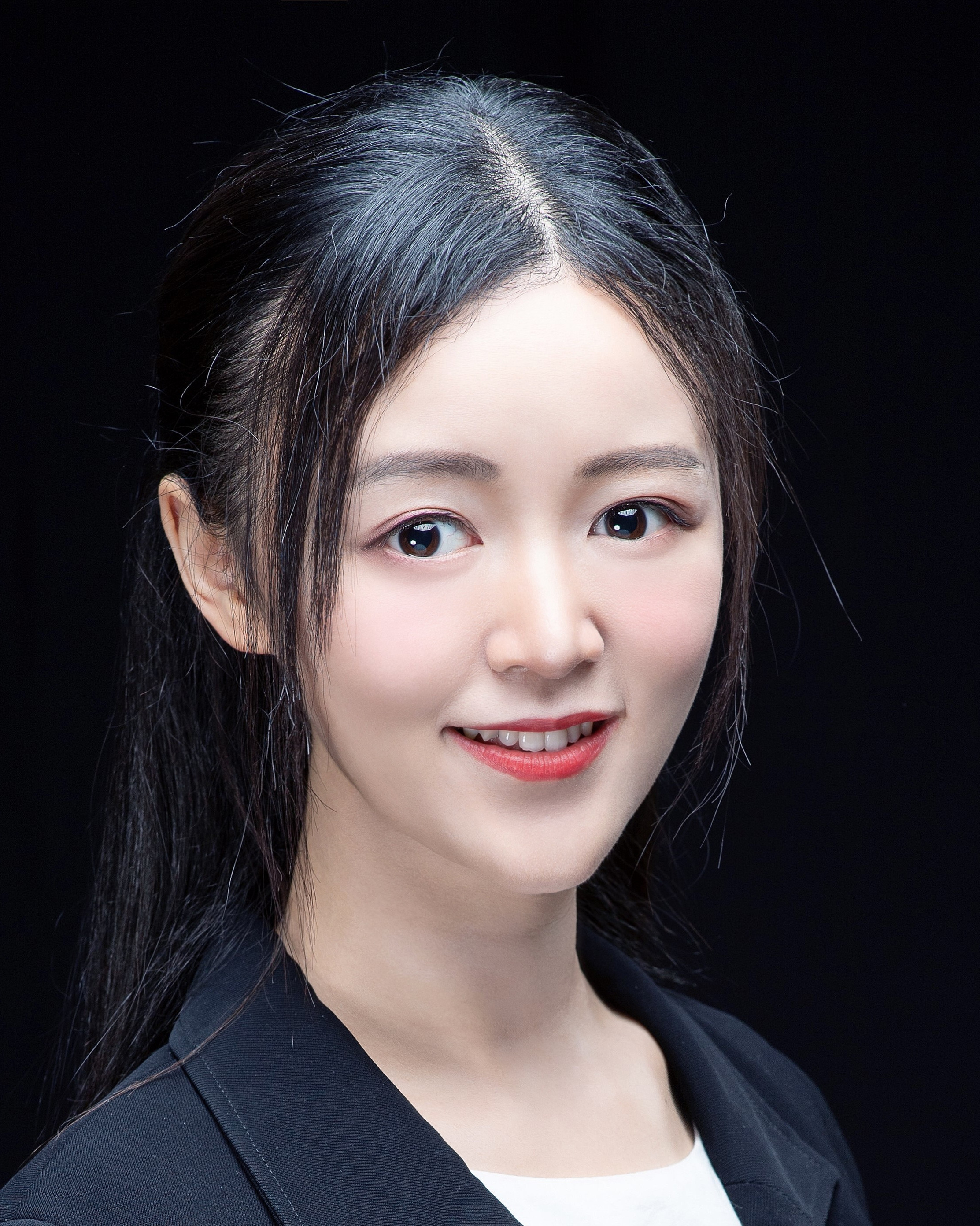}}]{Xiaoying Tang}
Xiaoying Tang (S'13-M'16) is currently a tenure-track assistant professor in The Chinese University of Hong Kong (Shenzhen) and Associate research fellow at the Shenzhen Institute of Artificial Intelligence and Robotics for Society, Shenzhen. Before that, she worked as a Postdoc/Research Scientist at EPFL and The Chinese University of Hong Kong (CUHK) during 2016-2019, respectively. She received her Ph.D. degree from CUHK in Jan. 2016, and B.Eng. degree from UESTC in 2011. Her current research interests include algorithm design and optimizations for EV charging and energy systems, and also theoretical machine learning such as federated learning and domain generalization. Dr. Tang received the Best Paper Award of IEEE SmartGridComm 2013.
\end{IEEEbiography}

\newpage
\appendices

\begin{figure*}[bp] % 因排版移动
\begin{subfigure}{0.25\textwidth}
  \includegraphics[width=\linewidth]{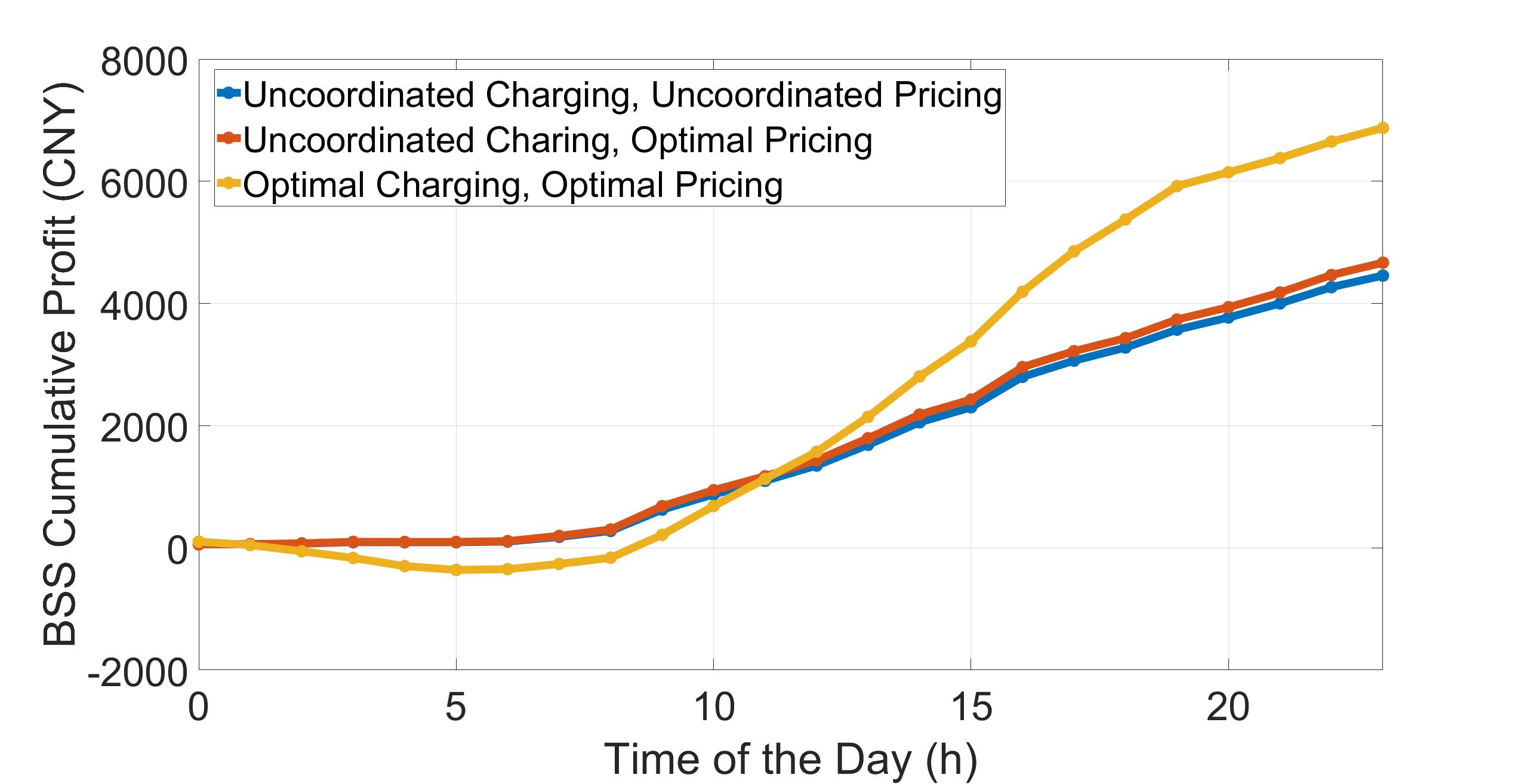}
  \caption{BSS 1 (+54.3\%)}
  \label{fig:BSS1_profit_comparision}
\end{subfigure}%
\hfill % maximize the horizontal separation
\begin{subfigure}{0.25\textwidth}
  \includegraphics[width=\linewidth]{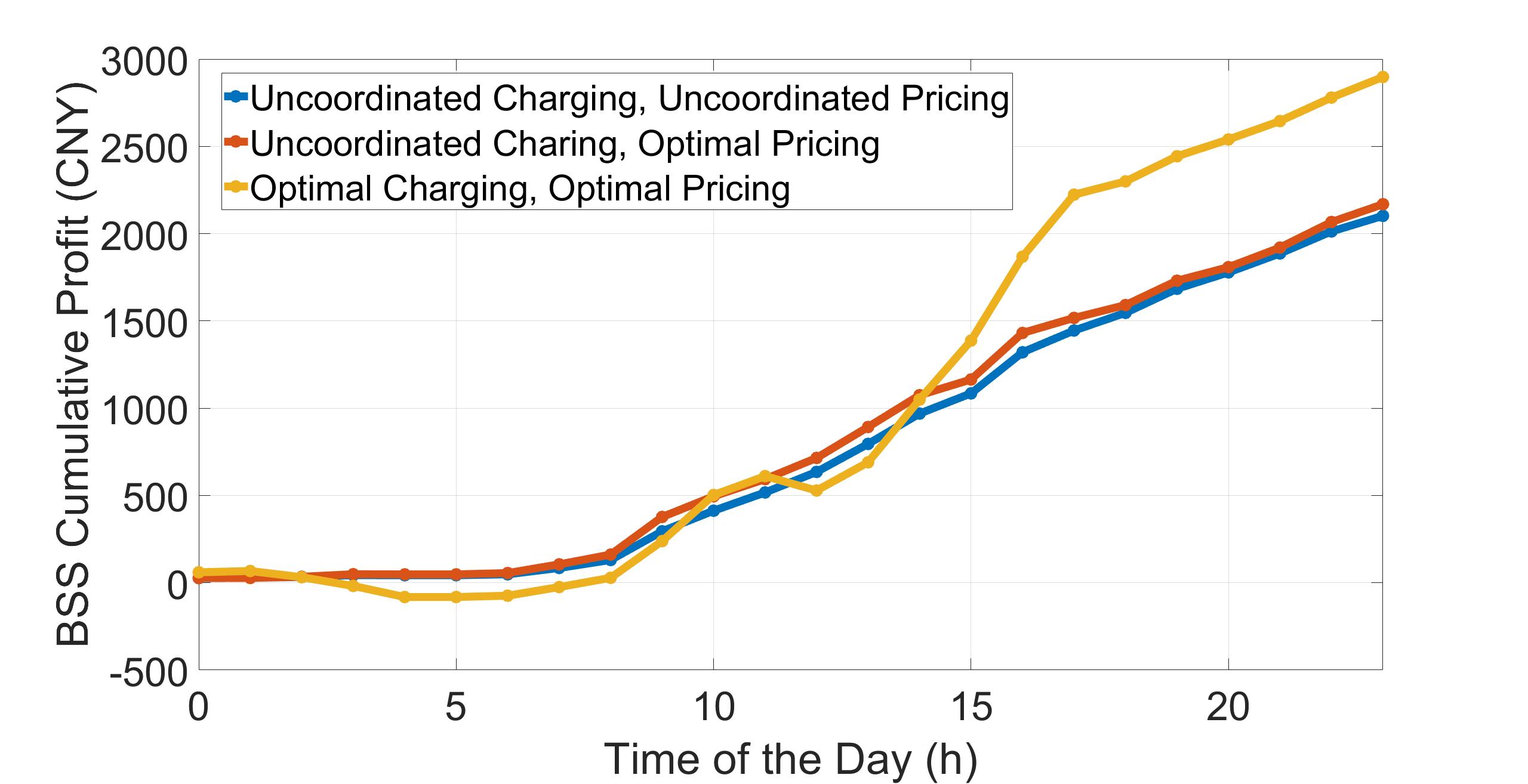}
  \caption{BSS 2 (+37.8\%)}
  \label{fig:BSS2_profit_comparision}
\end{subfigure}%
\hfill % maximize the horizontal separation
\begin{subfigure}{0.25\textwidth}
  \includegraphics[width=\linewidth]{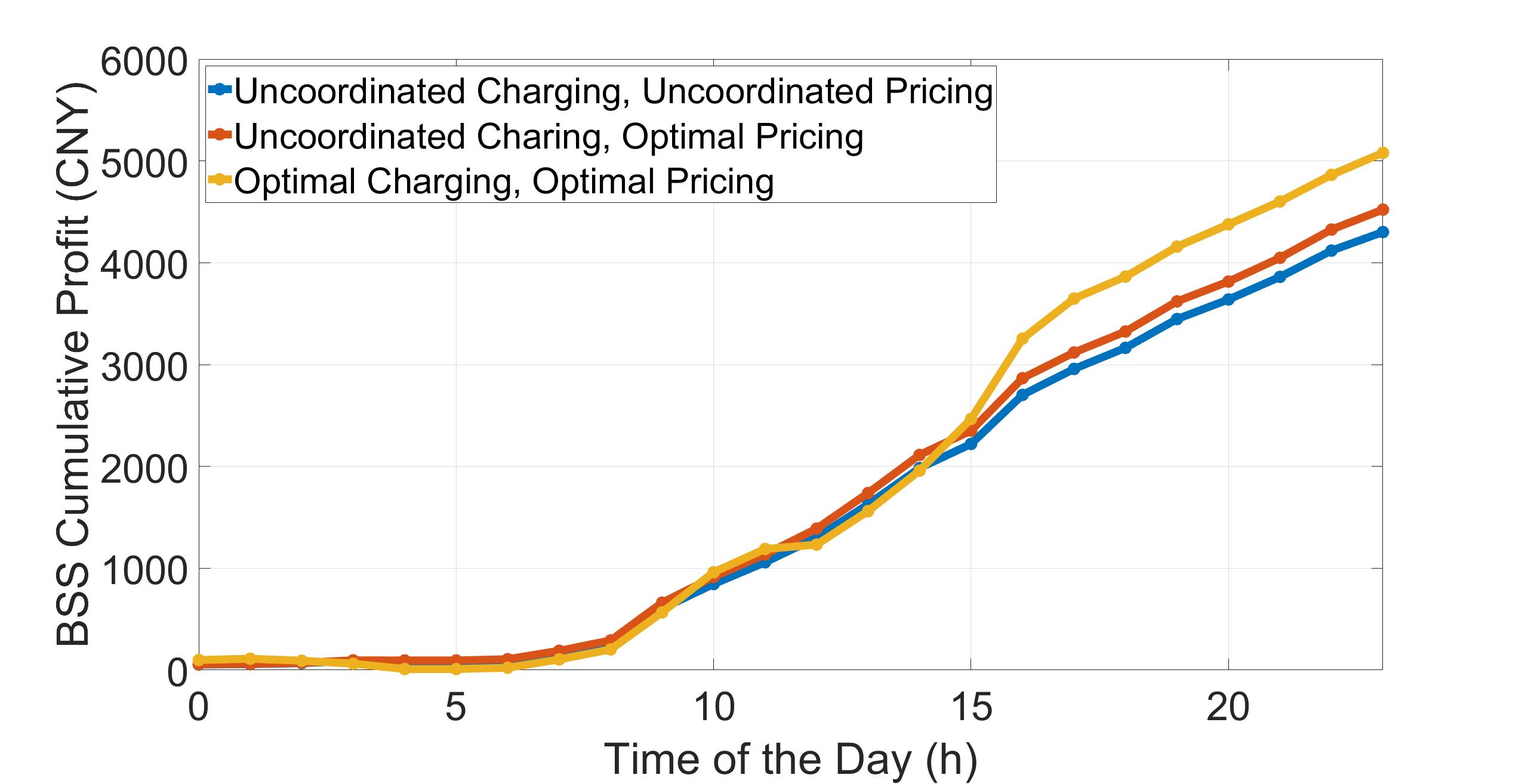}
  \caption{BSS 3 (+18.1\%)}
  \label{fig:BSS3_profit_comparision}
\end{subfigure}%
\hfill
\begin{subfigure}{0.25\textwidth}
  \includegraphics[width=\linewidth]{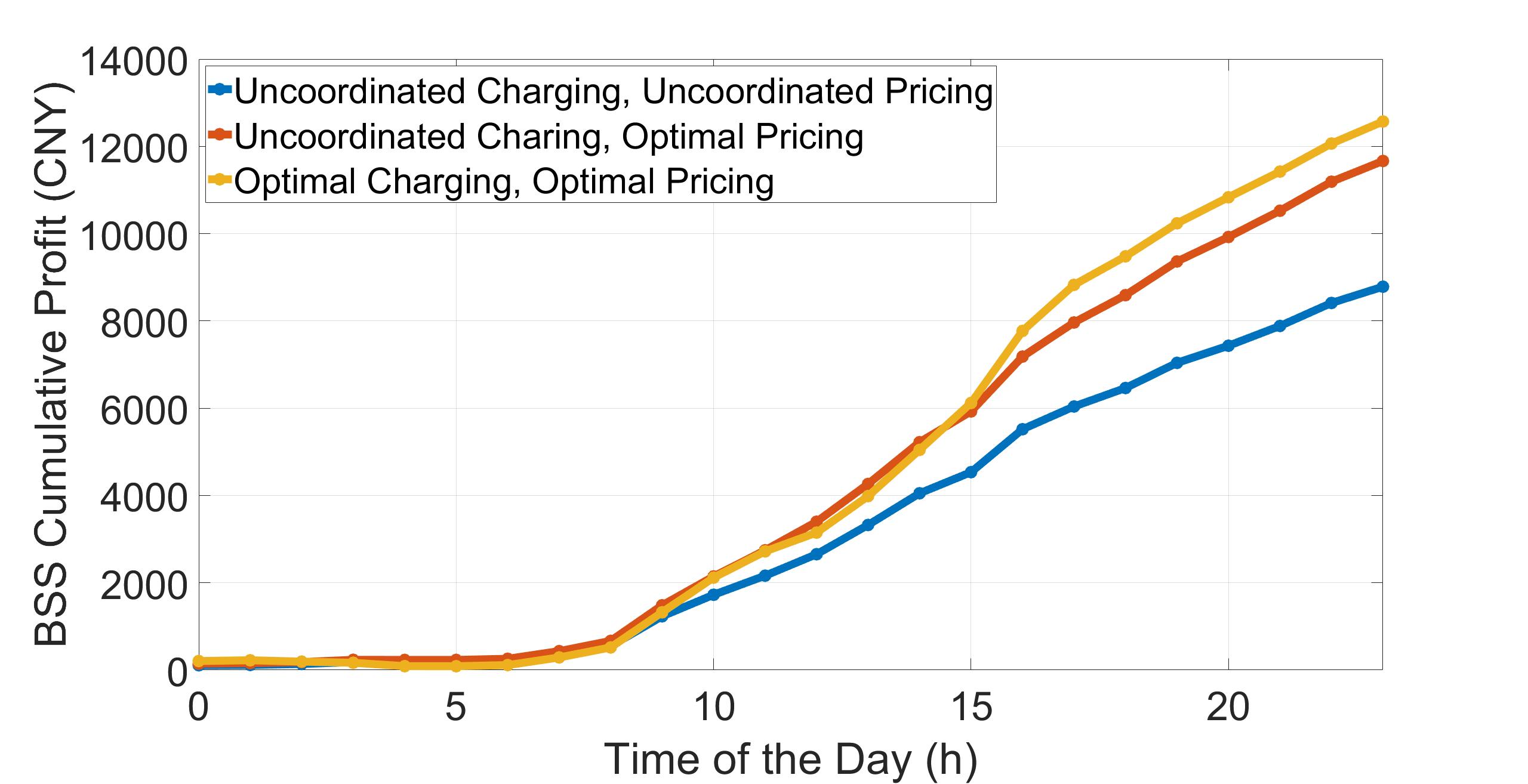}
  \caption{BSS 4 (+43.1\%)}
  \label{fig:BSS4_profit_comparision}
\end{subfigure}%
\hfill
\begin{subfigure}{0.25\textwidth}
  \includegraphics[width=\linewidth]{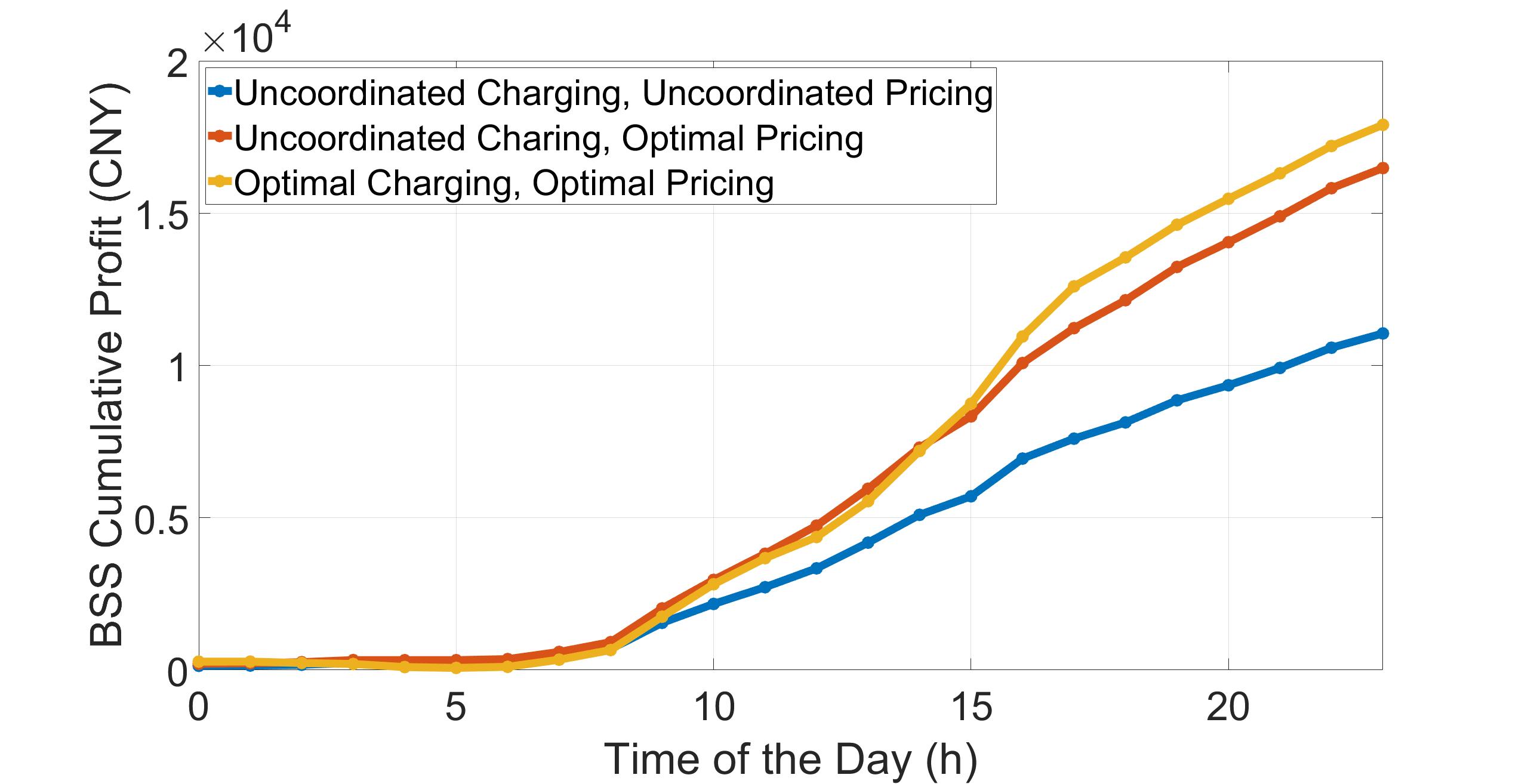}
  \caption{BSS 5 (+62.0\%)}
  \label{fig:BSS5_profit_comparision}
\end{subfigure}%
\hfill % maximize the horizontal separation
\begin{subfigure}{0.25\textwidth}
  \includegraphics[width=\linewidth]{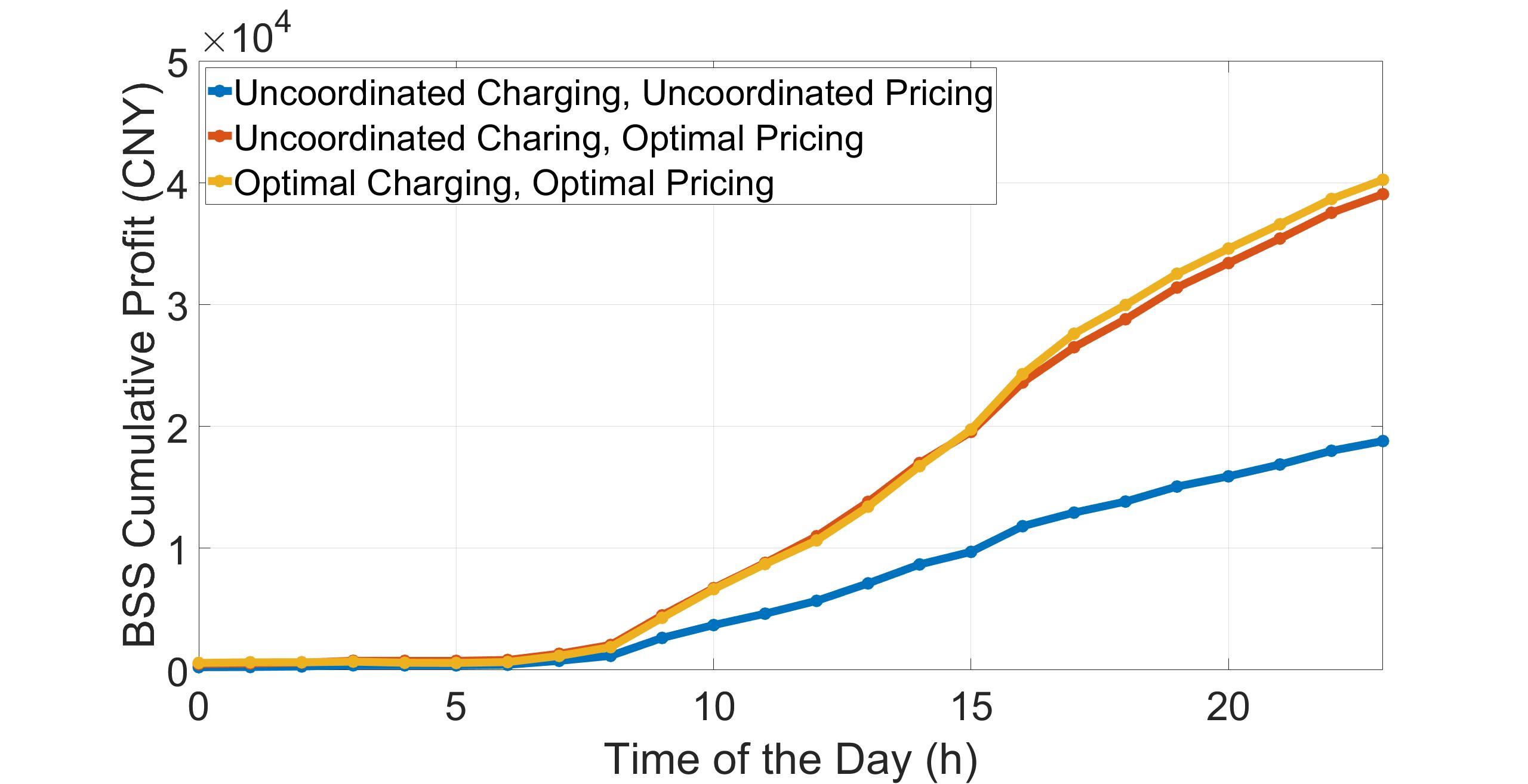}
  \caption{BSS 6 (+114.1\%)}
  \label{fig:BSS6_profit_comparision}
\end{subfigure}%
\hfill
\begin{subfigure}{0.25\textwidth}
  \includegraphics[width=\linewidth]{figures/BSS7_profit_time_slot_cumulative.jpg}
  \caption{BSS 7 (+47.1\%)}
  %\label{fig:BSS7_profit_comparision}
\end{subfigure}%
\hfill
\begin{subfigure}{0.25\textwidth}
  \includegraphics[width=\linewidth]{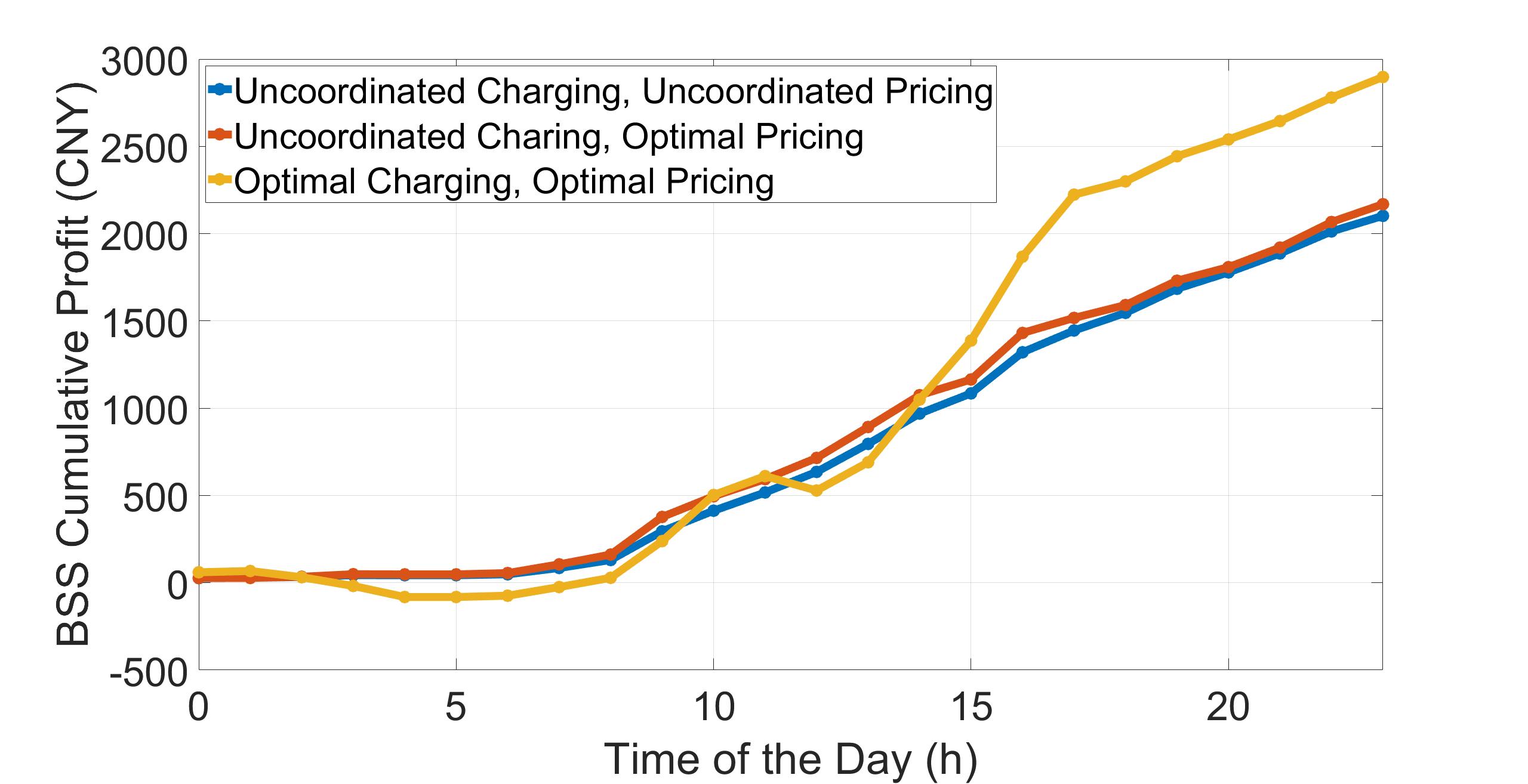}
  \caption{BSS 8 (+37.9\%)}
  \label{fig:BSS8_profit_comparision}
\end{subfigure}%
\hfill
\begin{subfigure}{0.25\textwidth}
  \includegraphics[width=\linewidth]{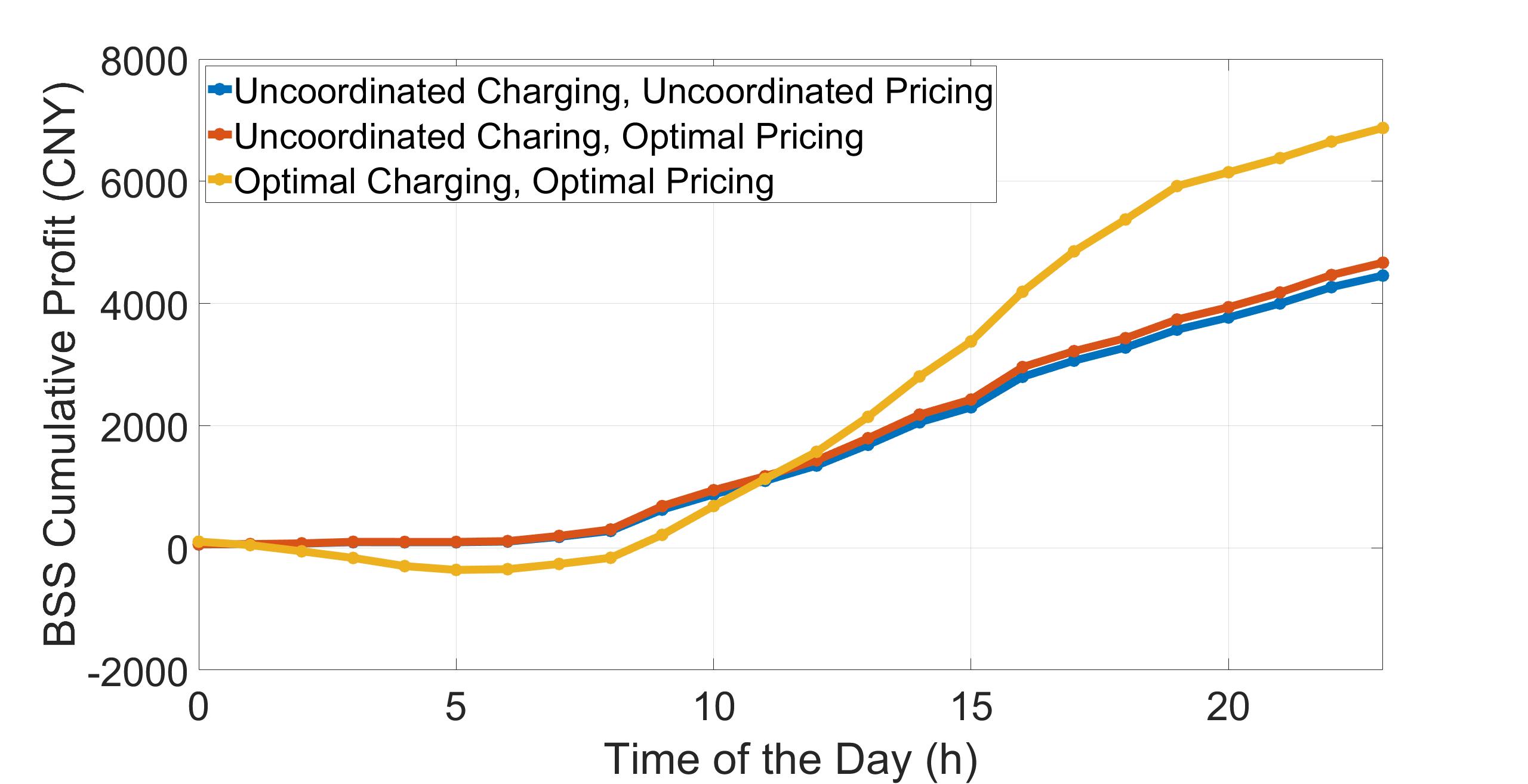}
  \caption{BSS 9 (+54.2\%)}
  \label{fig:BSS9_profit_comparision}
\end{subfigure}%
\hfill % maximize the horizontal separation
\begin{subfigure}{0.25\textwidth}
  \includegraphics[width=\linewidth]{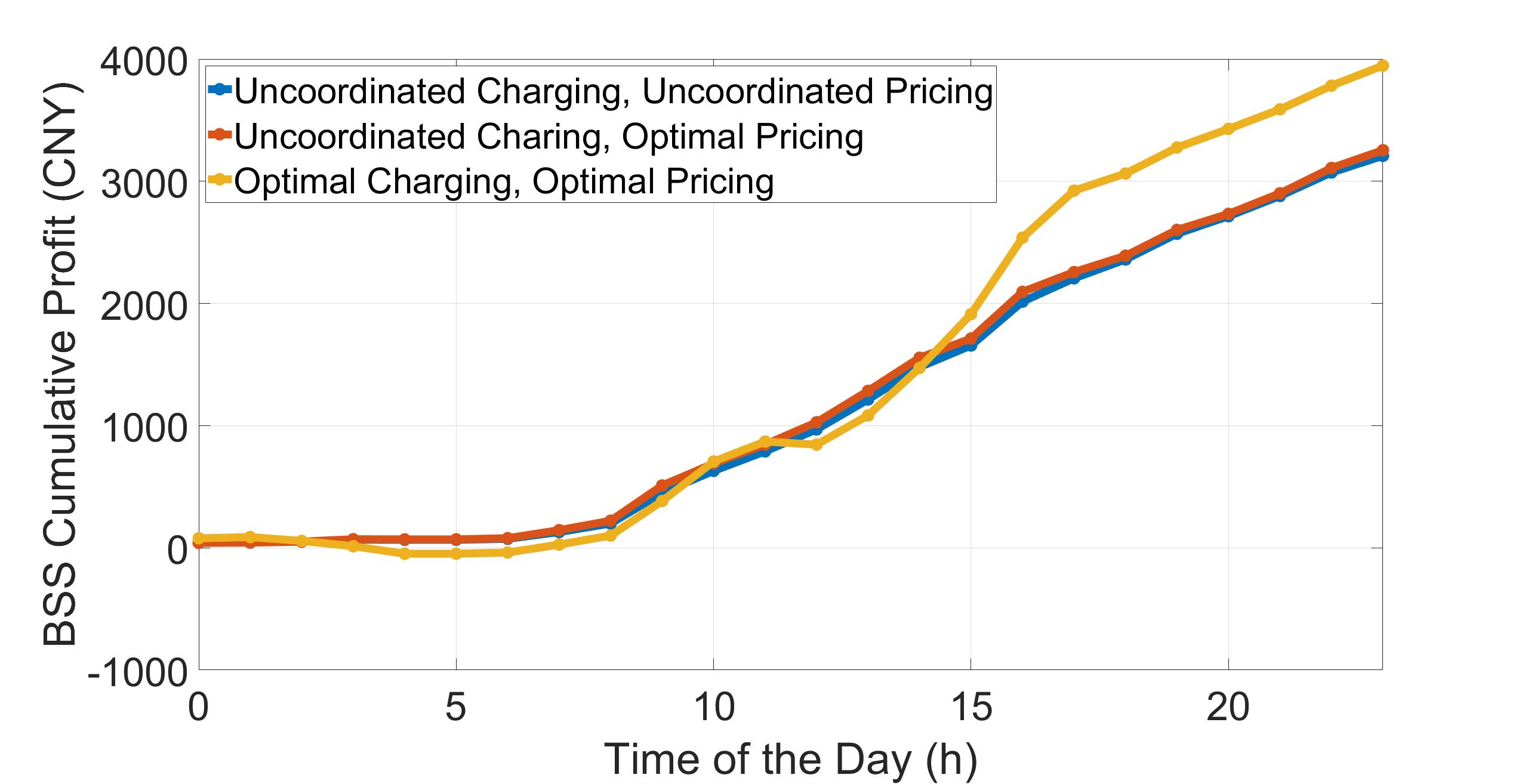}
  \caption{BSS 10 (+23.0\%)}
  \label{fig:BSS10_profit_comparision}
\end{subfigure}%
\hfill
\begin{subfigure}{0.25\textwidth}
  \includegraphics[width=\linewidth]{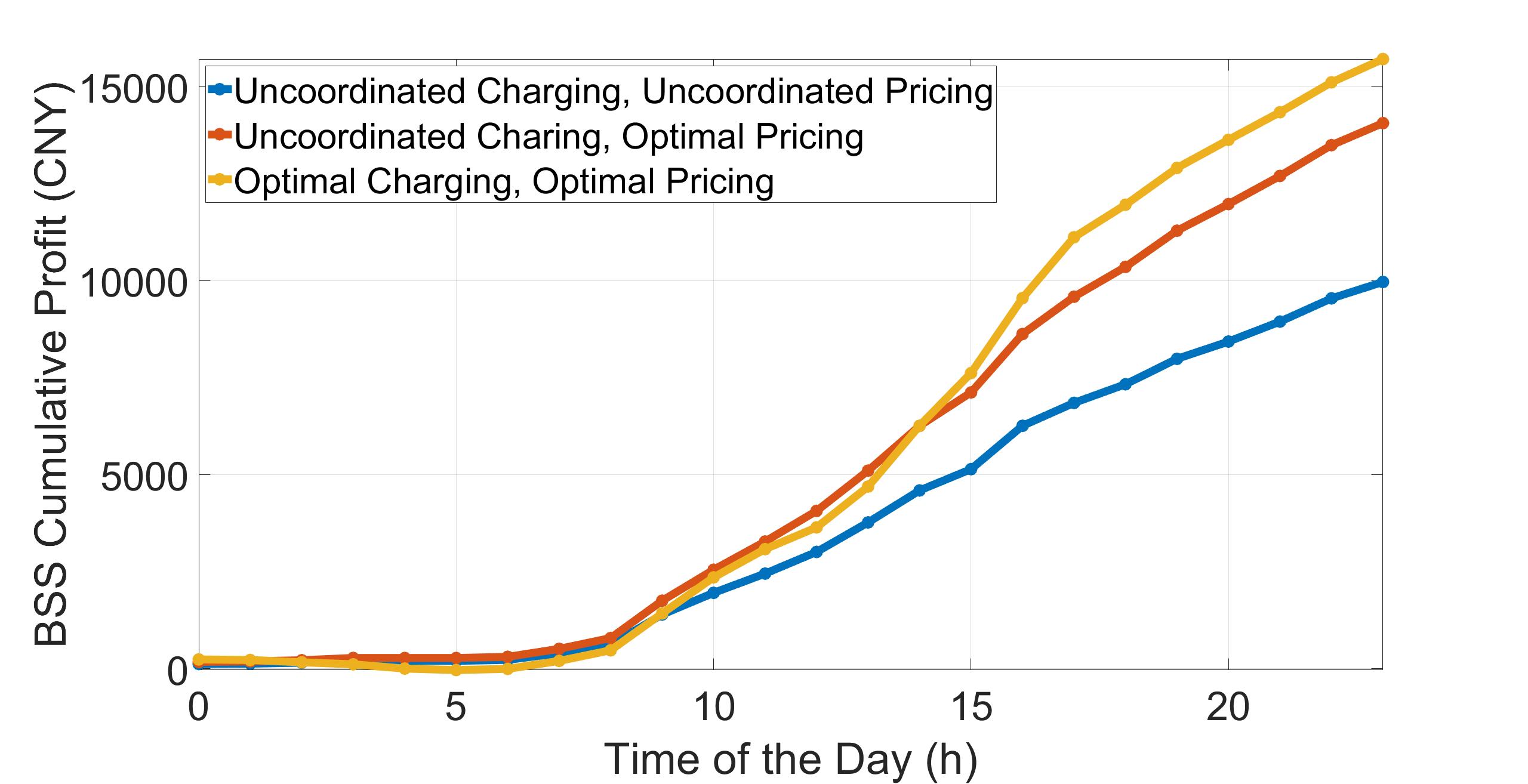}
  \caption{BSS 11 (+57.6\%)}
  \label{fig:BSS11_profit_comparision}
\end{subfigure}%
\hfill
\begin{subfigure}{0.25\textwidth}
  \includegraphics[width=\linewidth]{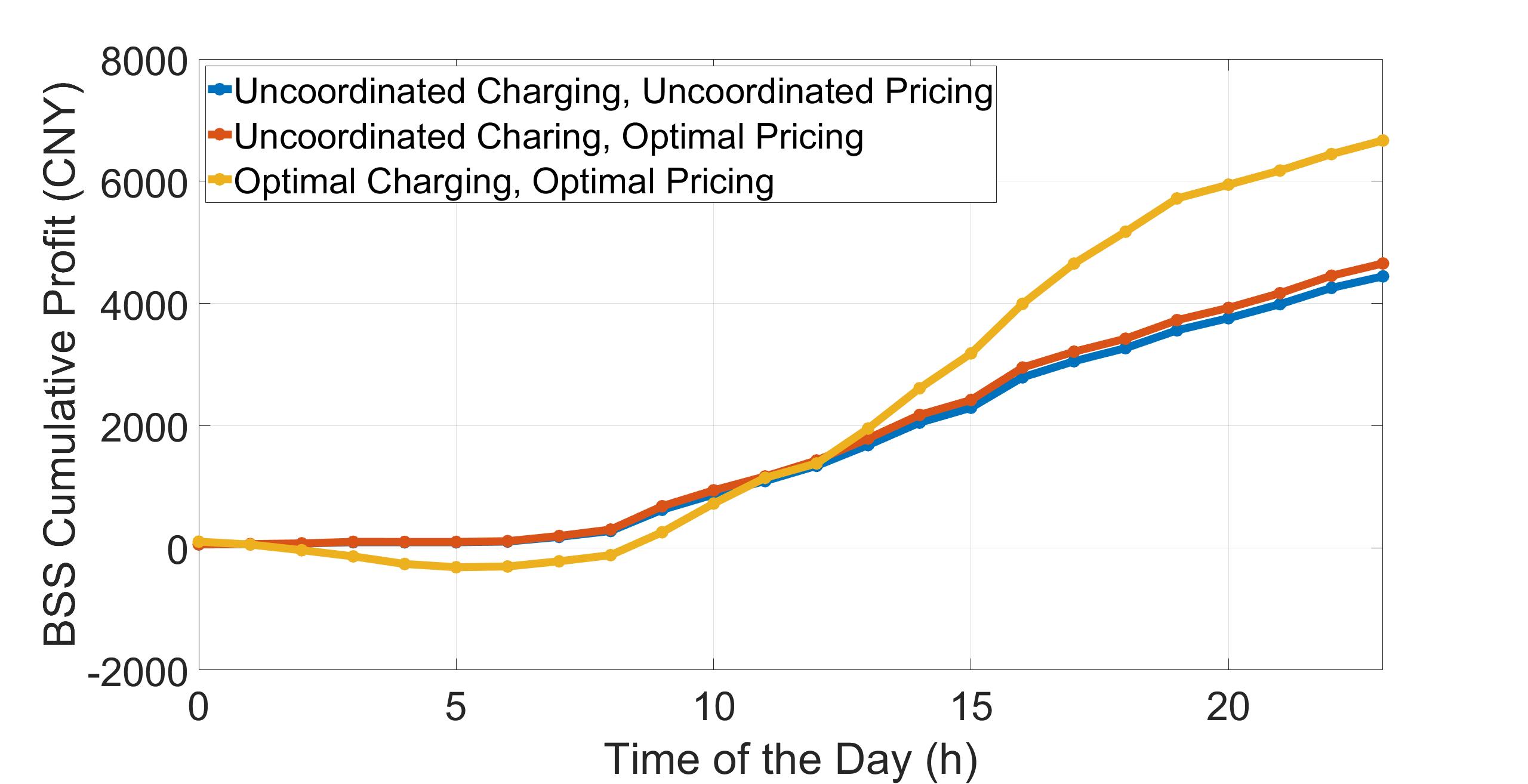}
  \caption{BSS 12 (+50.1\%)}
  \label{fig:BSS12_profit_comparision}
\end{subfigure}%
\caption{BSS profit under different strategies and the improvement. For instance, the profit of BSS 1 increases by 54.3\% from the uncoordinated pricing and charging strategy to the optimal pricing and charging strategy}
\label{fig:BSS_profit_comparision}
\end{figure*}

\section{BSS Cumulative Hourly Profit under Different Strategies}
\label{apdx-BSS_hourly_all}
In Section~\ref{ssec:BSS7_hourly_profit}, we analyze the hourly profit of BSS 7. In this appendix, we give the result of other BSSs in the system as shown in Figure~\ref{fig:BSS_profit_comparision}.

\section{Proof of Lemma~\ref{lmm:BS_Xk}}
\label{apdx:proof_lmm_BS_XK}

\begin{proof}
    As proved in Theorem~\ref{thm:existance}, the payoff function can be written as \eqref{eq:payoff_vector}. Therefore, the maximization the target payoff function in \eqref{eq:payoff_vector} is equivalent to the minimize the function $\mathbf{X_k}^\intercal \mathbf{H} \mathbf{X_k} + \mathbf{f}^\intercal \mathbf{X_k}$.

    Besides, the transmission constraints \eqref{eq:X_transmission_constraints} is equivalent to $\mathbf{lb} \leq \mathbf{X_k} \leq \mathbf{ub}$, and the swapping demand constraints \eqref{eq:X_demand_constraints} is equivalent to $\mathbf{A} \mathbf{X_k} \leq \mathbf{b}$. Therefore, the best response strategy is equivalent to the quadratic programming problem \eqref{eq:X_k_optimization}.
\end{proof}

\end{document}